%% file: main.tex
\documentclass{article}

\usepackage[margin=1in]{geometry}

\usepackage{amsmath,amssymb,amsfonts,amsthm}
\usepackage{mathtools}        
\usepackage{bbm}

\DeclareMathOperator{\diag}{diag}

\usepackage{dsfont}           

\usepackage{graphicx}
\usepackage{subcaption}
\usepackage{wrapfig}

\usepackage{booktabs}         
\usepackage{tabularx}         
\usepackage[table]{xcolor}    

\usepackage{algorithm}
\usepackage{algpseudocode}

\usepackage{enumitem}
\usepackage{nicefrac}
\usepackage{microtype}

\usepackage[numbers]{natbib}

\usepackage{hyperref}
\usepackage{cleveref}

\hypersetup{
  colorlinks=true,
  linkcolor=blue,
  citecolor=teal,
  urlcolor=cyan
}

\usepackage{color-edits}
\addauthor{ev}{red}
\addauthor{sh}{blue}

\usepackage{thmtools}
\usepackage{thm-restate}

\theoremstyle{plain}
\newtheorem{theorem}{Theorem}[section]

\newtheorem{lemma}[theorem]{Lemma}
\newtheorem{corollary}[theorem]{Corollary}

\theoremstyle{definition}
\newtheorem{definition}[theorem]{Definition}
\newtheorem{assumption}[theorem]{Assumption}
\newtheorem{example}[theorem]{Example}
\usepackage{comment}
\usepackage{listings}
\theoremstyle{remark}

\usepackage[textsize=tiny]{todonotes}

\usepackage{times}

\newcommand{\define}{:=}
\newcommand{\real}{\ensuremath{\mathbb{R}}}
\newcommand{\inputspace}{X}
\newcommand{\xdim}{p}

\newcommand{\Min}{\operatorname*{min}}

\newcommand{\argmin}{\operatorname*{\arg\min}}

\newcommand{\partialD}[2]{\frac{\partial}{\partial #2}#1}
\newcommand{\totalD}[2]{\frac{d}{d #2}#1}

\newcommand{\I}{\mathcal{I}}
\newcommand{\ei}{e^{(i)}}

\newcommand{\bigO}{\mathcal{O}}
\newcommand{\Approx}[1]{\hat{#1}}
\newcommand{\ghost}[1]{\tilde{#1}}

\newcommand{\EE}{\mathbb{E}}

\title{First-Order Bilevel Optimization for Non-truthful Auction Design}
\author{Shing-Hei Ho\thanks{\texttt{sho73@gatech.edu}. Georgia Institute of Technology.}
\and April Niu\thanks{\texttt{yniu74@gatech.edu}. Georgia Institute of Technology.}
\and Ellen Vitercik\thanks{\texttt{vitercik@stanford.edu}. Stanford University.}
\and Kai Wang\thanks{\texttt{kwang692@gatech.edu}. Georgia Institute of Technology.}
}
\begin{document}

\maketitle
\begin{abstract}
    Automated mechanism design has enabled the learning of complex auction mechanisms using neural networks. However, designing non-truthful auctions remains challenging because bidders' equilibrium strategies depend on the auction mechanism and thus evolve as the mechanism is optimized. 
    In this work, we formulate auction design as a Stackelberg game with multiple strategic followers and develop a fully first-order algorithm for revenue optimization in non-truthful auctions.
    We establish non-asymptotic guarantees for hypergradient approximation and show that, under standard smoothness and strong-monotonicity assumptions, our method reaches an $O(\epsilon)$-stationary point using $\tilde{\mathcal O}(\epsilon^{-2})$ first-order oracle calls. 
    Experiments on symmetric and asymmetric single-item auctions, as well as multi-item auctions, show that our method converges faster than second- and zeroth-order baselines and scales better as the number of bidders increases. Our method also discovers auctions that achieve revenue close to truthful baselines. These results provide a scalable framework for automated mechanism design beyond truthful auctions.
\end{abstract}

\input{intro}
\input{related_works}
\input{problem_statement}

\input{method}

\input{experiments}

\input{conclusion}

\bibliographystyle{iclr2027_conference}
\bibliography{references}
\input{appendix}

\end{document}

%% file: intro.tex
\section{Introduction}
Automated mechanism design (AMD) uses computation and data to optimize auction mechanisms when analytical solutions are unavailable. Early AMD approaches formulate mechanism design as large-scale mathematical programs (\cite{conitzer2002complexity, conitzer2004self}), while recent work uses neural networks to learn approximately revenue-optimal mechanisms from valuation samples (\cite{dutting2019optimal, feng2018deep, wang2024gemnet, pham2025advancing}). These methods primarily focus on truthful mechanisms, where incentive compatibility makes truthful reporting an equilibrium and therefore removes the need to compute bidders' strategic responses.  

Many widely deployed auctions, however, are non-truthful (\cite{bigler2019update, goke2021bidders, galgana2025learning}). Although the revelation principle often permits mechanism designers to restrict their attention to truthful mechanisms without loss, it does not apply in many settings, such as when bidders have type-dependent constraints on their feasible actions~\citep{berzack2026optimal}. Moreover, practical considerations may limit the mechanism to a prescribed family of non-truthful auctions, such as first-price auctions with reserves. Such auction formats are common in practice; for example, first-price auctions are widespread in online advertising due to their simplicity and transparency~\citep{Harada18:Ad}. However, optimizing non-truthful mechanisms is substantially harder: changing the mechanism changes bidders' equilibrium strategies, so the revenue must be optimized through the equilibrium of an induced multi-player game. 

We formulate non-truthful AMD as a multi-follower Stackelberg problem, where the auctioneer is the leader and bidders are strategic followers.  This formulation connects mechanism design to bilevel optimization. Existing approaches to non-truthful mechanism optimization either use zeroth-order optimization (\cite{feng2021reserve}) or assume access to a differentiable equilibrium oracle (\cite{thoma2026deep, wang2022coordinating}). Meanwhile, current fully first-order bilevel methods primarily consider a single-follower setting (\cite{kwon2023fully,kornowski2024first}). Unlike the single-follower bilevel problem, however, a multi-follower equilibrium couples the responses of all followers: changing one follower’s strategy affects the incentives of the others, and these cross-player effects are generally asymmetric, so standard first-order bilevel constructions do not extend directly. This leaves open how to develop a scalable fully first-order method for multiple strategic followers.  

\paragraph{Main contributions.} Our contributions are threefold. \textbf{First}, we formulate non-truthful auction design as a Stackelberg game with $N > 1$ followers and develop a fully first-order method for symmetric-Jacobian games, achieving $\epsilon$-stationarity with $\bigO(N\epsilon^{-2})$ gradient oracle complexity. We also show why this approach fails in asymmetric games: the true hypergradient requires Jacobian $J^{-\top}$ whereas our method yields $J^{-1}$. \textbf{Second}, we resolve this mismatch using a finite-difference potential function that induces an auxiliary game with a local Jacobian approximating $J^{-T}$. We prove local strong monotonicity and non-asymptotic hypergradient guarantees for the auxiliary game, provide a fully first-order algorithm that converges to $\epsilon$-stationarity using $\tilde{\bigO}(N^2\epsilon^{-2})$ oracle calls, and extend the analysis to linear inequality-constrained follower games. \textbf{Lastly}, we formulate a neural-network-parameterized auction design problem in which the auctioneer optimizes revenue while bidders respond strategically.
We develop a warm-start strategy for tracking differentiable branches of local Nash equilibria, allowing the first-order method to be applied without assuming a globally unique equilibrium. We evaluate the approach on symmetric and asymmetric single-item auctions and multi-item auctions, demonstrating its applicability to equilibrium-dependent mechanism design.

%% file: related_works.tex
\section{Related Works}

\paragraph{Neural Approaches to Learning Truthful Auctions.}
Classical approaches to revenue-maximizing mechanism design face two obstacles in practice. First, they typically assume the designer knows the distribution of buyers' values in advance~ \citep[e.g.,][]{myerson1981optimal}---information that is often unavailable. Second, a general closed-form characterization of revenue-optimal multi-item auctions is unknown~\citep[e.g.,][]{daskalakis2014complexity, rochet1998ironing}. Automated mechanism design (AMD) aims to circumvent these challenges by optimizing for revenue-maximizing mechanisms using sample valuation profiles and computational search~\citep{conitzer2002complexity, conitzer2004self}. The first neural approach to AMD is the work of~\citet{dutting2019optimal}, which represents allocation and payment rules using neural networks and penalizing bidders' gains from misreporting to approximately enforce dominant-strategy incentive compatibility. More recently, GemNet~\citep{wang2024gemnet} uses bidder-independent menus and a post-training transformation to guarantee incentive compatibility. In contrast, we optimize neural pay-your-bid mechanisms whose revenue depends on bidders’ strategic responses, requiring us to account for how equilibrium strategies change with the mechanism.

\paragraph{Non-truthful Auction Design:}
Optimizing non-truthful auctions requires accounting for how mechanism changes affect bidders' equilibrium strategies.
Several papers show that gradient dynamics empirically converge to Nash equilibria in auction games, but focus on computing bidders' strategies rather than optimizing the auction itself~\citep{bichler2021learning, bichler2023learning, kohring2023enabling, huang2025learning}.
More directly related, \citet{feng2021reserve} optimize first-price auction reserves via zeroth-order gradient descent. They compute gradient estimates based on bidders' responsiveness to experimental shocks in reserves, rather than differentiating through a computed equilibrium. Unlike reserve optimization, our framework learns the allocation rule itself, including for multi-item auctions. Meanwhile, \cite{thoma2026deep} formulate mechanism design as a bilevel optimization problem with (coarse) correlated equilibria as the lower-level solution. They assume access to a neural network that outputs a differentiable equilibrium for discrete-action games and its gradient with respect to mechanism parameters. In contrast, we develop novel fully first-order bilevel optimization algorithms with convergence guarantees for Stackelberg games with multiple followers, and apply them to auction design. We design auction mechanisms parameterized by neural networks and compute Nash equilibria through gradient dynamics. Scalability is a key limitation of \cite{thoma2026deep}; their models are only trained on two-player games with at most $16$ discrete actions for each player. In contrast, our action space is continuous.  



\paragraph{Bilevel Optimization:}
Bilevel optimization (BO) is an optimization problem constrained by the solution of another optimization problem. It is essentially a Stackelberg game with a single follower. A body of work studies fully first-order methods for BO. These works typically assume a strongly convex lower-level problem to ensure a unique lower-level solution and an invertible KKT matrix. For an unconstrained lower-level problem, \cite{kwon2023fully} derives an approximate hypergradient using a penalty function and achieves a complexity of $\mathcal{O}(\epsilon^{1.5})$ in deterministic BO. For linearly constrained lower-level problems, \cite{kornowski2024first,zhao2025fully} achieves a complexity of $\mathcal{O}(\epsilon^{-3} \delta^{-1})$ to converge to a $(\epsilon,\delta)$-Goldstein stationary point. 
To the best of our knowledge, existing first-order algorithm for solving Stackelberg games with multiple followers with implicit differentiation is only applicable when Jacobian is block-diagonal (\cite{niu2025finding}). 


%% file: problem_statement.tex
\section{Non-truthful Auction Design as a Stackelberg Game}\label{sec:setup}
Consider a sealed-bid auction of $N$ bidders and $J$ items. Bidder $i$ has a private valuation $t_i\in \mathcal{T}_i\subseteq \real^J_+$, drawn independently from a distribution $F_i$, and submits a bid $a_i\in \mathcal A_i\subseteq \real^J_+$. Let $t = (t_1,..., t_N)$ and $a = (a_1,..,a_N)$. An auction mechanism is parametrized by $x\in X\subseteq \real^p$ and consists of a randomized allocation rule $\rho_x$ and a payment rule $P_x$. 
Given the auction $x$, a bid profile $a$, and valuation $t_i$ for bidder $i$, bidder $i$'s expected utility can be computed by $u_i(x,a,t_i)$.

For a given mechanism $x$, each bidder $i$ chooses a parametrized bidding strategy $\beta_{y_i}: \mathcal{T}_i \rightarrow A_i$. This strategy maps the bidder's valuation $t_i$ to a bid in $A_i$, and $y_i$ is a parameter vector defining the strategy. 
Unless stated otherwise, $y_i \in Y_i(x)\define \real^{n_i}$. Let $y \define (y_1, \cdots, y_N)\in Y(x)\define \Pi_{i=1}^N Y_i(x) \subseteq \mathbb{R}^n$ and write $y_{\neg i} \define (y_j)_{j\neq i}$. The mechanism and bidding strategies induce an \emph{ex-ante expected cost} for each bidder. We work with costs rather than utilities so that the follower game can be expressed as a minimization problem.
\begin{definition}[Ex-ante expected cost]\label{def:ex_ante_cost} Bidder $i$'s ex-ante expected cost is
\begin{align*}
    g_i(x, y_i, y_{\neg i}) \define -\EE_{t_1,\cdots, t_N}[u_i(x,\{\beta_{y_j} (t_{j})\}_{j=1}^N, t_i)].
\end{align*}
Thus, minimizing $g_i$ is equivalent to maximizing bidder $i$'s ex-ante expected utility.
\end{definition}
For a fixed mechanism $x$, the bidders therefore play a normal-form game $\mathcal{G}(x) \define (x, [N], \{Y_i\}_{i\in[N]}, \{g_i\}_{i\in[N]})$. We assume bidders follow a Nash equilibrium with respect to the ex-ante expected cost functions. A Nash equilibrium is a strategy profile in which no player can reduce their expected cost by unilaterally changing their strategy parameter while holding everyone else's strategies fixed.
\begin{definition}[Nash equilibrium]\label{def:nash}
We say $y^* \in Y(x)$ is a Nash equilibrium of the game $\mathcal{G}(x)$ if and only if $g_i(x,y^*) \le g_i(x,y_i, y^*_{\neg i}) \, , \, \forall i \in [N] \, , \, \forall y_i \in Y_i(x)$. We denote the set of Nash equilibria as $\text{NE}(\mathcal{G}(x))$.
\end{definition}
\noindent
The auction design problem is to choose the expected payment rule $P_x$ and allocation rule $\rho_x$ that maximize some objective, assuming bidders play Nash equilibrium strategies. We focus on maximizing the auctioneer's expected revenue. Because our framework is written as a minimization problem, let $f(x,y)$ denote the negative expected revenue under mechanism $x$ and bidding-strategy profile $y$. We first consider the setting in which the induced lower-level game has a unique Nash equilibrium $y^*(x)$ and formulate the auction design problem as a Stackelberg Game.
\begin{align*}\label{eq:equal_p1}
    \min_{x\in \inputspace} \, F(x) \define f(x, y^*(x)) \quad\text{ s.t. } \, y^*(x) \in \text{NE}(\mathcal{G}(x)) \tag{P1}. 
\end{align*}
Section~\ref{sec:nn} extends this formulation to auction games with multiple equilibria.

To solve Problem~\ref{eq:equal_p1}, we run gradient descent on $F(x)$ to optimize the parameter $x$. 
Our goal is to find an $\epsilon$-stationary point, which is defined as any point $x \in \mathbb{R}^\xdim$ such that $\|\nabla F(x)\|\leq \epsilon$.


\noindent
{\textbf{Hypergradient and second-order information:}}
Due to the nested structure in Problem~\ref{eq:equal_p1}, computing the gradient of the function $F(x)$ requires differentiating through the Nash equilibrium $y^*(x)$. 
To show how to compute $\frac{d y^*}{d x}$, we define the gradient operator:
\begin{definition}[gradient operator]\label{def:grad_op}
$V(x,y) \define (\nabla_{y_1} g_1(x,y), \cdots, \nabla_{y_N} g_N(x,y)) \in \real^n$.
\end{definition}
We can write the optimality conditions of the Nash equilibrium $y^*$ using the gradient operator, and apply the Implicit Function Theorem to solve $\frac{d y^*}{d x}$: 
\begin{align}
    V(x,y^*) = 0 ~~\text{(optimality)}  \quad \Longrightarrow \quad \partialD{V(x,y^*)}{x} + \frac{\partial}{\partial y}V(x,y^*) \frac{d}{d x} y^*(x)  = 0 \label{eqn:implicit-function-theorem}
\end{align}

Solving Equation~\ref{eqn:implicit-function-theorem} to get $\frac{d y^*}{dx}$, we get the following exact hypergradient: 
\begin{align}
    \nabla F(x) 
    &= 
    \nabla_xf(x,y^*(x)) + (\totalD{y^*(x)}{x})^\top \nabla_yf(x,y^*(x)) \nonumber \\
    &=
    \nabla_xf(x,y^*(x)) - (\partialD{V(x,y^*)}{x})^\top
    J(x,y^*)^{-\top} 
    \nabla_yf(x,y^*), \label{eq:hypergrad}
\end{align}
where $J(x,y) \define \partialD{V(x,y)}{y}$ is the Jacobian of the gradient operator.

Equation~\ref{eq:hypergrad} gives us the exact hypergradient to solve Problem~\ref{eq:equal_p1}. However, the hypergradient involves the Jacobian $J(x,y^*)$ of the gradient operator, which is second-order and can be computationally intractable. The main goal of this work is to construct an approximate hypergradient estimator that bypasses second-order information and solves Problem~\ref{eq:equal_p1} efficiently.

\paragraph{Definitions and Assumptions:}
Below we summarize the definitions and assumptions needed for the later hypergradient and convergence analysis for solving Problem~\ref{eq:equal_p1}.
\begin{assumption}[Lower-level game]\label{assum:VI}We assume the lower-level game $\mathcal{G}$ is $\mu$-strongly monotone, $C$-smooth and $S$-Hessian smooth. We assume $V$ to be
\begin{enumerate}[leftmargin=*, label=\arabic*., align=left]
        \item $\mu$-strongly monotone: $\langle V(x,y^\prime) - V(x, y), y^\prime - y\rangle \geq \mu \|y^\prime - y\|^2, 
        \, \forall y^\prime, y \in \mathbb{R}^n$,
        \item Jointly $C$-Lipschitz: $\| V(x^\prime,y^\prime) - V(x, y)\| \leq C \|(x^\prime, y^\prime) - (x, y)\|, 
        \, \forall (x^\prime, y^\prime), (x,y) \in \inputspace \times \mathbb{R}^n $,
        \item Jointly $S$-Hessian Lipschitz: $\| \partialD{V(x^\prime,y^\prime)}{(x,y)} - \partialD{V(x, y)}{(x,y)}\| \leq S \|(x^\prime, y^\prime) - (x, y)\|, 
        \, \forall (x^\prime, y^\prime), (x,y) \in \inputspace \times \mathbb{R}^n$.
    \end{enumerate}
\end{assumption}
The strong monotonicity assumption guarantees a unique lower-level Nash equilibrium ~\citep{facchinei2003finite}.


\begin{assumption}[Upper-level objective]\label{assum:optim}We assume $f$ to be
\begin{enumerate}[leftmargin=*, label=\arabic*., align=left]
        \item Jointly $L_f$-Lipschitz
        $\| f(x^\prime,y^\prime) - f(x, y)\| \leq L_f \|(x^\prime, y^\prime) - (x, y)\|, 
        \, \forall (x^\prime, y^\prime), (x,y) \in \inputspace \times \mathbb{R}^n$, 
        \item $C_f$-gradient Lipschitz  $\| \nabla_{(x,y)}f(x^\prime,y^\prime) - \nabla_{(x,y)}f(x, y)\| \leq C_f \|(x^\prime, y^\prime) - (x, y)\|, 
        \, \forall (x^\prime, y^\prime), (x,y) \in \inputspace \times \mathbb{R}^n$, 
        \item $S_f$-Hessian Lipschitz $\| \partialD{\nabla_{(x,y)}f(x^\prime,y^\prime)}{(x,y)} - \partialD{\nabla_{(x,y)}f(x, y)}{(x,y)}\| \leq S_f \|(x^\prime, y^\prime) - (x, y)\|, 
        \, \forall (x^\prime, y^\prime), (x,y) \in \inputspace \times \mathbb{R}^n$.
    \end{enumerate}
\end{assumption}


%% file: method.tex
\section{Stackelberg Games with Symmetric Jacobian}\label{sec:sym}
\noindent
We first consider Stackelberg games where the Jacobian $J = \partialD{V(x,y^*)}{y}$ is symmetric, i.e., $J = J^\top$. This condition holds in the special case where followers do not interact through their decision variables, or when the pairwise strategic interactions between every follower $i$ and $j$ are reciprocal.
This section proposes a first-order method for solving Stackelberg games with a symmetric Jacobian.


Consider problem \ref{eq:equal_p1}. Define a perturbed game $\mathcal{G}_\delta \define (x, [N], \{Y_i\}_{i\in[N]}, \{g_i + \delta f\}_{i\in [N]})$, where the followers' objectives are perturbed by the leader's objective $\delta f$ with a small number $\delta > 0$. Let $y^*_\delta$ be the new Nash equilibrium for the perturbed game $\mathcal{G}_\delta$. 
We construct the following finite difference:
\begin{align}\label{eq:finite_diff}
    v_\delta(x,y^*_{\delta}) = \frac{1}{\delta} \sum_{i=1}^N \nabla_x g_i(x,y^*_{\delta,i}, y^*_{\neg i}) - \nabla_x g_i(x,y^*_{i}, y^*_{\neg i})
\end{align}
We first quantify the error of the finite-difference estimator in the idealized setting where the perturbed equilibrium $y_\delta^*$ is computed exactly.
\begin{restatable}{theorem}{gradApprox}\label{theorem:grad_approx}
    If the Jacobian $J = \partialD{V(x,y^*)}{y}$ is symmetric, then:
    \begin{align}
        \|v_\delta(x,y^*_\delta) + (\partialD{V(x, y^*)}{x})^\top J(x,y^*)^{-\top} \nabla_y f(x,y^*)\| \leq \bigO(N\delta).
    \end{align}
\end{restatable}
\begin{proof}[Proof sketch]
The perturbed equilibrium $y_\delta^*$ satisfies $V(x,y_\delta^*)+\delta\nabla_y f(x,y_\delta^*)=0$. Differentiating with respect to $\delta$ at 0 gives 
$\lim_{\delta\to 0}
\frac{y_\delta^*-y^*}{\delta} = \totalD{y^*_\delta}{\delta}|_{\delta=0} = - J(x,y^*)^{-1}\nabla_yf(x,y^*(x)).$
By the fundamental Theorem of calculus we have $v_\delta(x,y_\delta^*) = \left(\frac{\partial V(x,y^*)}{\partial x}\right)^\top\frac{y_\delta^*-y^*}{\delta}+\bigO(N\delta)$ where the remainder term comes from summing up the error $\|y_\delta^*-y^*\|=\bigO(\delta)$ for each player.
Thus, 
\begin{align}
    &v_\delta(x,y_\delta^*) \approx  -(\partialD{V(x,y^*)}{x})^\top
    J(x,y^*)^{-1} 
    \nabla_yf(x,y^*) + \bigO(N\delta) \\& = -(\partialD{V(x, y^*)}{x})^\top J(x,y^*)^{-\top} \nabla_y f(x,y^*)+\bigO(N\delta) \label{eq:approx_grad}
\end{align}
The last equality in Equation~\eqref{eq:approx_grad} is due to the symmetry of $J$. 
\end{proof}
\vspace{-10pt}

This approximation guarantee suggests that we can run inexact gradient descent to obtain a fully first-order approximation of the hypergradient with a non-asymptotic rate. This intuition is formalized in Theorem~\ref{theorem:sym_equality_convergence}. The proof of Theorem~\ref{theorem:grad_approx} highlights that the naive finite difference construction fails when $J(x,y^*)^{\top} \neq J(x,y^*)$.

\section{Stackelberg Games with Asymmetric Jacobian}\label{sec:asym}
As alluded to above, the current perturbed game cannot be used to solve the Stackelberg game without a symmetric Jacobian. We propose a novel approach to resolve the Jacobian mismatch in this section. 
\subsection{Finite difference approximation for unconstrained Stackelberg game} \label{sec:asym_unconstr}
Consider differentiating $\langle V(x,y), y- y^* \rangle$ with respect to $y$. The product rule gives $\nabla_y \langle V(x,y), y- y^* \rangle - V(x,y) = J (x,y)^\top (y - y^*)$.
This is exactly the transpose operator missing in (\ref{eq:approx_grad}). However, differentiating the scalar $\langle V(x,y), y- y^* \rangle$ through $y$ requires second-order derivatives of the follower objectives. Thus, we construct a finite-difference \textit{potential} whose gradient approximates $\nabla_y \langle V(x,y), y- y^* \rangle$.
\begin{definition}[Potential function]\label{def:potential}
Given a small number $\epsilon > 0$, define \newline \centerline{$\phi_\epsilon(x,y) 
\define\frac{1}{\epsilon} \sum_{j=1}^N [g_j(x, y_j + \epsilon (y_j - y^*_j), y_{\neg j}) - g_j(x,y)].$}
\end{definition}

By definition of directional derivative, each finite-difference term in $\phi_\epsilon(x,y)$ converges to $\langle \nabla_{y_i} g_i(x,y), y_i- y_i^*\rangle$. So $\phi_\epsilon(x,y)$ approaches $\langle V(x,y),y-y^*\rangle$ and $\nabla_y \phi_\epsilon(x,y)$ approaches to $\nabla_y \langle V(x,y),y-y^*\rangle$ as $\epsilon\to0$. Therefore, $\nabla_y\phi_\epsilon(x,y)-V(x,y)$ approximates $J(x,y)^\top(y-y^*)$.
Consider an unconstrained perturbed game $\bar{\mathcal{G}}_{\delta} \define (x, [N], \{Y_i\}_{i\in[N]},  \{\bar{g}_i\}_{i\in[N]})$ with cost function
$\bar{g}_i(x,y) = \phi_\epsilon(x,y) - g_i(x,y) + \delta f(x,y)$.
Denote the gradient operator for the perturbed game $\bar{\mathcal{G}}_\delta$ as $\bar{V}(x,y) = (\nabla_{y_1} \bar{g}_1(x,y) , \cdots , \nabla_{y_N} \bar{g}_N(x,y)) \in \mathbb{R}^n$. Then we have $\bar V = \nabla_y \phi_\epsilon - V + \delta \nabla_y f$ where the term $ \nabla_y \phi_\epsilon - V$ is exactly what we need to recover the Jacobian. 
Now consider the perturbed equilibrium condition \[0 = \bar V(x, \bar y^*_\delta)\approx J(x,y)^\top(\bar y^*_\delta-y^*) + \delta \nabla_y f(x,y^*)\implies \frac{\bar y^*_\delta-y^*}{\delta} \approx -J(x,y)^{-\top} \nabla_y f(x,y^*)\] Thus, the perturbed equilibrium $\bar{y}^*_\delta$ satisfies $\bar{y}^*_0 \approx y^*$ and thus $\totalD{\bar{y}^*_\delta}{\delta}|_{\delta=0} \approx - J(x,y^*)^{-\top} \nabla_y f(x,y^*)$, as required by the Stackelberg hypergradient. As a consequence, we use the finite difference $v_\delta(x,\bar{y}^*_\delta)$ induced by the perturbed equilibrium to recover the correct Jacobian:

Even though the construction resolves the Jacobian mismatch, it does not yield a fully first-order algorithm immediately. This is because the auxiliary game operator $\bar V$ does not necessarily inherit strong monotonicity from $V$. Thus the next technical challenge is to prove that the auxiliary game is strongly monotone locally around the neighborhood of the true equilibrium $y^*$ (Lemma \ref{lemma:asym_local_strong_monotone}), and that the perturbed equilibrium $\bar{y}^*_\delta$ stays inside this neighborhood (Lemma \ref{lemma:asym_y_delta}).
The following theorem puts the ideas together and shows that $v_\delta(x,\bar{y}^*_\delta)$ is a fully first-order approximation with a non-asymptotic rate.

\begin{theorem}[Non-Asymptotic hypergradient approximation]\label{theorem:asym_grad_approx}
For $0 < \delta < \frac{\mu}{2 C_f}$ and $0 < R < \frac{\mu - \delta C_f}{S(N+1)(2+\sqrt{N})}$ and $0 < \epsilon < 1$, $\bar{V}$ is locally $\bar{\mu}$-strongly monotone and $\bar{C}$-smooth on $\bar{B}_R(y^*) = \{y \in \mathbb{R}^n : \|y-y^*\| \leq R\}$ for some $\bar{\mu}, \, \bar{C}>0$. Furthermore, if $\delta < \min(\frac{\mu}{2 C_f}, \frac{\bar{\mu} R}{2 L_f})$, then we have a fully first-order hypergradient approximation with a non-asymptotic rate:
\newline\centerline{$\|v_\delta(x,\bar{y}^*_\delta) - \partialD{V(x,y^*)}{x}^\top {J(x,y^*)}^{-\top} (-\nabla_y f(x,y^*))\|
    \leq 
    \mathcal{O}(\frac{N \delta}{\bar{\mu}^2} + \epsilon).$}
\end{theorem}
\begin{proof}[Proof Sketch]
Lemma \ref{lemma:asym_local_strong_monotone} shows that the auxiliary operator $\bar{V}$ is locally $\bar{\mu}$-strongly monotone and $\bar{C}$-smooth on $\bar{B}_R(y^*)$. Under the stated bound on $\delta$, Lemma \ref{lemma:asym_y_delta} shows that the map $T(y) = y-\eta \bar{V}(x,y)$ is a contraction on $\bar{B}_R(y^*)$. Hence the auxiliary game admits a unique Nash equilibrium $\bar{y}^*_\delta$ in this neighborhood with $\|\bar{y}^*_\delta - y^*\| \leq \frac{2 L_f\delta}{\bar{\mu}}$. Now we relate the displacement of this auxiliary equilibrium to the exact hypergradient. 
Let $v^*_\epsilon \define (J(x,y^*)^{\top} + \epsilon D(x,y^*))^{-1} (-\nabla_y f(x,y^*))$ and $v^* \define J(x,y^*)^{-\top} (-\nabla_y f(x,y^*))$. Lemma \ref{lemma:asym_grad_1} shows that the normalized equilibrium displacement approximates $v^*_\epsilon$: $\|\frac{\bar{y}^*_\delta - y^*}{\delta} - v^*_\epsilon\| \leq\mathcal{O}( \frac{\sqrt{N}\delta}{\bar{\mu}^2})$. Lemma \ref{lemma:asym_grad_2} controls the bias introduced by the $\epsilon D$ regularization: $ \|v^*_\epsilon - v^*\| \leq \mathcal{O}(\epsilon)$. It remains to relate the equilibrium displacement to the finite-difference expression. By Taylor expanding $v_\delta(x,\bar{y}^*_\delta)$ around $y^*$, we have $v_\delta(x,\bar{y}^*_\delta) 
=
\partialD{V(x,y^*)}{x}^\top \frac{\bar{y}^*_{\delta} - y^*}{\delta} + r_\delta$  where $\|r_\delta\|\leq \mathcal{O}(\frac{N\delta}{\bar{\mu}^2})$. The sum of these three terms gives an upper bound on $ \|v_\delta(x,\bar{y}^*_\delta) - \partialD{V(x,y^*)}{x}^\top v^*\|$ which yields the result. 
\end{proof}

Algorithm~\ref{alg:ffo_algo} summarizes our fully first-order method for both symmetric and asymmetric Stackelberg games. At each outer iteration, we approximately solve the original follower game and the perturbed game, construct the finite-difference hypergradient estimator, and update the leader variable by inexact gradient descent. We now analyze the algorithm's convergence.

\begin{theorem}\label{theorem:asym_convergence}
Consider Problem \ref{eq:equal_p1}. Under Assumptions \ref{assum:VI} and \ref{assum:optim}, Algorithm \ref{alg:ffo_algo} outputs an $\mathcal{O}(\epsilon)$-stationary point of the hyperobjective $F$ with $\tilde{\bigO}(N^2\epsilon^{-2})$ oracle calls to $f$ and $g_i$. 
\end{theorem}
\begin{proof}
By the unconstrained version of Theorem~\ref{theorem:smooth_F}, the hyperobjective $F$ is smooth.
Choose $\delta \leq \frac{\epsilon}{2N}$ and use the potential function
$\phi_{\epsilon/2}$. By Lemma~\ref{lemma:inexact_lower_level}, if the approximate lower-level solutions
$\hat y,\hat y_\delta$ satisfy
$\|y^*-\hat y\|\leq \delta^2$ and
$\|\bar y_\delta^*-\hat y_\delta\|\leq \delta^2$, then the inexact
hypergradient
${\hat\nabla F}(x)=\nabla_x f(x,\hat y)+\hat v_\delta(x,\hat y_\delta)$
satisfies
$\|{\hat\nabla F}(x)-\nabla F(x)\|
\leq \bigO(N\delta+\epsilon/2)
\leq \bigO(\epsilon)$.
Therefore, by Lemma~\ref{lemma:inexact_GD}, inexact gradient descent requires
$\bigO(\epsilon^{-2})$ hypergradient evaluations to obtain an
$\bigO(\epsilon)$-stationary point of $F$. It remains to bound the cost of each hypergradient evaluation.
Since the potential function is constructed using the approximate equilibrium
$\hat y$, its gradient operator is also inexact, with an error of $\bigO(\epsilon)$ (Lemma~\ref{lemma:inexact_potential}). Moreover, the auxiliary game is only locally strongly
monotone and smooth on $\bar B_R(y^*)$. We therefore choose
$\delta \leq
\min\{\epsilon/(2N),\,\bar\mu R/(4L_f),\,\sqrt{R/2}\}$,
which ensures
$\|\bar y_\delta^*-\hat y\|
\leq
\|\bar y_\delta^*-y^*\|+\|y^*-\hat y\|
\leq R$.
Hence, gradient descent initialized at $\hat y$ remains in the local
neighborhood and converges linearly to the auxiliary equilibrium. For each call to ${\hat\nabla F}(x)$,
$\bigO(\log(1/\delta^2))$ lower-level iterations are sufficient to obtain
the required $\delta^2$-accurate solutions. Each evaluation of the auxiliary
game operator requires gradients of all $N$ surrogate costs, and evaluating
each surrogate cost requires $\bigO(N)$ calls to the original follower
objectives $g_i$. Thus, each lower-level iteration uses $\bigO(N^2)$
first-order oracle calls. Combining this with $\bigO(\epsilon^{-2})$ outer
iterations and $\bigO(\log(1/\delta^2))$ lower-level iterations gives $\bigO\!\left(
N^2\epsilon^{-2}\log(1/\delta^2)
\right)
=
\widetilde{\bigO}(N^2\epsilon^{-2}).$
\end{proof}
Note that, compared to the symmetric case, there is an extra factor of $N$ due to the sum over all $N$ players' objectives in the potential function.

\subsection{Constrained Stackelberg game}
We now extend the framework to follower games with inequality constraints. For each player $i$, let $Y_i= \{y_i\in \mathbb R^{n_i}: e^{(i)}(x,y_i): = A^{(i)}x + B^{(i)}y_i + b^{(i)}\le 0\}$. Linear inequality constraints are more challenging because the hyperobjective may be piecewise smooth. However, only the active inequalities affect the local equilibrium sensitivity, so the problem can be reduced locally to a linear equality-constrained Stackelberg game. Appendix~\ref{sec:ineq} formalizes this reduction. Thus, it suffices to consider the case where $Y_i$ is equality-constrained. 

\begin{definition}[Functional constraints]\label{def:constraints}
Each $\ei: \inputspace \times Y_i \rightarrow \real^{m_i}$ is a function denoting equality constraints for each player. Each $e^{(i)}$ is private to player $i$. Denote $m \define \sum_{i=1}^N m_i$. Denote $\lambda \define (\lambda_1, \cdots, \lambda_N) \in \real^m$ as the Lagrange multipliers for constraints of the players. Denote $\lambda^* \define (\lambda^*_1, \cdots, \lambda^*_N) \in \mathbb{R}^m$ as the Lagrange multipliers satisfying the KKT condition of $\Min_{y_i \in Y_i(x)} g_i(x,y_i,y^*_{\neg i}), \forall i \in [N]$.
\end{definition}

In this setting, we use the same potential $\phi_\epsilon$ as Definition~\ref{def:potential} and define a constrained auxiliary game $\bar{\mathcal{G}}_{\delta}$ by $\bar{g}_i(x,y) = \phi_\epsilon(x,y) - g_i(x,y) - (\lambda^*)^\top By + \delta f(x,y)$. It follows that $\bar{\mathcal{G}}_\delta$ has the same local strong monotonicity modulus $\bar{\mu}$ and local smoothness modulus $\bar{C}$ on the ball $\bar{B}_R(y^*)$ as the unconstrained case. 
Let $\bar{y}^*_\delta$ be the Nash equilibrium of $\bar{\mathcal{G}}_\delta$ and $\bar{\lambda}^*_\delta$ be the corresponding Lagrange multiplier. We modify the finite difference to address the constraints as:
$v_\delta(x,\bar{y}^*_\delta, \bar{\lambda}^*_\delta) 
    = 
    \frac{1}{\delta}\sum_{i=1}^N [\nabla_x g_i(x,\bar{y}^*_{\delta,i},y^*_{\neg i}) - \nabla_xg_i(x,y^*) 
    + \partialD{e^{(i)}(x,y^*_i)}{x}^\top (\bar{\lambda}^*_\delta - \lambda^*)]$

The constrained analysis follows the same argument as in the unconstrained case after replacing the game Jacobian by the corresponding KKT system. The transpose-aware construction in Section \ref{sec:asym_unconstr} therefore yields the same hypergradient approximation error order. The resulting auxiliary game remains locally strongly monotone, so the lower-level equilibrium can be computed to the required accuracy by first-order methods, and the approximate hypergradient can be used in inexact gradient descent. The full approximation and complexity guarantees are deferred to Appendices~\ref{sec:asym_constr} and~\ref{sec:ineq}.

\begin{algorithm}
    \caption{First-Order Algorithm for Solving Stackelberg Game with Multiple Followers}
    \label{alg:ffo_algo}
    \textbf{notation:}
    Let $S_{\mathcal{M}}(x, y): \inputspace \times \mathbb{R}^n \rightarrow Y(x)$ be a solver that returns a Nash equilibrium of an arbitrary game $\mathcal{M}(x)$ given initial iterate $y$. Define $\mathcal{Q}_\delta \define \mathcal{G}_\delta$ if $\partialD{V(x,y^*(x))}{y}$ is symmetric, otherwise define $\mathcal{Q}_\delta \define \bar{\mathcal{G}}_\delta$. 
    \\
    \textbf{initialize:} 
    Initial upper-level variable $x_0 \in \inputspace$, initial lower-level variable $y_0 \in \mathbb{R}^n$, finite difference parameter $0<\delta<1$, potential function parameter $0<\epsilon<1$, learning rate $\eta > 0$.
    \\
    \begin{algorithmic}[1]
        \For{$t = 0,1,2,\dots,T-1$}
            \State
            Compute Nash equilibrium $\Approx{y} = S_{\mathcal{G}}(x_t, y_t)$ (and the dual solution $\Approx{\lambda}$ if constrained) 
            \State
            Compute perturbed Nash equilibrium $\Approx{y}_\delta = S_{\mathcal{Q}_\delta}(x_t, \Approx{y})$ (and the dual solution $\Approx{\lambda}_\delta$ if constrained). 
            \State
            Compute approximate hypergradient $\Approx{\nabla}F(x_t) = \nabla_x f(x_t, \Approx{y}) + \Approx{v}_{\delta}(x, \Approx{y}_\delta)$ (use $\Approx{v}_{\delta}(x, \Approx{y}_\delta, \Approx{\lambda}_\delta)$ if constrained)
            \State
            $y_{t+1} = \Approx{y}$
            \State
            $x_{t+1} = x_t - \eta \Approx{\nabla}F(x_t)$  
        \EndFor \\
    return $x_t$, $y_t$
    \end{algorithmic}
\end{algorithm}
\vspace{-10pt}
\section{Learning Non-truthful Auctions} \label{sec:nn}
In this section, we instantiate the general first-order Stackelberg framework for learning non-truthful auction mechanisms, where the mechanism is parameterized by neural networks, and bidders respond strategically. Figure \ref{fig:framework} and Algorithm \ref{alg:ffo_algo} illustrate our framework.

\begin{figure}[h!]
    \centering
    \includegraphics[width=0.9\linewidth]{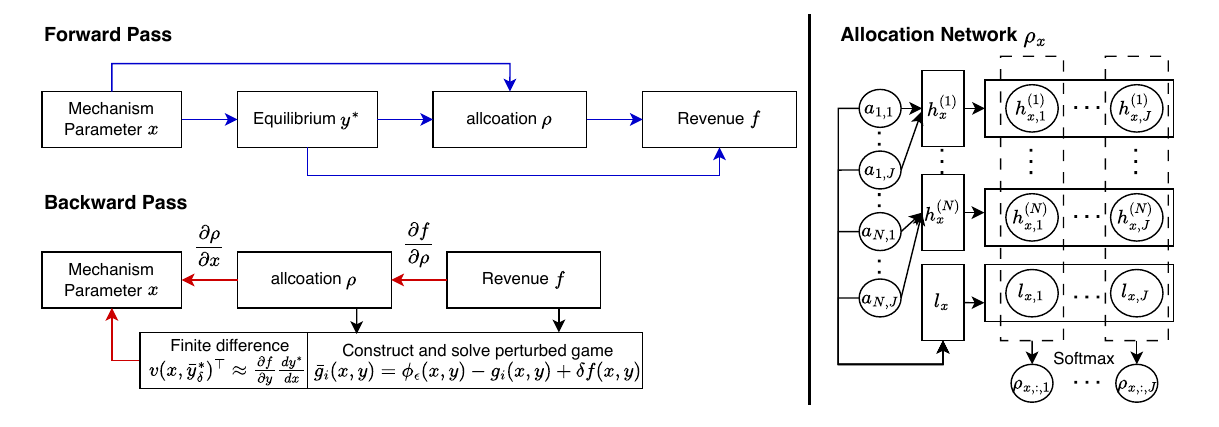}
    \vspace{-10pt}
    \caption{In forward pass, we compute the Nash equilibrium $y^*$ of the original game and compute the revenue $f(x,y^*)$. $a_{i,j}$ is the bidder $i$'s bid on item $j$. In the backward pass, we solve the perturbed game defined by $\bar{g}_i(x,y) = \phi_\epsilon(x,y)-g_i(x,y) + \delta f(x,y)$ and compute $v_\delta(x,\bar{y}^*_\delta) \approx \totalD{y^*(x)}{x}^\top \nabla_y f(x,y^*)$.}
    \vspace{-10pt}
    \label{fig:framework}
\end{figure}

\subsection{Parameterization of Auction Mechanism}
An $N$-bidder, $J$-item mechanism is given by an expected payment rule $P_x:\prod\mathcal{A}_i \rightarrow \real^N$ and a randomized allocation function $\rho_x: \prod\mathcal{A}_i \rightarrow \prod_{j=1}^J \Delta^{N+1}$. Without loss of generality, we assume $0 \leq t_i \leq \omega_i, \forall i \in [N]$, for some $\omega_i \in \mathbb{R}^J_{+}$. Let $\rho_{x,i,j}$ denote the probability that player $i$ receives item $j$. Define the additive expected utility function as $u_i(x,y,t_i) = \sum_{j=1}^J t_{i,j} \rho_{x,i,j}(\{\beta_{y_k}(t_k)\}_{k\in[N]}) - P_{x,i}(\{\beta_{y_k}(t_k)\}_{k\in[N]})$, with the bidding strategies $\{\beta_{y_k}\}_{k=1}^N$ being neural networks. 
We model the mechanism with temperature $\tau>0$, neural networks $l_x: \prod \mathcal{A}_i \rightarrow \mathbb{R}^{J}$, and $h^{(i)}_x: \mathcal{A}_i \rightarrow \mathbb{R}^{J}$. 
\begin{equation*}
       \rho_{x,i, j}(a) = \frac{\exp(h^{(i)}_{x,j}(a_i)/\tau)}{\sum_{k=1}^N \exp(h^{(k)}_{x,j}(a_k)/\tau) + \exp(l_{x,j}(a)/\tau)},\quad P_{x,i}(a) = \sum_{j=1}^J \rho_{x,i,j}(a) a_{i,j}
\end{equation*}
Similar to a first-price auction, bidders pay their own bids if they receive an item. The function $l_x$ controls the probability that no items are allocated.

\subsection{Tracking a Local Equilibrium Branch}
Algorithm \ref{alg:ffo_algo} is used to find optimal non-truthful auctions. However, Nash equilibria may not be unique or even exist in auction games \citep{bichler2023convergence}, in which case Problem \ref{eq:equal_p1} is ill-defined. 
When there exists multiple Nash equilibria, we select a differentiable branch of Nash equilibria $z^*: \inputspace \rightarrow Y(x)$ such that $z^*(x) \in \text{NE}(\mathcal{G}(x)), \forall x \in \inputspace$, and formulate the following Stackelberg game: $\Min_{x\in \inputspace} \, F(x) \define f(x, z^*(x)) \define -\mathbb{E}_{t}\left[\sum_{i=1}^N P_{x,i}(\{\beta_{z^*_j(x)}(t_j)\}_{j=1}^N)\right]$.
When Nash equilibrium does not exist, we may replace $\text{NE}(\mathcal{G})$ with $\text{LNE}(\mathcal{G})$. 
The problem is how we can select a differentiable branch of local Nash equilibria. We may assume local strong monotonicity around every local Nash equilibrium $y^*(x)$. 
Local Nash equilibria vary smoothly if the initial iterate is accurate and the mechanism parameters vary slightly (the hyperobjective's learning rate $\eta$ is small). To ensure the initial game can be solved accurately, we pretrain the mechanism neural network to the allocation rule of a simultaneous first-price auction. We also initialize the lower-level solution as a parameter $y$ that induces a truthful bidding strategy ($\beta_{y_i}(t_i) = t_i \, , \forall t_i \in \mathcal{T}_i$). The following theorem formalizes the warm-start strategy for auction design (see the proof in Appendix \ref{appx:exp}). 

\begin{restatable}{theorem}{warmStart}\label{thm: warm_start}
Suppose $\mathcal{G}(x)$ is locally strongly monotone on a ball $\bar{B}_{r}(y^*(x))$ with radius $r>0$ for all $y^*(x) \in \text{LNE}(\mathcal{G}(x))$ and $x \in \inputspace$. Given initial iterate $y$, let $S(x,y, \Approx{\delta}): \inputspace \times \mathbb{R}^n \times \mathbb{R}_{+} \rightarrow Y(x)$ be a solver can always return $\Approx{y}$ such that $\|\Approx{y} - y^*(x)\| \leq \Approx{\delta}$ whenever $\|y - y^*(x)\| \leq r$ and $y^*(x) \in \text{LNE}(\mathcal{G}(x))$.

In Algorithm \ref{alg:ffo_algo}, suppose $\|y_0 - y^*(x_0)\| \leq \frac{r}{2}$ for some $y^*(x_0) \in \text{LNE}(\mathcal{G}(x_0))$ and the learning rate $\eta$ is sufficiently small, then $y_k = z^*(x_k), \, \forall k \in [T]$ for some continuously differentiable function $z^*: X \rightarrow Y(x)$ such that $z^*(x) \in \text{LNE}(\mathcal{G}(x))$.
\end{restatable}

%% file: experiments.tex
\section{Experiments}

\paragraph{Experimental settings.} We evaluate our method in three two-bidder auction settings. The first is a symmetric single-item auction, with bidders' values drawn independently and identically (iid) from $\text{Uniform}[0, 1]$. In this setting, a second-price auction with a reserve of $0.5$ is revenue-optimal. A first-price auction with the same reserve achieves the same expected revenue at its symmetric Bayes-Nash equilibrium~\citep{myerson1981optimal, krishna2009auction}. The second setting is an asymmetric single-item auction, with values drawn independently from $\text{Uniform}[0, \frac{4}{3}]$ (strong bidder) and $\text{Uniform}[0, \frac{4}{5}]$ (weak bidder). Under the optimal auction~\citep{myerson1981optimal}, the item is allocated to the bidder with the highest non-negative virtual value, where the strong and weak bidders' virtual values are $2t_1 - \frac{4}{3}$ and $2t_2 - \frac{4}{5}$, with corresponding reserve values of $0.667$ and $0.4$, respectively. The third setting has two items and additive bidders, with per-item valuations sampled iid from $\text{Uniform}[0,1]$. We use GemNet~\citep{wang2024gemnet} as a learned truthful benchmark; no closed-form revenue-optimal mechanism is known. 

At the start of each training run, we sample $10,000$ random valuation profiles and reuse this fixed sample to estimate expected revenue throughout training. We replace softmax with max when reporting the expected revenue.

\begin{figure}[h!]
    \centering
    \vspace{-10pt}
    \includegraphics[width=\linewidth]{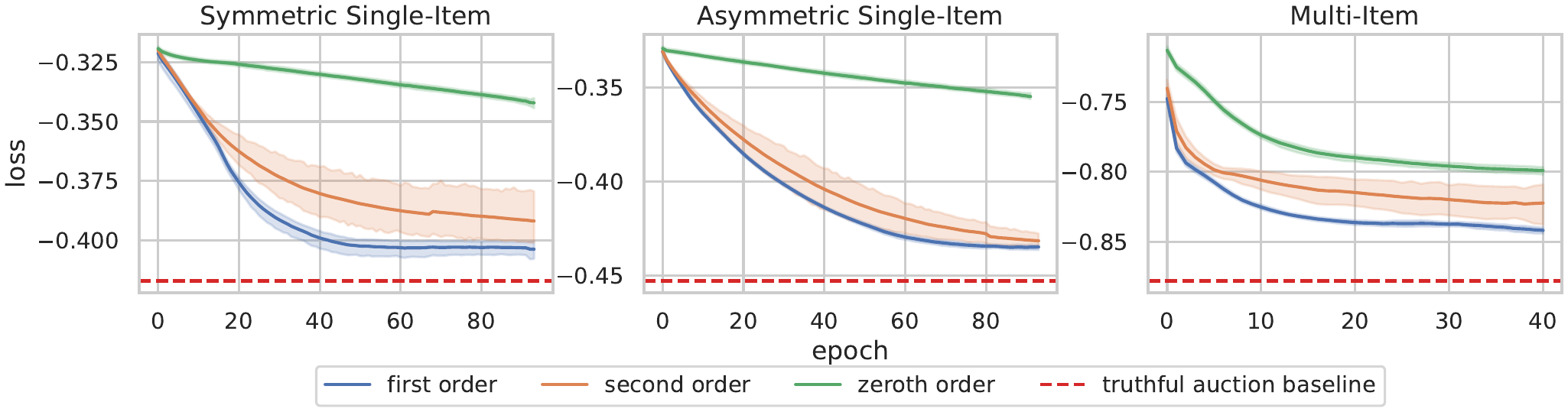}
    \caption{Convergence of negative expected revenue. Symmetric single-item auction (left), asymmetric single-item auction (middle), multi-item auction (right). The truthful multi-item auction baseline is GemNet.}
    \label{fig:convergence}
    \vspace{-10pt}
\end{figure}


\begin{minipage}[c]{0.55\textwidth}
    \paragraph{Optimization performance.} Figure \ref{fig:convergence} compares our first-order method with a second-order method (exact hypergradient) and a zeroth-order method~\citep{kornowski2024algorithm}. Across all three settings, our method converges faster to a higher revenue. Moreover, the second-order method may be more sensitive to bias and noise based on its high variance.
    Figure \ref{fig:time} evaluates computational scaling by increasing the number of bidders in the symmetric single-item setting. Our first-order method has a significantly lower backward-pass runtime compared to the second-order implementation. This is consistent with the growing computational cost of the linear system required by the second-order method. 
\end{minipage}\hfill
\begin{minipage}[c]{0.4\textwidth}
    \centering
    \includegraphics[width=\linewidth]{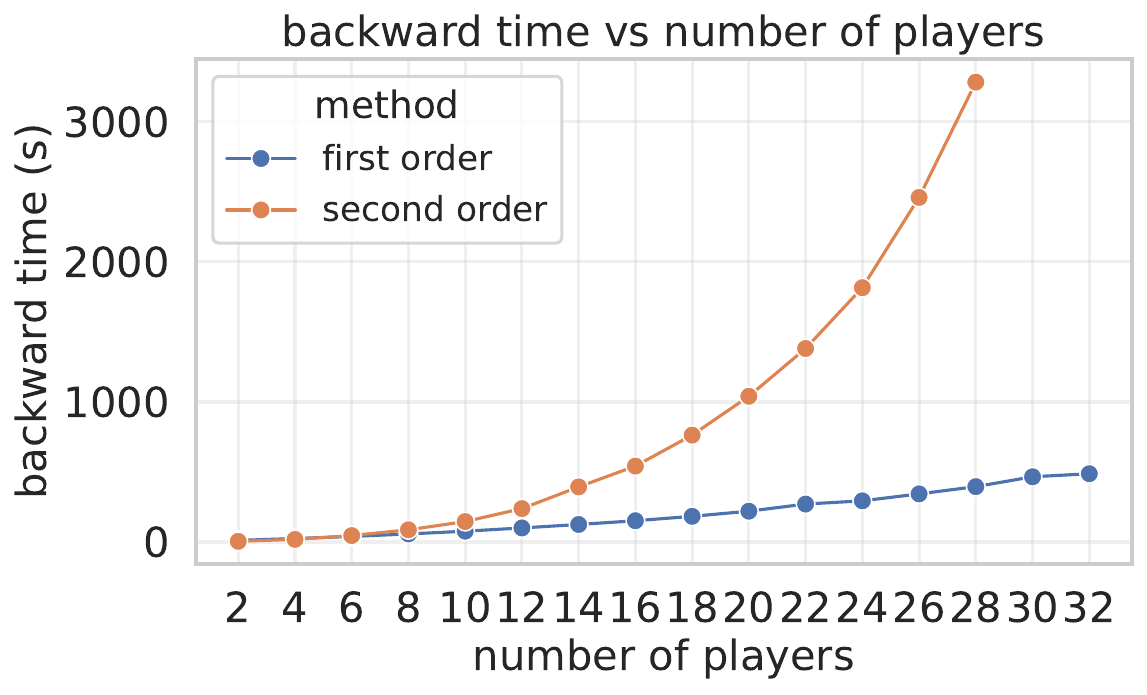}
    \captionof{figure}{Backward time vs number of players in symmetric single-item auction}
    \label{fig:time}
\end{minipage}
\vspace{-10pt}


\paragraph{Comparison with truthful benchmarks.} Our goal in these experiments is to assess whether optimizing pay-your-bid mechanisms can achieve revenues close to truthful benchmarks while retaining a simple, transparent payment rule----not to establish a revenue advantage of non-truthfulness itself. Indeed, Table \ref{tab:revenue} demonstrates that the revenues of our non-truthful auction approximate those of the optimal auctions in single-item settings and GemNet in the multi-item setting. Appendix~\ref{appx:exp} and Figure \ref{fig:sym_vis} include allocation visualizations that provide qualitative comparisons with these benchmarks. 


\begin{table}[h!]
    \vspace{-8pt}
    \centering
    \caption{Expected revenue comparison}
    \label{tab:revenue}
    \begin{tabular}{lccc}
        \hline
        Methods & Symmetric 2-bidder 1-item & Asymmetric 2-bidder 1-item & 2-bidder 2-item \\
        \hline
        Truthful Baselines & 0.417 (optimal) & 0.453 (optimal) & 0.878 (GemNet)\\
        Ours & 0.404 & 0.435 & 0.842 \\
        \hline
    \end{tabular}
    \vspace{-5pt}
\end{table}

\begin{figure}[h!]
    \centering
    \vspace{-10pt}
    \includegraphics[width=\linewidth]{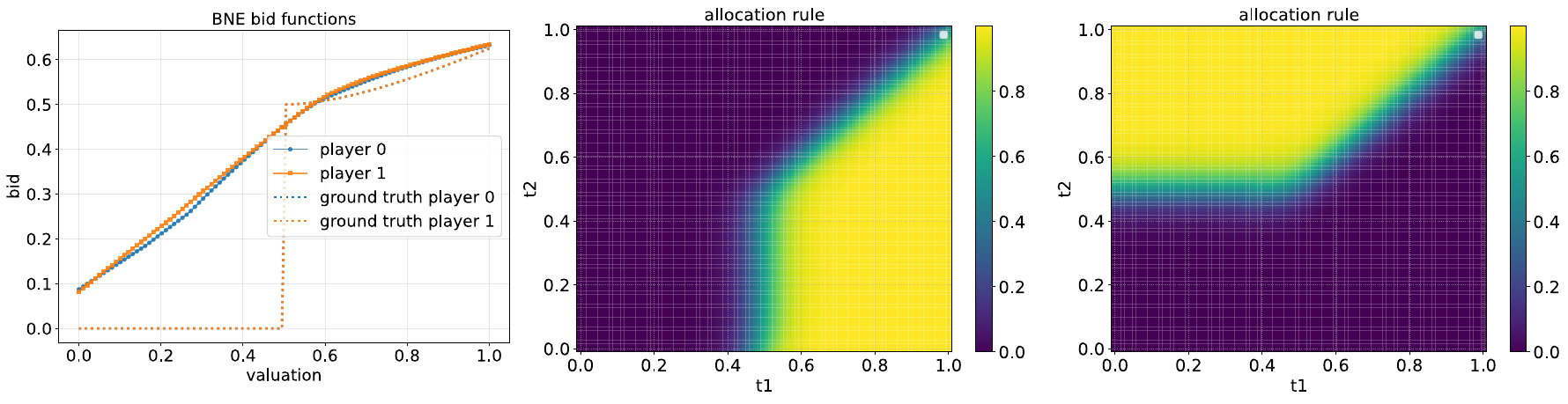}
    \caption{Equilibrium bidding strategy (left) and allocation rule (middle, right) for symmetric single-item auction. For allocation probability, x-axis and y-axis are the bid of bidder 1 and bidder 2 respectively. As shown in the leftmost plot, the learned auction is non-truthful, yet its revenue matches second-price auction with optimal reserve price of $0.5$.}
    \label{fig:sym_vis}
    \vspace{-10pt}
\end{figure}

%% file: conclusion.tex
\section{Conclusion}
We formulate non-truthful auction design as a Stackelberg game and develop a novel fully first-order algorithm to solve it. Our key contribution is establishing non-asymptotic convergence guarantees even when the Jacobian is asymmetric. Our algorithm discovers non-truthful auction with revenue close to optimal truthful auctions even when the mechanism and bidders are neural networks.

\section{Statement of LLM Use}
In this work, we used generative AI tools for assistance with mathematical derivation and method implementation. We have not used generative AI tools for other tasks with required disclosure. Additionally, we used generative AI tools for assistance with debugging and copy-editing to improve readability.
We have reviewed all AI-assisted work. We take responsibility for the final content of this work, including text, claims or artifacts produced with the aid of generative AI.

\section{Ethics Statement}
Our first order algorithm solves Stackelberg games, which is a widely used model in Economics. We recommend domain-specific validation when our algorithm is deployed in the real world.

\section{Reproducibility Statement}
For reproducibility, all assumptions and mathematical proofs are provided in the main text and Appendix with sufficient explanation. Experimental details are listed in Appendix \ref{appx:exp}.

%% file: appendix.tex
\appendix
\section*{Appendix}
\section{Auxiliary Lemmas}
\begin{lemma}[Variational inequality characterization of Nash Equilibrium]\label{lemma:local_nash_VI}
Suppose feasible set $Y_i(x)$ is convex, for all $i\in [N]$. If $(y^*_1, \cdots, y^*_N) \in \text{LNE}(\mathcal{G}(x))$, then
\begin{align*}
    \forall y_i \in U_i: \, \langle \nabla_{y_i} g_i(x, y^*_i y^*_{\neg i}) , y_i - y^*_i\rangle \geq 0
\end{align*}
for some neighborhood $U_i \subseteq Y_i$ of $y^*_i$, for all $i \in [N]$. Under Assumption \ref{assum:VI}, each $g_i$ is convex in $y_i$, the reverse direction also holds with $(y^*_1, \cdots, y^*_N) \in \text{NE}(\mathcal{G}(x))$. Thus, local Nash equilibrium is a necessary and sufficient condition for Nash equilibrium in a strongly monotone and smooth game.
\end{lemma}
\begin{proof}
Let $U = \Pi_i U_i$ and $Y(x)=\Pi_i Y_i(x)$.
Suppose $y^* = (y^*_1, \cdots, y^*_N) \in \text{LNE}(\mathcal{G}(x))$, then for all $i\in[N]$, $y^*_i \in \argmin_{y_i \in U_i} g_i(x, y_i, y^*_{\neg i})$ is a local minimum. Then, 
$\langle \nabla_{y_i} g_i(x,y^*_i, y^*_{\neg i}), d\rangle \geq 0$ for all feasible directions $d \in \mathbb{R}^{n_i}$ such that $y^*_i + td \in U_i$ for some $t>0$. Since $U_i \subseteq Y_i(x)$ is convex, $\langle \nabla_{y_i} g_i(x,y^*_i, y^*_{\neg i}), y_i - y^*_i\rangle \geq 0$, for all $y_i \in U_i$. Then $\langle V(x,y^*), y - y^*\rangle \geq 0$, for all $y \in U$.

\medskip
\noindent
Suppose $\langle V(x,y^*), y - y^*\rangle \geq 0$, for all $y \in Y(x)$. Then, for all $i \in [N]$, $\langle \nabla_{y_i} g_i(x,y^*_i, y^*_{\neg i}), y_i - y^*_i\rangle \geq 0$. Hence, $y^*_i$ is a stationary point. If $V$ is monotone (which implies each $g_i$ is convex in $y_i$), then $y^*_i \in \argmin_{y_i \in Y_i(x)} g_i(x, y_i, y^*_{\neg i})$ is global minimum, then $y^* \in \text{NE}(\mathcal{G}(x)) \subseteq \text{LNE}(\mathcal{G}(x))$. 
\end{proof}

\noindent
Since feasible set for each player is represented by linear constraints, we provide an alternative characterization of Nash equilibrium using KKT condition,
\begin{corollary}[KKT characterization of Nash Equilibrium]\label{corollary:local_nash_KKT}
Under Assumption \ref{assum:LICQ}, if $(y^*_1, \cdots, y^*_N) \in \text{LNE}(\mathcal{G}(x))$, then there exists $\lambda^* \define (\lambda^*_1, \cdots, \lambda_N) \in \mathbb{R}^m$ such that  
\begin{align*}
    (y^*_i, \lambda^*_i) \text{ satisfies the KKT condition of } \Min_{y_i \in Y_i(x)} g_i(x,y_i,y^*_{\neg i}), \forall i \in [N]
\end{align*}
Under Assumption \ref{assum:VI}, each cost function $g_i$ is convex in $y_i$. There exists optimal dual solution $\lambda^* \define (\lambda^*_1, \cdots, \lambda_N) \in \mathbb{R}^m$ such that
$(y^*_1, \cdots, y^*_N) \in \text{NE}(\mathcal{G})$ if and only if 
\begin{align*}
    (y^*_i, \lambda^*_i) \text{ satisfies the KKT condition of } \Min_{y_i \in Y_i(x)} g_i(x,y_i,y^*_{\neg i}), \forall i \in [N]
\end{align*}
\end{corollary}

\begin{lemma}[Fixed point characterization of variational inequality]\label{lemma:fixed_point}
Let $Y \subseteq \mathbb{R}^n$ be a closed, convex set. Denote the Euclidean projection operator as $\Pi_Y(z) = \argmin_{y \in Y} \frac{1}{2}(y-z)^\top (y-z)$. Let $V: \mathbb{R}^n \rightarrow \mathbb{R}^n$. Let $\eta>0$ be an arbitrary number.
$y^* \in Y$ satisfies $\langle V(y^*), y - y^*\rangle \geq 0, \, \forall y \in Y$ if and only if $y^* = \Pi_Y(y^* - \eta V(y^*))$
\end{lemma}
\begin{proof}
By optimality condition of projection, $y^* = \Pi_Y(y^* - \eta V(y^*))$ if and only if $\forall y \in Y$, 
\begin{align*}
    \langle y^* - (y^* - \eta V(y^*)),  y - y^*\rangle & \geq 0 \\
    \langle V(y^*),  y - y^*\rangle & \geq 0
\end{align*}
\end{proof}

\begin{lemma}\label{lemma:matrix_norm}
Let $A \in \mathbb{R}^{m \times n}$. Let $\sigma_{\max}(A) = \sigma_1(A) \geq \sigma_2(A) \geq \cdots \geq \sigma_R(A) = \sigma_{\min}(A) > 0$ be non-zero singular values of $A$. Then,
\begin{align*}
    \|A\| = \max_{x\neq 0} \frac{\|Ax\|}{\|x\|} = \sigma_{\max}(A)
\end{align*}
If $A$ is an invertible square matrix or $A$ has full column rank, then 
\begin{align*}
    \min_{x \neq 0} \frac{\|Ax\|}{\|x\|} = \sigma_{\min}(A)
\end{align*}
\end{lemma}
\begin{proof}
Let $A = U \Sigma V^\top$ be the singular value decomposition of $A$ where $U \in \mathbb{R}^{m \times m}$, $V \in \mathbb{R}^{n \times n}$ and $\Sigma \in \mathbb{R}^{m \times n}$. $U$ and $V$ are orthogonal matrix with orthonormal columns. $\Sigma = \text{diag}(\sigma_1, \cdots, \sigma_R, 0, \cdots, 0)$. Let $x$ be a unit vector, 
\begin{align*}
    \|Ax\|^2
    &=
    x^\top A^\top A x \\
    &= 
    x^\top V \Sigma^\top U^\top U \Sigma V^\top x \\
    &=
    x^\top V \Sigma^\top \Sigma V^\top x
\end{align*}
Let $y=V^\top x$, $\|y\|=\|x\|=1$, 
\begin{align*}
    x^\top A^\top A x
    &= 
    y^\top \Sigma^\top \Sigma y \\
    &\leq
    \sigma_{\max}^2(A)
\end{align*}
which is an achievable upper bound if we choose $x$ to be the corresponding right singular vector. If $A$ is an invertible square matrix, then the diagonal $\Sigma$ does not contain any zeros. If $A$ has full column rank, then 
$\Sigma = \begin{bmatrix} \text{diag}(\sigma_1, \cdots, \sigma_R) \\ 0 \end{bmatrix}$, and the diagonal of $\Sigma^\top \Sigma$ does not contain any zeros. We can minimize over $x$ to get the minimum singular value.  
\end{proof} 

\begin{lemma}\label{lemma:singular_value_1}
If the function $V(y)$ is $\mu$-strongly monotone with respect to $y \in \mathbb{R}^n$, then $\langle d, \totalD{V(y)}{y}d\rangle \geq \mu \|d\|^2$, $\forall y \in Y$, $\forall d \in \mathbb{R}^n$. Furthermore, $\|\totalD{V(y)}{y}\| \geq \sigma_{\min}(\totalD{V(y)}{y}) \geq \mu$ and $\|\totalD{V(y)}{y}^{-1}\| = \sigma_{\max}(\totalD{V(y)}{y}^{-1}) \leq \frac{1}{\mu}$.
\end{lemma}
\begin{proof}
Let $d \in \mathbb{R}^n$ be an arbitrary vector, 
\begin{align*}
    \langle V(y+td)-V(y), td\rangle 
    &\geq 
    \mu \|td\|^2 \\
    \frac{\langle V(y+td)-V(y), d\rangle}{t}
    & \geq
    \mu \|d\|^2 \\
    \langle \totalD{V(y)}{y}d, d\rangle &\geq \mu \|d\|^2 \quad \text{(take $t \rightarrow 0$)}
\end{align*}
Thus, $\totalD{V(y)}{y}$ is invertible and has full rank. The singular value decomposition of $\totalD{V(y)}{y}$ is given by $U \Sigma V^\top$ where $U \in \mathbb{R}^{n \times n}$, $V \in \mathbb{R}^{n \times n}$ and $\Sigma \in \mathbb{R}^{n \times n}$. $U$ and $V$ are orthogonal matrix. $\Sigma = \text{diag}(\sigma_1, \cdots, \sigma_n)$. Following the proof in lemma \ref{lemma:matrix_norm}, we have
\begin{align*}
    \min_{\|v\|=1} \|\totalD{V(y)}{y} v\| = \sigma_{\min}(\totalD{V(y)}{y})
\end{align*}

Let $v \in \mathbb{R}^n$ be a unit vector with $\|v\|=1$,
\begin{align*}
    \|\totalD{V(y)}{y} v\|
    & = 
    \|\totalD{V(y)}{y} v\| \|v\| \\
    &\geq
    \langle \totalD{V(y)}{y} v, v\rangle \\
    &\geq
    \mu \|v\|^2 \\
    &=
    \mu
\end{align*}
Taking minimum over $v$ yields $\sigma_{\min}(\totalD{V(y)}{y}) \geq \mu$. Inverting $\totalD{V(y)}{y}$ using its singular value decomposition yields $\|\totalD{V(y)}{y}^{-1}\| = \sigma_{\max}(\totalD{V(y)}{y}^{-1}) \leq \frac{1}{\mu}$.
\end{proof}

\begin{lemma}\label{lemma:singular_value_2}
If the function $V(y)$ is $C$-Lipschitz with respect to $y \in  \mathbb{R}^n$, then 
$\|\totalD{V(y)}{y}\| = \sigma_{\max}(\totalD{V(y)}{y}) \leq C$. If $\totalD{V(y)}{y}^{-1}$ exists, then $\|\totalD{V(y)}{y}^{-1}\| \geq \sigma_{\min}(\totalD{V(y)}{y}^{-1}) \geq \frac{1}{C}$.
\end{lemma}
\begin{proof}
By definition of $C$-Lipschitzness and using Lemma \ref{lemma:matrix_norm}, $\|\totalD{V(y)}{y}\| = \sigma_{\max}(\totalD{V(y)}{y}) \leq C$. If $\totalD{V(y)}{y}^{-1}$ exists, using singular value decomposition of a full rank matrix, we get $\sigma_{\min}(\totalD{V(y)}{y}^{-1}) \geq \frac{1}{C}$.
\end{proof}

\begin{lemma}\label{lemma:singular_value_3}
Let $G$ be a square matrix such that $\langle d, Gd\rangle \geq \mu \|d\|^2$, $\forall d \in \mathbb{R}^n$ for some $\mu>0$. Let $B$ be a matrix with full row rank, then $B G^{-1}B^\top$ is invertible. If $\sigma_{\max}(G) \leq C$, then $\sigma_{\min}(B G^{-1}B^\top) \geq \frac{\mu}{C^2}\sigma_{\min}(B)^2$ and $\sigma_{\max}(B G^{-1}B^\top) \leq \frac{\sigma_{max}(B)^2}{\mu}$.
\end{lemma}
\begin{proof}
$G$ is invertible. Let $0 \neq v \in \mathbb{R}^n$, $d = G^{-1}v$
\begin{align*}
    v^\top G^{-1} v = v^\top d = d^\top G^\top d = d^\top G d > 0
\end{align*}
Let $0 \neq x \in \mathbb{R}^n$, suppose 
\begin{align*}
    B G^{-1}B^\top x &= 0 \\
    x^\top B G^{-1}B^\top x &= 0
\end{align*}
Since $B$ has full row rank and $v^\top G^{-1} v > 0$ for any $v\neq0$, this results in a contradiction. Hence, $B G^{-1}B^\top$ is invertible. 

\noindent
\begin{align*}
    \sigma_{\max}(B G^{-1}B^\top)
    &\leq
    \|B\|^2 \|G^{-1}\| \\
    & \leq 
    \frac{\|B\|^2}{\mu} \\
    & =
    \frac{\sigma_{\max}(B)^2}{\mu}
\end{align*}

\noindent
Let $x$ be a unit vector. Let $v=B^\top x$ and $d=G^{-1}v$,
\begin{align*}
    \|B G^{-1}B^\top x\|
    &=
    \|B G^{-1}B^\top x\|\|x\|\\
    & \geq 
    x^\top B G^{-1}B^\top x \\
    & =
    d^\top G d \\
    & \geq
    \mu \|d\|^2 \\
    & =
    \mu \|G^{-1}B^\top x\|^2 \\
    & \geq 
    \frac{\mu}{C^2}\|B^\top x\|^2 \\
    & \geq
    \frac{\mu}{C^2}\sigma_{\min}(B)^2
\end{align*}
Minimizing over $x$, we get $\sigma_{\min}(B G^{-1}B^\top) \geq \frac{\mu}{C^2}\sigma_{\min}(B)^2$
\end{proof}

\begin{lemma}\label{lemma:invH}
Given a block matrix $H = \begin{bmatrix} G & B^\top \\ B & 0\end{bmatrix}$. If $B$ has full row rank and $(\forall d\neq 0 \text{\ s.t.\ Bd=0} )\ (d^\top G d > 0)$, then $H$ is invertible.
\end{lemma}
\begin{proof}
We want to show that $H\begin{bmatrix}a \\ b\end{bmatrix} = 0 \implies \begin{bmatrix}a \\b\end{bmatrix} = 0$. Suppose $H\begin{bmatrix}a \\ b\end{bmatrix} = 0$.
Suppose $Ba \neq 0$, then $H\begin{bmatrix}a \\ b\end{bmatrix} = 0$ is impossible by construction. Suppose $Ba=0$. If $a\neq 0$, $a^\top G a > 0$, then $\begin{bmatrix}a & b\end{bmatrix} H \begin{bmatrix}a \\ b\end{bmatrix} = a^\top G a + 2 b^\top B a > 0$, so $ H \begin{bmatrix}a \\ b\end{bmatrix} \neq 0$ which violates the hypothesis. If $a=0$, then $H\begin{bmatrix}a \\ b\end{bmatrix} = \begin{bmatrix} B^\top b \\ 0\end{bmatrix} = \begin{bmatrix}0 \\ 0\end{bmatrix}$. Since $B^\top$ has full column rank, $b=0$. 
\end{proof}

\section{Lipschitzness and Smoothness of Stackelberg Game with Linear Equality Constraints}

\begin{lemma}[Implicit differentiation of KKT system]\label{lemma:KKT_diff}
\begin{align}\label{eq:orig_KKT}
    \begin{bmatrix} \totalD{y^*(x)}{x} \\ \totalD{\lambda^*(x)}{x} \end{bmatrix}
    = 
    \begin{bmatrix}\partialD{V(x,y^*)}{y} & B^\top \\ B & 0\end{bmatrix}^{-1} 
    \begin{bmatrix}-\partialD{V(x,y^*)}{x} \\ -A\end{bmatrix}
\end{align}
\end{lemma}
\begin{proof}
Consider problem \ref{eq:equal_p1}. By strong monotonicity assumption \ref{assum:VI}, each cost function $g_i$ is convex in $y_i$. Hence, KKT condition is the necessary and sufficient condition for each $y^*_i$ to be the Nash equilibrium. Denote $B \define \diag((\partialD{\ei(x,y^*_i)}{y})_{i\in[N]}) \in \real^{m \times n}$ as a block diagonal matrix, and $A \define \begin{bmatrix} \partialD{e^{(1)}(x,y^*_1)}{x} \\ \vdots \\ \partialD{e^{(N)}(x,y^*_N)}{x}\end{bmatrix} \in \real^{m \times \xdim}$. The KKT condition is
\begin{align*}
    \Gamma(x, y^*, \lambda^*) = \begin{bmatrix} V(x,y^*) + B^\top \lambda^* \\ \begin{bmatrix} e^{(1)}(x,y^*_1) \\ \vdots \\ e^{(N)}(x,y^*_N) \end{bmatrix} \end{bmatrix} = 0
\end{align*}
By Lemma \ref{lemma:invH}, the KKT matrix $\partialD{\Gamma(x, y^*, \lambda^*)}{(y^*, \lambda^*)} = \begin{bmatrix}\partialD{V(x,y^*)}{y} & B^\top \\ B & 0\end{bmatrix}$ is invertible. Finally, by implicit function theorem, $y^*(x)$, $\lambda^*(x)$ are continuously differentiable functions of x and their derivatives match the statement of this Lemma.
\end{proof}

\begin{corollary}[Hypergradient]\label{corollary:hypergrad}
\begin{align}
    \nabla F(x)
    &=
    \nabla_x f(x,y^*) + \begin{bmatrix}\partialD{V(x,y^*)}{x}^\top & A^\top\end{bmatrix}
    \begin{bmatrix}\partialD{V(x,y^*)}{y} & B^\top \\ B & 0\end{bmatrix}^{-\top} 
    \begin{bmatrix} -\nabla_yf(x,y^*) \\ 0\end{bmatrix}
\end{align}
\end{corollary}
\begin{proof}
\begin{align*}
    & (\totalD{y^*(x)}{x})^\top \nabla_yf(x,y^*(x)) \\
    & = 
    \begin{bmatrix} \totalD{y^*(x)}{x}^\top &  \totalD{\lambda^*(x)}{x}^\top \end{bmatrix}\begin{bmatrix} \nabla_yf(x,y^*) \\ 0\end{bmatrix} \\
    & = 
    \begin{bmatrix}\partialD{V(x,y^*)}{x}^\top & A^\top\end{bmatrix}
    \begin{bmatrix}\partialD{V(x,y^*)}{y} & B^\top \\ B & 0\end{bmatrix}^{-\top} 
    \begin{bmatrix} -\nabla_yf(x,y^*) \\ 0\end{bmatrix}
\end{align*}
\end{proof}

\begin{lemma}\label{lemma:dy_char}
Suppose
\begin{align}\label{eq:linear_sys_dy}
    \begin{bmatrix}G & B^\top \\ B & 0\end{bmatrix}\begin{bmatrix} d_y \\ d_\lambda \end{bmatrix}
    = 
    \begin{bmatrix}-G_{xy} \\ -A\end{bmatrix}
\end{align}
where 
$G$ is a square matrix satisfying $\langle d, Gd\rangle \geq \mu \|d\|^2$, $\forall d \in \mathbb{R}^n$ for some $\mu>0$
, and $B$ has full row rank, then 
\begin{align*}
    d_y = (I - P^G_B)B^{\dagger}(-A) + P^G_B (-G^{-1} G_{xy})
\end{align*}
where $P^G_B \define I - G^{-1}B^\top(BG^{-1}B^\top)^{-1}B$
\end{lemma}
\begin{proof}

Expanding the equations,
\begin{align}
    d_y &= -G^{-1}B^\top d_\lambda - G^{-1}G_{xy}, \\
    Bd_y &= -A.
\end{align}
Solving for $d_\lambda$, we have 
\begin{align*}
    d_\lambda 
    &= (BG^{-1}B^\top)^{-1} [A-BG^{-1}G_{xy}] \\
    &=(BB^\top)^{-1}B G (I-P^{G}_B) (B^{\dagger}A - G^{-1}G_{xy}) \\
    & = (B^\dagger)^\top G (I-P^{G}_B) (B^{\dagger}A - G^{-1}G_{xy})
\end{align*}
Solving for $d_y$, we have 
\begin{align*}
    d_y 
    & = -G^{-1}B^\top (BG^{-1}B^\top)^{-1} [A-BG^{-1}G_{xy}] - G^{-1}G_{xy} \\
    & = (I - P^G_B)B^{\dagger}(-A) + P^G_B (-G^{-1} G_{xy})
\end{align*}
\end{proof}

\begin{lemma}\label{lemma:proj_norm_1}
Let $P^G_B(x) = I - G^{-1}B^\top(BG^{-1}B^\top)^{-1}B$, where square matrix G satisfies $\langle d, Gd\rangle \geq \mu_g \|d\|^2$ for some $\mu_g>0$ and matrix $B$ has full row rank. Given $\sigma_{\max}(G) \leq C_g$ for some $C_g$, let $\kappa_g = \frac{C_g}{\mu_g}$, then $\|P^G_B\| \leq 1+\frac{\sigma_{\max}(B)^2}{\sigma_{\min}(B)^2}\kappa_g^2 \define M_P$.
\end{lemma}
\begin{proof}
\begin{align*}
   \|P^G_B\| 
   &\leq
   1 + \frac{\|B\|^2}{\mu_g}\frac{1}{\sigma_{\min}(BG^{-1}B^\top)} \\
   & \leq
   1 + \frac{\|B\|^2}{\mu_g}\frac{C_g^2}{\mu_g \sigma_{\min}(B)^2} \\
   & = 
   1+\frac{\sigma_{\max}(B)^2}{\sigma_{\min}(B)^2}\kappa_g^2
\end{align*}
The first inequality uses Lemma \ref{lemma:singular_value_1} and Lemma \ref{lemma:matrix_norm}. The second inequality uses Lemma \ref{lemma:singular_value_3}. The last equality uses Lemma \ref{lemma:matrix_norm}.
\end{proof}

\begin{lemma}\label{lemma:proj_norm_2}
Let $P^G_B \define I - G^{-1}B^\top(BG^{-1}B^\top)^{-1}B$, where square matrix G satisfies $\langle d, Gd\rangle \geq \mu_g \|d\|^2$ for some $\mu_g>0$ and matrix $B$ has full row rank. Given $\sigma_{\max}(G) \leq C_g$ for some $C_g$. Suppose another square matrix $\Approx{G}$ also satisfies the above properties with the same constants $\mu_g$ and $C_g$. Let $\kappa_g = \frac{C_g}{\mu_g}$, then $\|P^G_B - P^{\Approx{G}}_B\| \leq L_P\|G-\Approx{G}\|$ where $L_P \define \frac{C_g^2 \sigma_{\max}(B)^2}{\mu_g^3 \sigma_{\min}(B)^2} + \frac{\sigma_{\max}(B)^2}{\mu_g}\frac{C_g^4 \|B\|^2}{\mu_g^4 \sigma_{\min}(B)^4}$.
\end{lemma}
\begin{proof}
Denote $M = BG^{-1}B^\top$ and $\Approx{M}=B\Approx{G}^{-1}B^\top$,
\begin{align*}
    & P^{\Approx{G}}_B - P^G_B \\
    & = 
    G^{-1}B^\top M^{-1}B 
    - \Approx{G}^{-1}B^\top\Approx{M}^{-1}B \\
    & = 
    (G^{-1}B^\top M^{-1}B 
    - \Approx{G}^{-1}B^\top M^{-1}B)
    + (\Approx{G}^{-1}B^\top M^{-1}B 
    - \Approx{G}^{-1}B^\top\Approx{M}^{-1}B) \\
    & = 
    (G^{-1} - \Approx{G}^{-1})B^\top{M}^{-1}B + \Approx{G}^{-1}B^\top(M^{-1} - \Approx{M}^{-1})B
\end{align*}
Now we proceed to bound $\|G^{-1}-\Approx{G}^{-1}\|$ and $\|M^{-1}-\Approx{M}^{-1}\|$ using the identity $C^{-1}-D^{-1}=C^{-1}(D-C)D^{-1}$ for any invertible matrices $C$ and $D$,
\begin{align*}
    & \|G^{-1}-\Approx{G}^{-1}\| \\
    & \leq
    \frac{1}{\mu_g^2}\|G - \Approx{G}\|
\end{align*}
Using Lemma \ref{lemma:singular_value_3}, we have
$\sigma_{\min}(B G^{-1}B^\top) \geq \frac{\mu_g}{C_g^2}\sigma_{\min}(B)^2$ and $\sigma_{\max}(B G^{-1}B^\top) \leq \frac{\sigma_{max}(B)^2}{\mu_g}$.
\begin{align*}
    & \|M^{-1}-\Approx{M}^{-1}\| \\
    & \leq
    \frac{1}{\sigma_{\min}(B G^{-1}B^\top)^2}\|M - \Approx{M}\| \\
    & \leq
    \frac{C_g^4 \|B\|^2}{\mu_g^2 \sigma_{\min}(B)^4}\|G^{-1} - \Approx{G}^{-1}\| \\
    & \leq 
    \frac{C_g^4 \|B\|^2}{\mu_g^4 \sigma_{\min}(B)^4}\|G - \Approx{G}\|\\
\end{align*}
Now we combine the results together,
\begin{align*}
    & \|P^{\Approx{G}}_B - P^G_B\| \\
    & \leq  
    \|G^{-1} - \Approx{G}^{-1}\|\|B\|^2\|{M}^{-1}\| + \|\Approx{G}^{-1}\|\|B\|^2\|M^{-1} - \Approx{M}^{-1}\| \\
    & \leq
    (\frac{C_g^2 \sigma_{\max}(B)^2}{\mu_g^3 \sigma_{\min}(B)^2} + \frac{\sigma_{\max}(B)^2}{\mu_g}\frac{C_g^4 \|B\|^2}{\mu_g^4 \sigma_{\min}(B)^4})\|G - \Approx{G}\| \\
\end{align*}
\end{proof}

\begin{lemma}[Lipshitzness and smoothness of lower-level solution mappings]\label{lemma:lipschitz_y}
Denote $\kappa = \frac{C}{\mu}$, $y^*(x)$ is $L_y$-Lipschitz and $C_y$-smooth in $x$ for some constants $L_y$ and $C_y$. Moreover, $\lambda^*(x)$ is $L_\lambda$-lipschitz and $C_\lambda$-smooth in $x$ for some constants $L_\lambda$ and $C_\lambda$.
\end{lemma}
\begin{proof}
Denote $G=\partialD{V(x,y^*(x))}{y}$, $G_{xy}=\partialD{V(x,y^*(x))}{x}$. Directly applying Lemma \ref{lemma:dy_char} and Lemma \ref{lemma:proj_norm_1}, we can prove Lipschitzness of $y^*(x)$,
\begin{align*}
    \|\totalD{y^*(x)}{x}\| 
    & = 
    \|(I - P^G_B)B^{\dagger}(-A) + P^G_B (-G^{-1} G_{xy})\| \\
    & \leq \|I-P^G_B\|\|B^\dagger A\| + \|P^G_B\|\|G^{-1}\|\|G_{xy}\| \\
    & \leq (1+M_P)\|B^\dagger A\| + M_P \frac{C}{\mu}\\
    & \define L_y
\end{align*}
Denote $\Approx{y}=y^*(\Approx{x})$, $\Approx{G}=\partialD{V(\Approx{x},\Approx{y})}{y}$ and $\Approx{G}_{xy}=\partialD{V(\Approx{x},\Approx{y})}{x}$. We can also prove smoothness of $y^*(x)$,
\begin{align*}
    & \|\totalD{y^*(\Approx{x})}{x} - \totalD{y^*(x)}{x}\| \\
    & = 
    \|
    (I - P^{\Approx{G}}_B)B^{\dagger}(-A) + P^{\Approx{G}}_B (-\Approx{G}^{-1} \Approx{G}_{xy}) - (I - P^G_B)B^{\dagger}(-A) - P^G_B (-G^{-1} G_{xy})
    \| \\
    & \leq
    \|(P^G_B - P^{\Approx{G}}_B)B^{\dagger}(-A)\|
    + \|(P^{\Approx{G}}_B - P^G_B)(-\Approx{G}^{-1} \Approx{G}_{xy})\|
    + \|P^G_B(G^{-1} G_{xy} - \Approx{G}^{-1} \Approx{G}_{xy})\| \\
    & \leq 
    L_P \|G-\Approx{G}\| \|B^\dagger A\|
    + L_P \|G-\Approx{G}\|\frac{C}{\mu}
    + M_P\|G^{-1} G_{xy} - \Approx{G}^{-1} \Approx{G}_{xy}\|
\end{align*}
where Lemma \ref{lemma:proj_norm_2} and \ref{lemma:proj_norm_1} are applied in the last inequality. Now we can bound $\|G - \Approx{G}\|$. Denote $L_G \define S(1+L_y)$, we have
\begin{align*}
    \|G - \Approx{G}\|
    & \leq 
    \|\partialD{V({x},{y})}{y} - \partialD{V(\Approx{x},{y})}{y}\|
    + \|\partialD{V(\Approx{x},{y})}{y} - \partialD{V(\Approx{x},\Approx{y})}{y}\| \\
    & \leq S(1+L_y) \|x-\Approx{x}\| \\
    & = L_G \|x-\Approx{x}\|
\end{align*}
and we can bound the difference between unconstrained implict gradients $\|G^{-1} G_{xy} - \Approx{G}^{-1} \Approx{G}_{xy}\|$, denote $L_u \define \frac{S(1+L_y)}{\mu} + \frac{CS(1+L_y)}{\mu^2}$, we have
\begin{align*}
    \|G^{-1} G_{xy} - \Approx{G}^{-1} \Approx{G}_{xy}\|
    & \leq \|G^{-1} G_{xy} - {G}^{-1} \Approx{G}_{xy}\| + \|{G}^{-1} \Approx{G}_{xy} - \Approx{G}^{-1} \Approx{G}_{xy}\| \\
    & \leq \frac{1}{\mu}\|G_{xy} - \Approx{G}_{xy}\| + \|\Approx{G}_{xy}\|\|G^{-1}(\Approx{G} - G)\Approx{G}^{-1}\| \\
    & \leq (\frac{S(1+L_y)}{\mu} + \frac{CS(1+L_y)}{\mu^2})\|x-\Approx{x}\| \\
    & =
    L_{u} \|x-\Approx{x}\|
\end{align*}
We have obtained the smoothness bound for $y^*(x)$,
\begin{align*}
    & \|\totalD{y^*(\Approx{x})}{x} - \totalD{y^*(x)}{x}\| \\
    & \leq 
    (L_PL_G\|B^\dagger A\| + \frac{C L_P L_G}{\mu} + M_P(L_u)) \|x-\Approx{x}\| \\
    & \define C_y \|x-\Approx{x}\|
\end{align*}

\noindent
Directly applying Lemma \ref{lemma:dy_char} and Lemma \ref{lemma:proj_norm_1}, we can prove Lipschitzness of $\lambda^*(x)$,
\begin{align*}
    \|\totalD{\lambda^*(x)}{x}\|
    & = \|(B^\dagger)^\top G (I-P^{G}_B) (B^{\dagger}A - G^{-1}G_{xy})\| \\
    & \leq \|B^\dagger\| \|G\| (1+M_P)(\|B^\dagger A\|+\frac{C}{\mu}) \\
    & \define L_\lambda
\end{align*}
The smoothness of $\lambda^*(x)$ is given by
\begin{align*}
    & \|\totalD{\lambda^*(x)}{x} - \totalD{\lambda^*(\Approx{x})}{x}\| \\
    & = 
    \|
    B^\dagger G (I-P^G_B)(B^\dagger A - G^{-1}G_{xy})
    - B^\dagger \Approx{G} (I-P^G_B)(B^\dagger A - G^{-1}G_{xy}) 
    \| \\
    & + 
    \|
    B^\dagger \Approx{G} (I-P^G_B)(B^\dagger A - G^{-1}G_{xy})
    - B^\dagger \Approx{G} (I-P^{\Approx{G}}_B)(B^\dagger A - G^{-1}G_{xy})
    \| \\
    & + 
    \|
    B^\dagger \Approx{G} (I-P^{\Approx{G}}_B)(B^\dagger A - G^{-1}G_{xy})
    - B^\dagger \Approx{G} (I-P^{\Approx{G}}_B)(B^\dagger A - \Approx{G}^{-1}\Approx{G}_{xy})
    \| \\
    & \leq 
    \|B^\dagger\|(1+M_P)(\|B^\dagger A\| + \frac{C}{\mu})\|G-\Approx{G}\| 
    + C \|B^\dagger\| (\|B^\dagger A\| + \frac{C}{\mu})\|P^{\Approx{G}}_B - P^G_B\|
    + C \|B^\dagger\| (1+M_P) \|G^{-1} G_{xy} - \Approx{G}^{-1} \Approx{G}_{xy}\| \\
    & \leq
    (L_G\|B^\dagger\|(1+M_P)(\|B^\dagger A\| + \frac{C}{\mu}) 
    + L_P C \|B^\dagger\| (\|B^\dagger A\| + \frac{C}{\mu}) 
    + C L_u \|B^\dagger\|(1+M_P))
    \|x-\Approx{x}\| \\
    & \define C_\lambda  \|x-\Approx{x}\|
\end{align*}
\end{proof}

\begin{theorem}[Lipschitzness and smoothness of constrained hyper-objective $F$]\label{theorem:smooth_F}
The hyper-objective $F$ is $L_F$-lipschitz and $C_F$-smooth in $x$,
Where $L_F \define L_f (1 + L_y)$ and $C_F \define (C_f(1+L_y) + L_f C_y + L_y C_f + L_y^2C_f)$.
\end{theorem}
\begin{proof}
Let $x, \Approx{x} \in \inputspace$. Using Lemma \ref{lemma:lipschitz_y}, Lipschitzness of $F$ is given by
\begin{align*}
    \|\nabla F(x)\|
    & \leq \|\nabla_xf(x,y^*(x))\| + \|(\totalD{y^*(x)}{x})\|\| \nabla_yf(x,y^*(x))\| \\
    & = L_f (1 + L_y) \\
    & \define L_F
\end{align*}
Smoothness of $F$ is given by
\begin{align*}
    & \|\nabla F(x) - \nabla F(\Approx{x})\| \\
    & \leq 
    \|
    \nabla_xf(x,y^*(x)) - \nabla_xf(\Approx{x},y^*(\Approx{x}))
    \|
    + 
    \|
    \totalD{y^*(x)}{x}^\top \nabla_y f(x,y^*(x)) - \totalD{y^*(\Approx{x})}{x}^\top \nabla_y f(\Approx{x}, y^*(\Approx{x}))
    \| \\
    & \leq
    (C_f(1+L_y) + L_f C_y + L_y C_f + L_y^2C_f)\|x - \Approx{x}\| \\
    & \define
    C_F\|x - \Approx{x}\|
\end{align*}
\end{proof}

\begin{corollary}[Smoothness of the unconstrained hyper-objective]
\label{cor:unconstrained_F_smooth}
Consider Problem~\ref{eq:equal_p1} with unconstrained follower strategy sets
$Y_i=\mathbb{R}^{n_i}$ for all $i\in[N]$. Suppose the equilibrium mapping $y^*(x)$ is
$L_y$-Lipschitz and its Jacobian $\frac{d}{dx}y^*(x)$ is
$C_y$-Lipschitz. Then the hyper-objective
\[
F(x)=f(x,y^*(x))
\]
is $L_F$-Lipschitz and $C_F$-smooth, where
\[
L_F:=L_f(1+L_y)
\]
and
\[
C_F
:=
C_f(1+L_y)
+
L_f C_y
+
L_y C_f
+
L_y^2 C_f.
\]
\end{corollary}

\begin{proof}
In the unconstrained setting, the hypergradient is
\[
\nabla F(x)
=
\nabla_x f(x,y^*(x))
+
\left(\frac{d}{dx}y^*(x)\right)^\top
\nabla_y f(x,y^*(x)).
\]
The proof of Theorem~\ref{theorem:smooth_F} only uses the
$L_y$-Lipschitzness of $y^*(x)$, the $C_y$-Lipschitzness of
$\frac{d}{dx}y^*(x)$, and the Lipschitzness and smoothness of $f$.
Hence, after removing all constraint and dual-variable terms, the same
argument gives
\[
\|\nabla F(x)\|\leq L_f(1+L_y)=L_F,
\]
and, for any $x,\hat x\in X$,
\[
\|\nabla F(x)-\nabla F(\hat x)\|
\leq
\left(
C_f(1+L_y)
+
L_f C_y
+
L_y C_f
+
L_y^2 C_f
\right)
\|x-\hat x\|.
\]
Therefore, $F$ is $L_F$-Lipschitz and $C_F$-smooth.
\end{proof}

\section{Symmetric Jacobian Stackelberg Games: Unconstrained and Equality-Constrained Cases}\label{appx:sym}
We first provide auxiliary lemmas for section~\ref{sec:sym} with linear equality constraints. We then provide complete proof for theorem~\ref{theorem:grad_approx} and~\ref{theorem:sym_equality_convergence} specialized to the unconstrained case. 

\subsection{Auxiliary Lemmas in constrained setting}
\begin{lemma}\label{lemma:pert_game}
The perturbed game $\mathcal{G}_\delta \define (x, [N], \{Y_i\}_{i\in[N]}, \{g_i + \delta f\}_{i\in [N]})$. $\mathcal{G}_\delta$ is a $C_\delta \define (C + \delta C_f)$-smooth , $\mu_\delta \define (\mu - \delta C_f)$-strongly monotone and $S_\delta \define (S+\delta S_f)$-Hessian smooth. If we assume $\delta \leq \frac{\mu}{2C_f}$, then we set $\mu_\delta = \frac{\mu}{2}$, $C_\delta = C+\frac{\mu}{2}$ and $S_\delta = S + \frac{S_f\mu}{2C_f}$, so that $\mu_\delta$, $C_\delta$ and $S_\delta$ are independent of $\delta$.
\end{lemma}
\begin{proof}
Let $x, x^\prime \in \inputspace$. $y, y^\prime \in Y$, To prove strong monotonicity,
\begin{align*}
    &\langle V(x,y^\prime) + \delta \nabla_y f(x,y^\prime) - (V(x,y) + \delta \nabla_y f(x,y)), y^\prime - y \rangle \\
    & \geq \mu \|y^\prime - y\|^2 + \langle \delta \nabla_y f(x,y^\prime) - \delta \nabla_y f(x,y), y^\prime - y \rangle \\
    & \geq \mu \|y^\prime - y\|^2  -\delta \|\nabla_y f(x,y^\prime) - \nabla_y f(x,y)\|\| y^\prime - y \|  \\
    & \geq (\mu - \delta C_f)\| y^\prime - y \|^2
\end{align*}
To prove smoothness,
\begin{align*}
    &  \|V(x^\prime,y^\prime) + \delta \nabla_y f(x^\prime,y^\prime) - (V(x,y) + \delta \nabla_y f(x,y))\| \\
    & \leq C\|y^\prime - y\| + \delta \|\nabla_y f(x,y^\prime) - \nabla_y f(x,y)\| \\
    & \leq (C + \delta C_f)\| y^\prime - y \|
\end{align*}
To prove Hessian-smoothness,
\begin{align*}
    &  \partialD{V(x^\prime,y^\prime)}{(x,y)} + \delta \partialD{\nabla_y f(x^\prime, y^\prime)}{(x,y)}- \partialD{V(x, y)}{(x,y)} - \delta \partialD{\nabla_y f(x, y)}{(x,y)}\|\\
    & \leq S \|(x^\prime, y^\prime) - (x, y)\| + \delta \|\partialD{\nabla_y f(x^\prime, y^\prime)}{(x,y)} - \partialD{\nabla_y f(x, y)}{(x,y)}\| \\
    & \leq (S + \delta S_f)\|(x^\prime, y^\prime) - (x, y)\|
\end{align*}
If we assume $\delta \leq \frac{\mu}{2C_f}$, then $\mu_\delta \geq \frac{\mu}{2}$, $C_\delta \leq C+\frac{\mu}{2}$ and $S_\delta \leq S + \frac{S_f\mu}{2C_f}$. 
\end{proof}

\begin{lemma}\label{lemma:y_delta}
Let $y^*_\delta$ be the Nash equilibrium of the perturbed game $\mathcal{G}_\delta$. If $\delta \in [0, \Delta]$ where $\Delta \leq \frac{\mu}{2C_f}$, then $$\|y^*_\delta - y^*\| \leq \frac{2L_f}{\mu}\delta = \bigO(\delta)$$ 
\end{lemma}
\begin{proof}
By optimality of $y^*_\delta$ and $y^*$ in their respective games, 
\begin{align*}
   & \langle V(x,y^*_\delta) + \delta \nabla_yf(x,y^*_\delta), y^* -y^*_\delta\rangle \geq 0 \\
   & \langle V(x,y^*), y^*_\delta -y^* \rangle \geq 0 \\
\end{align*}
By summing up the inequalities,
\begin{align*}
   \langle V(x,y^*) - V(x,y^*_\delta)- \delta \nabla_yf(x,y^*_\delta), y^*_\delta -y^* \rangle \geq 0 \\
\end{align*}
By srtong monotonicity of $V + \delta \nabla_yf$, 
\begin{align*}
    \langle V(x,y^*_\delta)+ \delta \nabla_yf(x,y^*_\delta) - V(x,y^*) - \delta \nabla_yf(x,y^*), \ y^*_\delta -y^* \rangle \geq (\mu - \delta C_f)\|y^*_\delta - y^*\|^2 \\
\end{align*}
We can bound the left-hand side by
\begin{align*}
    & \langle V(x,y^*_\delta)+ \delta \nabla_yf(x,y^*_\delta) - V(x,y^*) - \delta \nabla_yf(x,y^*), \ y^*_\delta -y^* \rangle  \\
    & =  \langle V(x,y^*_\delta)+ \delta \nabla_yf(x,y^*_\delta) - V(x,y^*) + V(x, y^*) - (V(x,y^*) + \delta \nabla_yf(x,y^*)), \ y^*_\delta -y^* \rangle \\
    & \leq \langle V(x, y^*)- (V(x,y^*) + \delta \nabla_yf(x,y^*)), \ y^*_\delta -y^* \rangle \\
    & \leq \delta \|\nabla_yf(x,y^*))\| \|y^*_\delta -y^* \| \\
    & \leq \delta L_f \|y^*_\delta -y^* \|
\end{align*}
Hence, we get the desired inequality,
\begin{align*}
    \|y^*_\delta - y^*\|
    & \leq \frac{\delta}{\mu - \delta C_f} L_f \\
    & \leq \frac{2L_f}{\mu}\delta
\end{align*}
\end{proof}

\begin{lemma}[Implicit differentiation of perturbed game]\label{lemma:KKT_diff_perturb}
\begin{align*}
    \begin{bmatrix} \totalD{y^*_\delta}{\delta}|_{\delta=0} \\ \totalD{\lambda^*_\delta}{\delta}|_{\delta=0} \end{bmatrix}
    =  \begin{bmatrix}\partialD{V(x,y^*)}{y} & B^\top \\ B & 0\end{bmatrix}^{-1} \begin{bmatrix}- \nabla_{y}f(x,y^*) \\ 0\end{bmatrix}
\end{align*}
\end{lemma}
\begin{proof}
By Lemma \ref{lemma:pert_game}, $\mathcal{G}_\delta$ is a $(\mu - \delta C_f)$-strongly monotone, $(C + \delta C_f)$-smooth and $(S+\delta S_f$-Hessian smooth. Let $y^*_\delta$, $\lambda^*_\delta$ be the Nash equilibrium and lagrange multiplier for $\mathcal{G}_\delta$. We can again apply implicit function theorem and show that $y^*_\delta$ and $\lambda^*_\delta$ are continuously differentiable functions of $\delta$ and the derivatives of them are given by 

\begin{align}\label{eq:pert_KKT}
    \begin{bmatrix}\partialD{V(x,y^*_\delta)}{y}+\delta \nabla^2_{yy}f(x,y^*_\delta) & B^\top \\ B & 0\end{bmatrix}\begin{bmatrix} \totalD{y^*_\delta}{\delta} \\ \totalD{\lambda^*_\delta}{\delta} \end{bmatrix}
    =  \begin{bmatrix}- \nabla_{y}f(x,y^*_\delta) \\ 0\end{bmatrix}
\end{align}
we can prove that $y^* = \lim_{\delta \rightarrow 0} y^*_\delta$ (Lemma \ref{lemma:y_delta}). Take $\delta \rightarrow 0$, we have the desired result.
\end{proof}

\begin{lemma}\label{lemma:limit_v}
Assume that $\partialD{V(x,y^*)}{y}$ is symmetric. Denote $y^*$, $\lambda^*$ as constants with respect to $x$, then the hypergradient $\nabla F(x)$ can be rewritten without second-order derivatives:
\begin{align*}
    & v_\delta(x,y^*_\delta,\lambda^*_\delta) = \frac{1}{\delta} \sum_{i=1}^N \nabla_x[g_i(x,y^*_{\delta,i}, y^*_{\neg i}) + (\lambda^*_{\delta,i})^\top \ei(x,y^*_i)] - \nabla_x[g_i(x,y^*_{i}, y^*_{\neg i}) + (\lambda^*_{i})^\top \ei(x,y^*_i)] \\ 
    & (\totalD{y^*(x)}{x})^\top \nabla_yf(x,y^*(x))
    = \lim_{\delta \rightarrow0} v_\delta(x,y^*_\delta,\lambda^*_\delta)
\end{align*}
\end{lemma}
\begin{proof}
\begin{align*}
    & (\totalD{y^*(x)}{x})^\top \nabla_yf(x,y^*(x)) \\
    & = 
    \begin{bmatrix} \totalD{y^*(x)}{x}^\top &  \totalD{\lambda^*(x)}{x}^\top \end{bmatrix}\begin{bmatrix} \nabla_yf(x,y^*) \\ 0\end{bmatrix} \\
    & = 
    \begin{bmatrix}\partialD{V(x,y^*)}{x}^\top & A^\top\end{bmatrix}
    \begin{bmatrix}\partialD{V(x,y^*)}{y} & B^\top \\ B & 0\end{bmatrix}^{-1} 
    \begin{bmatrix} -\nabla_yf(x,y^*) \\ 0\end{bmatrix} \\
    & =  
    \begin{bmatrix}\partialD{V(x,y^*)}{x}^\top & A^\top\end{bmatrix}
    \begin{bmatrix} \totalD{y^*_\delta}{\delta} \\ \totalD{\lambda^*_\delta}{\delta} \end{bmatrix}_{\delta=0} \\
    & = 
    \begin{bmatrix} \partialD{\nabla_{y_1}g_1(x,y^*)}{x_1} & \cdots & \partialD{\nabla_{y_1}g_1(x,y^*)}{x_\xdim} \\ \vdots & \cdots & \vdots \\ 
    \partialD{\nabla_{y_N}g_N(x,y^*)}{x_1} & \cdots & \partialD{\nabla_{y_N}g_N(x,y^*)}{x_\xdim}\end{bmatrix}^\top
    \totalD{y^*_\delta}{\delta}|_{\delta=0} 
    + \sum_{i=1}^N (\partialD{\ei(x,y^*_i)}{x})^\top \totalD{\lambda^*_{\delta, i}}{\delta}|_{\delta=0} \\
    & = 
    \sum_{i=1}^N  \partialD{\nabla_xg_i(x,y^*)}{y_i} \totalD{y^*_{\delta,i}}{\delta}|_{\delta=0} + 
    \partialD{\ei(x,y^*_i)}{x})^\top \totalD{\lambda^*_{\delta, i}}{\delta}|_{\delta=0}\\
    & = 
    \sum_{i=1}^N  \lim_{\delta\rightarrow 0}\frac{1}{\delta}[\nabla_xg_i(x,y^*_{\delta, i}, y^*_{\neg i}) - \nabla_xg_i(x,y^*_{i}, y^*_{\neg i})]  + 
    \partialD{\ei(x,y^*_i)}{x})^\top \lim_{\delta\rightarrow 0} \frac{1}{\delta}[\lambda^*_{\delta, i} - \lambda^*_i] \\
    & = 
    \lim_{\delta\rightarrow 0} \frac{1}{\delta} \sum_{i=1}^N \nabla_x[g_i(x,y^*_{\delta,i}, y^*_{\neg i}) + (\lambda^*_{\delta,i})^\top \ei(x,y^*_i)] - \nabla_x[g_i(x,y^*_{i}, y^*_{\neg i}) + (\lambda^*_{i})^\top \ei(x,y^*_i)]
\end{align*}
\end{proof}

\begin{lemma}[Lipshitzness and smoothness of perturbed lower-level solution mappings]\label{lemma:lipschitz_y_delta}
Denote $\kappa_\delta = \frac{C_\delta}{\mu_\delta}$ and suppose $\delta \leq \Delta \define \frac{\mu}{2C_f}$, $y^*_\delta$ is $L_{y,\delta}$-Lipschitz and $C_{y,\delta}$-smooth in $\delta \in \mathbb{R}_{>0}$. $\lambda^*_\delta$ is $L_{\lambda, \delta}$-lipschitz and $C_{\lambda, \delta}$-smooth in $\delta$. Moreover, $L_{y,\delta}$, $C_{y,\delta}$, $L_{\lambda, \delta}$, $C_{\lambda, \delta}$ are independent of the value of $\delta$, and can be treated as constants.
\end{lemma}
\begin{proof}
$(\totalD{y^*_\delta}{\delta}, \totalD{\lambda^*_\delta}{\delta})$ are the solutions to equation \ref{eq:pert_KKT}. Let $0 < \delta, \Approx{\delta} \leq \Delta$.
Using Lemma \ref{lemma:lipschitz_y}, substitute 
$G=\partialD{V(x,y^*_\delta)}{y}+\delta \nabla^2_{yy}f(x,y^*_\delta)$, 
$G_{xy}=\nabla_{y}f(x,y^*_\delta)$,
$A = 0$,
$\Approx{y}=y^*_{\Approx{\delta}}$, $\Approx{G}=\partialD{V(x,y^*_{\Approx{\delta}})}{y}+{\Approx{\delta}} \nabla^2_{yy}f(x,y^*_{\Approx{\delta}})$,
$\Approx{G}_{xy} = \nabla_{y}f(x,y^*_{\Approx{\delta}})$
$\kappa_\delta \define \frac{C_\delta}{\mu_\delta}$. Denote
$\|P^G_B - P^{\Approx{G}}_B\| \leq L_{P,\delta}\|G-\Approx{G}\|$ and $\|P^G_B\| \leq M_{P,\delta}$ according to Lemmas \ref{lemma:proj_norm_1} and \ref{lemma:proj_norm_2}.

\noindent
The Lipschitzness of $y^*_\delta$ is given by
\begin{align*}
    \|\totalD{y^*_\delta}{\delta}\| 
    & = 
    \|(I - P^G_B)B^{\dagger}(-A) + P^G_B (-G^{-1} G_{xy})\| \\
    & \leq \|I-P^G_B\|\|B^\dagger A\| + \|P^G_B\|\|G^{-1}\|\|G_{xy}\| \\
    & \leq (1+M_{P,\delta})(0 + \frac{L_f}{\mu_\delta}) \\
    & \define L_{y,\delta}
\end{align*}

\noindent
Denote $L_{G,\delta} \define (S L_{y,\delta} + C_f + \Delta S_f L_{y,\delta})$, we have
\begin{align*}
    \|G - \Approx{G}\|
    & = 
    \|
    \partialD{V(x,y^*_\delta)}{y}+\delta \nabla^2_{yy}f(x,y^*_\delta)
    - (\partialD{V(x,y^*_{\Approx{\delta}})}{y}+{\Approx{\delta}} \nabla^2_{yy}f(x,y^*_{\Approx{\delta}}))
    \| \\
    & \leq
    S L_{y,\delta} \|\delta - \Approx{\delta}\| 
    + \|\delta \nabla^2_{yy}f(x,y^*_\delta) - \Approx{\delta} \nabla^2_{yy}f(x,y^*_\delta)\|
    + \| \Approx{\delta} \nabla^2_{yy}f(x,y^*_\delta) - {\Approx{\delta}} \nabla^2_{yy}f(x,y^*_{\Approx{\delta}})\|
    \\
    & \leq  
    (S L_{y,\delta} + C_f + \Delta S_f L_{y,\delta}) \|\delta - \Approx{\delta}\| \\
    & \define L_{G,\delta} \|\delta-\Approx{\delta}\|
\end{align*}

\noindent
Denote $L_{u,\delta} = C_fL_{y,\delta} + \frac{L_f}{\mu_\delta^2}L_{G,\delta}$, we have
\begin{align*}
    \|G^{-1} G_{xy} - \Approx{G}^{-1} \Approx{G}_{xy}\|
    & \leq 
    \|G^{-1} G_{xy} - {G}^{-1} \Approx{G}_{xy}\| 
    + \|{G}^{-1} \Approx{G}_{xy} - \Approx{G}^{-1} \Approx{G}_{xy}\| \\
    & \leq 
    \frac{1}{\mu_\delta}\|G_{xy} - \Approx{G}_{xy}\| 
    + \|\Approx{G}_{xy}\|\|G^{-1}(\Approx{G} - G)\Approx{G}^{-1}\| \\
    & \leq
    (C_fL_{y,\delta} + \frac{L_f}{\mu_\delta^2}L_{G,\delta})\|\delta-\Approx{\delta}\| \\
    & \define
    L_{u,\delta} \|\delta-\Approx{\delta}\|
\end{align*}

\noindent
The smoothness of $y^*_\delta$ is given by
\begin{align*}
    & \|\totalD{y^*_{\Approx{\delta}}}{x} - \totalD{y^*_{\delta}}{x}\| \\
    & \leq 
    L_{P,\delta} \|G-\Approx{G}\| 0
    + L_{P,\delta} \|G-\Approx{G}\|\frac{L_f}{\mu_\delta}
    + M_{P, \delta}\|G^{-1} G_{xy} - \Approx{G}^{-1} \Approx{G}_{xy}\| \\
    & \leq (L_{P,\delta} + L_{P,\delta}L_{G,\delta}\frac{L_f}{\mu_\delta} + M_{P, \delta}L_{u,\delta}) \|\delta - \Approx{\delta}\| \\
    & \define C_{y,\delta} \|\delta - \Approx{\delta}\|
\end{align*}

\noindent
The Lipschitzness of $\lambda^*_\delta$ is given by
\begin{align*}
    \|\totalD{\lambda^*_\delta}{\delta}\|
    & = \|(B^\dagger)^\top G (I-P^{G}_B) (B^{\dagger}A - G^{-1}G_{xy})\| \\
    & \leq \|B^\dagger\| \|C_\delta\| (1+M_{P, \delta})(0+\frac{L_f}{\mu_\delta}) \\
    & \define L_{\lambda,\delta}
\end{align*}
The smoothness of $\lambda^*_\delta$ is given by
\begin{align*}
    & \|\totalD{\lambda^*_\delta}{\delta} - \totalD{\lambda^*_{\Approx{\delta}}}{\delta}\| \\
    & = 
    \|
    B^\dagger G (I-P^G_B)(B^\dagger A - G^{-1}G_{xy})
    - B^\dagger \Approx{G} (I-P^G_B)(B^\dagger A - G^{-1}G_{xy}) 
    \| \\
    & + 
    \|
    B^\dagger \Approx{G} (I-P^G_B)(B^\dagger A - G^{-1}G_{xy})
    - B^\dagger \Approx{G} (I-P^{\Approx{G}}_B)(B^\dagger A - G^{-1}G_{xy})
    \| \\
    & + 
    \|
    B^\dagger \Approx{G} (I-P^{\Approx{G}}_B)(B^\dagger A - G^{-1}G_{xy})
    - B^\dagger \Approx{G} (I-P^{\Approx{G}}_B)(B^\dagger A - \Approx{G}^{-1}\Approx{G}_{xy})
    \| \\
    & \leq 
    \|B^\dagger\|(1+M_{P, \delta})(0 + \frac{L_f}{\mu_\delta})\|G-\Approx{G}\| 
    + C_\delta \|B^\dagger\| (0 + \frac{L_f}{\mu_\delta})\|P^{\Approx{G}}_B - P^G_B\|
    + C_\delta \|B^\dagger\| (1+M_{P, \delta}) \|G^{-1} G_{xy} - \Approx{G}^{-1} \Approx{G}_{xy}\| \\
    & \leq
    (\|B^\dagger\|(1+M_{P, \delta})\frac{L_f}{\mu_\delta}L_{G,\delta}
    + C_\delta \|B^\dagger\| \frac{L_f}{\mu_\delta} L_{P,\delta}
    + C_\delta \|B^\dagger\| (1+M_{P, \delta}) L_{u,\delta}
    )
    \|\delta-\Approx{\delta}\| \\
    & \define C_{\lambda,\delta}  \|\delta-\Approx{\delta}\|
\end{align*}
\end{proof}

\subsection{Omitted proof in Section \ref{sec:sym} in unconstrained games}
\gradApprox*
\begin{proof}
Let $y_t^*$ denote the Nash equilibrium of the perturbed game
\[
G_t := (x,[N],\{\mathbb{R}^{n_i}\}_{i\in[N]},\{g_i+t f\}_{i\in[N]}),
\qquad t\in[0,\delta].
\]
The equilibrium condition is
\[
V(x,y_t^*)+t\nabla_y f(x,y_t^*)=0.
\]
Differentiating with respect to $t$ gives
\[
\left(
J(x,y_t^*)+t\nabla_{yy}^2 f(x,y_t^*)
\right)
\frac{d}{dt}y_t^*
=
-\nabla_y f(x,y_t^*).
\]
Taking $t=0$,
\[
\left.\frac{d}{dt}y_t^*\right|_{t=0}
=
J(x,y^*)^{-1}\left(-\nabla_y f(x,y^*)\right).
\]

By implicit differentiation of the original equilibrium condition
$V(x,y^*)=0$,
\[
\left(\frac{d}{dx}y^*(x)\right)^\top
\nabla_y f(x,y^*)
=
\frac{\partial}{\partial x}V(x,y^*)^\top
J(x,y^*)^{-\top}
\left(-\nabla_y f(x,y^*)\right).
\]
Since $J(x,y^*)$ is symmetric,
\[
J(x,y^*)^{-\top}=J(x,y^*)^{-1}.
\]
Hence,
\[
\left(\frac{d}{dx}y^*(x)\right)^\top
\nabla_y f(x,y^*)
=
\frac{\partial}{\partial x}V(x,y^*)^\top
\left.\frac{d}{dt}y_t^*\right|_{t=0}.
\]
By equality of mixed partial derivatives,
\[
\frac{\partial}{\partial x}V(x,y^*)^\top
\left.\frac{d}{dt}y_t^*\right|_{t=0}
=
\sum_{i=1}^N
\frac{\partial}{\partial y_i}\nabla_x g_i(x,y^*)
\left.\frac{d}{dt}y_{t,i}^*\right|_{t=0}.
\]

For $0\leq t\leq\delta\leq \mu/(2C_f)$, the perturbed operator
$V+t\nabla_y f$ is at least $\mu/2$-strongly monotone. Under
Assumptions~3.4 and~3.5, there exist constants $L_{y,\delta}$ and
$C_{y,\delta}$, independent of $\delta$, such that
\[
\|y_t^*-y^*\|\leq L_{y,\delta}t,
\qquad
\left\|\frac{d}{dt}y_t^*\right\|\leq L_{y,\delta},
\]
and
\[
\left\|
\frac{d}{dt}y_t^*
-
\left.\frac{d}{dt}y_t^*\right|_{t=0}
\right\|
\leq C_{y,\delta}t.
\]

By definition of finite difference in~\eqref{eq:finite_diff},
\[
v_\delta(x,y_\delta^*)
=
\frac{1}{\delta}
\sum_{i=1}^N
\left[
\nabla_x g_i(x,y_{\delta,i}^*,y_{-i}^*)
-
\nabla_x g_i(x,y_i^*,y_{-i}^*)
\right].
\]
For all $i\in[N]$, by the fundamental theorem of calculus,
\[
\nabla_x g_i(x,y_{\delta,i}^*,y_{-i}^*)
-
\nabla_x g_i(x,y_i^*,y_{-i}^*)
=
\int_0^\delta
\frac{\partial}{\partial y_i}
\nabla_x g_i(x,y_{t,i}^*,y_{-i}^*)
\frac{d}{dt}y_{t,i}^*\,dt.
\]

We bound each term by
\[
\begin{aligned}
&
\left\|
\frac{
\nabla_x g_i(x,y_{\delta,i}^*,y_{-i}^*)
-
\nabla_x g_i(x,y_i^*,y_{-i}^*)
}{\delta}
-
\frac{\partial}{\partial y_i}
\nabla_x g_i(x,y^*)
\left.\frac{d}{dt}y_{t,i}^*\right|_{t=0}
\right\|
\\
&\leq
\frac{1}{\delta}
\int_0^\delta
\left\|
\frac{\partial}{\partial y_i}
\nabla_x g_i(x,y_{t,i}^*,y_{-i}^*)
\frac{d}{dt}y_{t,i}^*
-
\frac{\partial}{\partial y_i}
\nabla_x g_i(x,y^*)
\left.\frac{d}{dt}y_{t,i}^*\right|_{t=0}
\right\|dt
\\
&\leq
\frac{1}{\delta}
\int_0^\delta
\left(
S L_{y,\delta}^2 t
+
C C_{y,\delta}t
\right)dt
\\
&=
\frac{1}{2}
\left(
S L_{y,\delta}^2
+
C C_{y,\delta}
\right)\delta.
\end{aligned}
\]

Summing over $i=1,\ldots,N$, we obtain
\[
\begin{aligned}
&
\left\|
v_\delta(x,y_\delta^*)
-
\left(\frac{d}{dx}y^*(x)\right)^\top
\nabla_y f(x,y^*(x))
\right\|
\\
&\leq
\frac{N}{2}
\left(
S L_{y,\delta}^2
+
C C_{y,\delta}
\right)\delta
=
\bigO(N\delta).
\end{aligned}
\]
\end{proof}

\begin{restatable}{theorem}{symEqualityConvergence}\label{theorem:sym_equality_convergence}
    Consider problem \ref{eq:equal_p1} and assume that the Jacobian $J = \partialD{V(x,y^*)}{y}$ is symmetric. Under assumptions \ref{assum:VI} and  \ref{assum:optim}, Algorithm \ref{alg:ffo_algo} specialized to unconstrained case outputs an $\mathcal{O}(\epsilon)$-stationary point of the hyperobjective $F$ with $\tilde{\bigO}(N\epsilon^{-2})$ oracle calls to $f$ and $g_i$. 
\end{restatable}
\begin{proof}
By Assumptions~3.4 and~3.5 and the implicit function theorem, the
lower-level solution mapping $y^*(x)$ is Lipschitz and smooth. By corollary~\ref{cor:unconstrained_F_smooth}, the hyperobjective $F$ is $C_F$-smooth for some constant
$C_F>0$.

Suppose we have access to approximate lower-level solutions
$\hat y,\hat y_\delta$ such that
\[
\|y^*-\hat y\|\leq \delta^2,
\qquad
\|y_\delta^*-\hat y_\delta\|\leq \delta^2.
\]
We construct the inexact hypergradient
\[
\nabla \hat F(x)
=
\nabla_x f(x,\hat y)
+
\hat v_\delta(x,\hat y_\delta),
\]
where
\[
\hat v_\delta(x,\hat y_\delta)
=
\frac{1}{\delta}
\sum_{i=1}^N
\left[
\nabla_x g_i(x,\hat y_{\delta,i},\hat y_{-i})
-
\nabla_x g_i(x,\hat y_i,\hat y_{-i})
\right].
\]

We first bound the error due to the approximate lower-level solutions:
\[
\begin{aligned}
&
\|\hat v_\delta(x,\hat y_\delta)
-
v_\delta(x,y_\delta^*)\|
\\
&\leq
\frac{1}{\delta}
\sum_{i=1}^N
\|
\nabla_x g_i(x,\hat y_{\delta,i},\hat y_{-i})
-
\nabla_x g_i(x,y_{\delta,i}^*,y_{-i}^*)
\|
\\
&\quad+
\frac{1}{\delta}
\sum_{i=1}^N
\|
\nabla_x g_i(x,y_i^*,y_{-i}^*)
-
\nabla_x g_i(x,\hat y_i,\hat y_{-i})
\|
\\
&\leq
\frac{1}{\delta}
\sum_{i=1}^N
\left(
2C\delta^2+C\delta^2
\right)
\\
&=
3CN\delta
=
\bigO(N\delta).
\end{aligned}
\]
Also,
\[
\|\nabla_x f(x,\hat y)-\nabla_x f(x,y^*)\|
\leq C_f\|\hat y-y^*\|
\leq C_f\delta^2.
\]
Combining the above inequalities with Theorem~4.2,
\[
\begin{aligned}
\|\nabla \hat F(x)-\nabla F(x)\|
&\leq
\|\nabla_x f(x,\hat y)-\nabla_x f(x,y^*)\|
\\
&\quad+
\|\hat v_\delta(x,\hat y_\delta)
-
v_\delta(x,y_\delta^*)\|
\\
&\quad+
\left\|
v_\delta(x,y_\delta^*)
-
\left(\frac{d}{dx}y^*(x)\right)^\top
\nabla_y f(x,y^*)
\right\|
\\
&\leq
C_f\delta^2+\bigO(N\delta)
\\
&=
\bigO(N\delta).
\end{aligned}
\]

Set
\[
\delta=\frac{\epsilon}{N}.
\]
Then
\[
\|\nabla \hat F(x)-\nabla F(x)\|
\leq \bigO(\epsilon).
\]
Using $\nabla \hat F(x)$ as the inexact gradient in
Lemma~D.2, an $\bigO(\epsilon)$-stationary point of $F$ can be obtained
with
\[
\bigO(\epsilon^{-2})
\]
calls to $\nabla \hat F$.

It remains to bound the oracle complexity for computing each
$\nabla \hat F$. Since the lower-level game is unconstrained, PGD
reduces to gradient descent,
\[
y^{k+1}=y^k-\eta V(x,y^k).
\]
By Assumption~3.4 and Lemma~E.1, gradient descent converges linearly to
$y^*$. Moreover, for $\delta\leq\mu/(2C_f)$, the perturbed operator
\[
V_\delta(x,y)
=
V(x,y)+\delta\nabla_y f(x,y)
\]
is at least $\mu/2$-strongly monotone and $(C+\delta C_f)$-smooth, so
gradient descent also converges linearly to $y_\delta^*$.

For each call to $\nabla \hat F(x)$,
\[
\bigO\left(\log\frac{1}{\delta^2}\right)
=
\bigO\left(\log(N^2\epsilon^{-2})\right)
\]
lower-level iterations are sufficient to obtain $\delta^2$-accurate
solutions $\hat y$ and $\hat y_\delta$. Each lower-level iteration
requires one first-order oracle call to each of the $N$ follower cost
functions. Therefore, the total oracle complexity is
\[
\begin{aligned}
&
\bigO\left(
N\epsilon^{-2}
\log(N^2\epsilon^{-2})
\right)
\\
&=
\widetilde{\bigO}(N\epsilon^{-2}).
\end{aligned}
\]
\end{proof}

\section{Inexactness in Lower-Level solutions and Hypergradient}
\begin{lemma}[Inexactness in lower-level solutions]\label{lemma:inexact_lower_level}
Assume we have access to approximate lower-level solutions $\Approx{y}$, $\Approx{\lambda}$, $\Approx{y}_\delta$, $\Approx{\lambda}_\delta$ such that 
for some $\Approx{\delta} > 0$,
$\|y^* - \Approx{y}\| \leq \Approx{\delta}$, 
$\|\lambda^* - \Approx{\lambda}\| \leq \Approx{\delta}$,
$\|{y}^*_\delta - \Approx{y}_\delta\| \leq \Approx{\delta}$,
$\|{\lambda_\delta}^* - \Approx{\lambda}_\delta\| \leq \Approx{\delta}$. 
If $\Approx{\delta} = \delta^2$, then the approximate hypergradient $\Approx{\nabla} F(x) \define \nabla_x f(x,y^*(x)) + \Approx{v}_\delta(x,\Approx{y}_\delta, \Approx{\lambda}_\delta),$ has an error of 
\begin{align*}
    \|\Approx{\nabla} F(x) - \nabla F(x)\| \leq \bigO(N\delta)
\end{align*}
\end{lemma}
\begin{proof}
\begin{align*}
    v_\delta(x,y^*_{\delta},\lambda^*_\delta) = \frac{1}{\delta} \sum_{i=1}^N \nabla_x[g_i(x,y^*_{\delta,i}, y^*_{\neg i}) + (\lambda^*_{\delta,i})^\top \ei(x,y^*_i)] - \nabla_x[g_i(x,y^*_{i}, y^*_{\neg i}) + (\lambda^*_{i})^\top \ei(x,y^*_i)]
\end{align*}
\begin{align*}
    & \|\Approx{v}_\delta(x,\Approx{y}_\delta,\Approx{\lambda}_\delta)  - v_\delta(x,y^*_{\delta},\lambda^*_\delta)\| \\
    & \leq
    \frac{1}{\delta} \sum_{i=1}^N
    \|
    \nabla_x g_i(x,\Approx{y}_{\delta,i}, \Approx{y}_{\neg i})
    - \nabla_x g_i(x,y^*_{\delta,i}, y^*_{\neg i})
    + \nabla_x g_i(x,y^*_{i}, y^*_{\neg i})
    - \nabla_x g_i(x,\Approx{y}_{i}, \Approx{y}_{\neg i})
    \| \\
    & + \frac{1}{\delta} \sum_{i=1}^N
    \|
    \nabla_x[(\Approx{\lambda}_{\delta,i})^\top \ei(x,\Approx{y}_i)
    - (\lambda^*_{\delta,i})^\top \ei(x,y^*_i)
    + (\lambda^*_{i})^\top \ei(x,y^*_i)
    - (\Approx{\lambda}_{i})^\top \ei(x,\Approx{y}_i)]
    \| \\
    & \leq
    \frac{1}{\delta} \sum_{i=1}^N (3C \Approx{\delta} + 2\|A^{(i)}\|\Approx{\delta}) \\
    & = \bigO(N\delta)
\end{align*}

\begin{align*}
    & \|\Approx{\nabla}F(x) - \nabla F(x)\| \\
    & \leq 
    \|\nabla_xf(x,\Approx{y}) - \nabla_x f(x,y^*)\| 
    + \|\Approx{v}_\delta(x,\Approx{y}_\delta,\Approx{\lambda}_\delta)  - v_\delta(x,y^*_{\delta},\lambda^*_\delta)\|
    + \|v_\delta(x,y^*_{\delta},\lambda^*_\delta) - (\totalD{y^*(x)}{x})^\top \nabla_y f(x,y^*)\| \\
    & \leq C_f \Approx{\delta} + \bigO(N\delta) \\
    & = \bigO(N\delta)
\end{align*}
\end{proof}

\begin{lemma}[Inexact gradient descent on smooth function ]\label{lemma:inexact_GD}
Suppose the inexactness in the gradient of $C_F$-smooth function $F$ is given by $\|\Approx{\nabla} F(x) - \nabla F(x)\| \leq \alpha = \bigO(\epsilon)$. If learning rate $\eta=\frac{1}{C_F}$, then gradient descent using the inexact gradient $\Approx{\nabla} F(x)$ requires $\bigO(\epsilon^{-2})$ to reach $\mathcal{O}(\epsilon)$-stationarity. 
\end{lemma}
\begin{proof}
Suppose $(\forall x \in \inputspace)\,(\Approx{\nabla} F(x) = \nabla F(x) + \tau(x))$ where $\|\tau(x)\| \leq \bigO(\alpha)$. 
Denote $x^+ = x - \eta \Approx{\nabla} F(x)$,
\begin{align*}
    & F(x^+) \\
    & \leq 
    F(x) + \nabla F(x)^\top (x^+ - x) + \frac{C_F}{2}\|x^+ - x\|^2 \\
    & = 
    F(x) - \eta \nabla F(x)^\top \Approx{\nabla} F(x) + \frac{C_F}{2}\|\eta \Approx{\nabla} F(x)\|^2 \\
    & = 
    F(x) - \eta \|\nabla F(x)\|^2 + \frac{\eta^2 C_F}{2} \|\nabla F(x) + \tau(x)\|^2 - \eta \nabla F(x)^\top \tau(x)\\
    & = 
    F(x) - \eta \|\nabla F(x)\|^2 + \frac{\eta^2 C_F}{2}\|\nabla F(x)\|^2 + (\eta^2 C_F - \eta)\nabla F(x)^\top \tau(x) + \frac{\eta^2 C_F}{2}\|\tau(x)\|^2 \\
    & = 
    F(x) - \frac{2\eta - \eta^2 C_F}{2} \|\nabla F(x)\|^2 
    + (\eta^2 C_F - \eta)\nabla F(x)^\top \tau(x) + \frac{\eta^2 C_F}{2}\|\tau(x)\|^2 \\
    & =
    F(x) - \frac{\eta}{2}\|\nabla F(x)\|^2 + 0 + \frac{\eta}{2}\|\tau(x)\|^2 \\
    & \leq
    F(x) - \frac{\eta}{2}\|\nabla F(x)\|^2 + \frac{\eta}{2}\alpha^2
\end{align*}
Denote each step of inexact gradient descent by $x_k = x_{k-1} + \eta \Approx{\nabla}F(x_{k-1})$ for $k=0,1,2,\cdots,T-1$. Applying a telescope sum, we have
\begin{align*}
    F(x_T) - F(x_0) 
    &\leq -\frac{\eta}{2}\sum_{k=0}^{T-1}\|\nabla F(x_k)\|^2 + \frac{\eta}{2}\sum_{k=0}^{T-1}\alpha^2 \\
    \frac{1}{T}\sum_{k=0}^{T-1}\|\nabla F(x_k)\|^2 
    &\leq \frac{2(F(x_0)-F(x_T))}{\eta T} + \alpha^2 \\
    \Min_{k=0,1,\cdots,T-1}\|\nabla F(x_k)\|^2 
    &\leq \frac{2(F(x_0)-F(x^*))}{\eta T} + \alpha^2 \quad \text{  ($x^*$ is the minimizer of $F$)}\\
\end{align*}
When $T=\bigO(\epsilon^{-2})$ and $\alpha=\bigO(\epsilon)$, we have $\Min_{k=0,1,\cdots,T-1}\|\nabla F(x_k)\| \leq \bigO(\epsilon)$.
\end{proof}

\section{Solving Smooth and Strongly Monotone Game}
 Smooth and Strongly monotone game can be solved using projected gradient descent (PGD) (\cite{facchinei2003finite}). Let $\eta>0$ be the learning rate and let $\{y_k\}$ be the iterates that converges to the Nash equilibrium, the update rule is given by $y_{k+1} = \Pi_Y(y_k - \eta V(y_k))$. Here, we omit the upper-level variable for convenience. Following \cite{facchinei2003finite}, we can show a linear convergence guarantee. We extend the guarantee to the case when we only have access to approximate projection and approximate gradient operator $V$ in Lemma \ref{lemma:proj_GD_VI_convergence}.

\noindent
Finding a Nash equilibrium is equivalent to solving a variational inequality in smooth and strongly monotone game (Lemma \ref{lemma:local_nash_VI}). We consider solving for $y^*$ in the following variational inequality
\begin{align*}
    \forall y \in Y: \langle V(y^*), y-y^*\rangle \geq 0
\end{align*}
where $Y$ is a closed convex set and $V$ is a $\mu$-strongly monotone and $C$-Lipschitz vector-valued function. We also define Euclidean projection as $\Pi_Y(z) = \argmin_{y\in Y} \|y-z\|^2$. Lemma \ref{lemma:fixed_point} shows that $y^*$ satisfies the variational inequality above if and only $y^* = \Pi_Y(y^* - \eta V(y^*))$ for some $\eta > 0$. 

\begin{lemma}\label{lemma:proj_GD_VI_contraction}
Let $\gamma=1-2\mu \eta + \eta^2 C^2$
$\forall x,y \in Y$,
\begin{align*}
    \|\Pi_Y(x-\eta V(x)) - \Pi_Y(y-\eta V(y))\|^2
    & \leq \gamma\|y-x\|^2
\end{align*}
If $0 < \eta < \frac{2\mu}{C^2}$, then $0 < \gamma < 1$ and thus $\Pi_Y(y-\eta V(y))$ is a contraction mapping in $y$.
\end{lemma}
\begin{proof}
\begin{align*}
    & \|\Pi_Y(x-\eta V(x)) - \Pi_Y(y-\eta V(y))\|^2 \\
    & \leq
    \|x - \eta V(x) - y + \eta V(y)\|^2 \\
    & =
    \|x-y\|^2 - 2\eta\langle y-x, V(y) - V(x)\rangle + \eta^2\|V(y)-V(x)\|^2 \\
    & \leq
    \|x-y\|^2 - 2\mu\eta\|x-y\|^2 + \eta^2 C^2\|y-x\|^2 \\
    & = 
    (1-2\mu\eta+\eta^2C^2)\|y-x\|^2
\end{align*}
Solving $(1-2\mu\eta+\eta^2C^2)<1$, we get the $0 < \eta < \frac{2\mu}{C^2}$.
\end{proof}

\begin{lemma}[PGD with inexact projection and gradient]\label{lemma:proj_GD_VI_convergence}
Suppose for any vector $z$, we only have access to an approximate projection $\Approx{\Pi}_Y(z)$ such that $\|\Approx{\Pi}_Y(z) - \Pi_Y(z)\| \leq \epsilon_0$ for some $\epsilon_0$.
Let $\Approx{V}(z)$ be an approximate gradient operator such that $\|\Approx{V}(z) - V(z)\| \leq \delta_0$ for some $\delta_0 > 0$. 
Let $\gamma=1-2\mu \eta + \eta^2 C^2$ and $0 < \eta < \frac{2\mu}{C^2}$. Define an arbitrary initial iterate $y_0\in Y$ and apply the following update rule recursively: 
\begin{align*}
    y_{t+1} = \Approx{\Pi}_Y(y_t - \eta \Approx{V}(y_t))
\end{align*}
If $\epsilon_0 = \frac{(1-\sqrt{\gamma})\epsilon}{4}$ and $\delta_0 = \frac{(1-\sqrt{\gamma})\epsilon}{4 \eta}$, then $\|y_{t} - y^*\| \leq \epsilon$ when $T \geq \bigO(\log(\frac{2 \|y_0 - y^*\|}{\epsilon}))$
\end{lemma}
\begin{proof}
For every vector $z$, $\Approx{\Pi}_Y(z) = \Pi_Y(z) + \tau_0$ for some vector $\tau_0$ such that $\|\tau_0\| \leq \epsilon_0$, 
\begin{align*}
    & \|y_{t+1} - y^*\| \\
    &=
    \|\Approx{\Pi}_Y(y_t - \eta \Approx{V}(y_t)) - \Pi_Y(y_t - \eta V(y_t)) + \Pi_Y(y_t - \eta V(y_t)) - \Pi_Y(y^* - \eta V(y^*))\| \\
    & \leq
    \|\Pi_Y(y_t - \eta \Approx{V}(y_t)) - \Pi_Y(y_t - \eta V(y_t))\| + \|\Pi_Y(y_t - \eta V(y_t)) - \Pi_Y(y^* - \eta V(y^*))\| + \epsilon_0 \\
    & \leq
    \|- \eta \Approx{V}(y_t) + \eta V(y_t)\| + \|\Pi_Y(y_t - \eta V(y_t)) - \Pi_Y(y^* - \eta V(y^*))\| + \epsilon_0 \\
    & \leq
    \sqrt{\gamma} \|y_t - y^*\| + \eta \delta_0 + \epsilon_0 \\
\end{align*}
\begin{align*}
    & \|y_{t+1} - y^*\| \\
    & \leq
    (\sqrt{\gamma})^{t+1} \|y_0 - y^*\| + \sum_{k=0}^t (\sqrt{\gamma})^k\epsilon_0 + 
    \eta \sum_{k=0}^t (\sqrt{\gamma})^k\delta_0\\
    & = 
    (\sqrt{\gamma})^{t+1} \|y_0 - y^*\| + 
    \frac{1-\sqrt{\gamma}^{t+1}}{1-\sqrt{\gamma}} \epsilon_0 +
    \frac{1-\sqrt{\gamma}^{t+1}}{1-\sqrt{\gamma}} \eta \delta_0\\
    & \leq 
    (\sqrt{\gamma})^{t+1} \|y_0 - y^*\| + \frac{1}{1-\sqrt{\gamma}} (\epsilon_0 + \eta \delta_0)\\
    & =
    (\sqrt{\gamma})^{t+1} \|y_0 - y^*\| + \frac{\epsilon}{2}
\end{align*}
Solving for $(\sqrt{\gamma})^{t+1} \|y_0 - y^*\| \leq \frac{\epsilon}{2}$, we have
\begin{align*}
    (\sqrt{\gamma})^{t+1} \|y_0 - y^*\| & \leq \frac{\epsilon}{2} \\
    (t+1)\log(\frac{1}{\sqrt{\gamma}}) + \log(\frac{1}{\|y_0 - y^*\|}) &\geq \log(\frac{2}{\epsilon}) \\
    t+1 & \geq \log(\frac{2\|y_0 - y^*\|}{\epsilon})/\log(\frac{1}{\sqrt{\gamma}})
\end{align*}
\end{proof}

\noindent
We leverage Lemma \ref{lemma:natural_residual} to verify convergence of PGD in practice.
\begin{lemma}[Sufficient condition for convergence to $y^*$]\label{lemma:natural_residual}
Let $r(y) \define y - \Pi_Y(y-V(y))$ be the natural residual. Suppose $\|r(\Approx{y})\| \leq \beta$, then $\|\Approx{y} - y^*\| \leq \frac{C+1}{\mu} \beta$. 
\end{lemma}
\begin{proof}
Let $u \define \Approx{y} - V(\Approx{y})$, by first-order optimality of Euclidean projection, 
\begin{align*}
    \langle \Pi_Y(u) - u, z - \Pi_Y(u) \rangle 
    & \geq 0 \quad \text{, $\forall z \in Y$} \\
    \langle V(\Approx{y}) - r(\Approx{y}), y^* - \Approx{y} + r(\Approx{y}) \rangle & \geq 0 \quad \text{, substitute $z=y^*$}\\
\end{align*}
Since $y^*$ solves the variational inequality $\langle V(y^*), y-y^* \rangle \geq 0, \, \forall y \in Y$, we have
\begin{align*}
    \langle V(y^*), \Approx{y} - r(\Approx{y}) - y^*\rangle \geq 0
\end{align*}
Summing the inequalities toegther, we have
\begin{align*}
    \langle V(\Approx{y}) - V(y^*) - r(\Approx{y}),  y^* - \Approx{y} + r(\Approx{y})\rangle 
    & \geq 0 \\
\end{align*}
Rearranging terms and applying Lipschitzness of $V$,
\begin{align*}
    \langle V(\Approx{y}) - V(y^*),  \Approx{y} - y^*\rangle 
    & \leq 
    \langle V(\Approx{y}) - V(y^*), r(\Approx{y}) \rangle + \langle r(\Approx{y}), \Approx{y} - y^*\rangle \\
    & \leq 
    \| V(\Approx{y}) - V(y^*)\| \| r(\Approx{y}) \| + \| r(\Approx{y})\| \| \Approx{y} - y^*\| \\
    & \leq \beta (C+1) \|\Approx{y} - y^*\| 
\end{align*}
By strong monotonicity and rearranging terms, we get the desired bound,
\begin{align*}
\mu\|\Approx{y} - y^*\|^2
& \leq  \langle V(\Approx{y}) - V(y^*),  \Approx{y} - y^*\rangle \\
& \leq \beta (C+1) \|\Approx{y} - y^*\| 
\end{align*}
\end{proof}

\begin{lemma}[Inexact dual solution oracle]\label{lemma:dual_oracle}
Consider a $\mu$-strongly monotone and $C$-smooth game $\mathcal{G} \define (x, [N], \{Y_i\}_{i\in[N]}, \{g_i\}_{i\in[N]})$. Given an approximate lower-level primal solution $\Approx{y}$ such that $\|\Approx{y} - y^*\|\leq \Approx{\delta}$, we can obtain an approximate lower-level dual solution $\Approx{\lambda}$ such that $\|\Approx{\lambda} - \lambda^*\| \leq \Approx{\delta}$ using the same method in \cite{kornowski2024first} within $\bigO(1)$ oracle calls. 
\end{lemma}
\begin{proof}
The KKT system for the Nash equilibrium of $\mathcal{G}$ is given by
\begin{align*}
    \begin{bmatrix} V(x,y^*) + B^\top \lambda^* \\ \begin{bmatrix} e^{(1)}(x,y^*_1) \\ \vdots \\ e^{(N)}(x,y^*_N) \end{bmatrix} \end{bmatrix} 
    &= 
    0 \\
    \begin{bmatrix}  V(x,y^*) + B^\top \lambda^* \\ By^* \end{bmatrix} 
    &= 
    \begin{bmatrix} 0 \\ -Ax - b\end{bmatrix}
\end{align*}
Denote $c_x = -Ax-b$, $(y^*, \lambda^*)$ are given by
\begin{align*}
    \begin{bmatrix}I & B^\top \\ B & 0\end{bmatrix}
    \begin{bmatrix} y^* \\ \lambda^* \end{bmatrix}
    & = 
    \begin{bmatrix} -V(x,y^*)+y^* \\ c_x \end{bmatrix}
\end{align*}
and the approximate lower-level solution $(\Approx{y}, \Approx{\lambda})$ are given by
\begin{align*}
    \begin{bmatrix}I & B^\top \\ B & 0\end{bmatrix}
    \begin{bmatrix} \Approx{y} \\ \Approx{\lambda} \end{bmatrix}
    & = 
    \begin{bmatrix} B^\top\Approx{\lambda} + \Approx{y} \\ c_x \end{bmatrix}
\end{align*}
Given $\Approx{y}$, we can obtain $\Approx{\lambda}$ by solving
\begin{align*}
    B^\top \Approx{\lambda} = \argmin_{s \in \text{span}(B^\top)} \|-V(x,\Approx{y}) - s\|^2
\end{align*}
From the KKT system, observe that
\begin{align*}
    -V(x,y^*) = B^\top\lambda^* = \argmin_{s \in \text{span}(B^\top)} \|-V(x,y^*) - s\|^2
\end{align*}
\begin{align*}
    \|B^\top \Approx{\lambda} - (-V(x,y^*))\|
    & =
    \|B^\top \Approx{\lambda} - B^\top \lambda^*\| \\
    & \leq
    \|V(x,y^*)-V(x,\Approx{y})\| \\
    & \leq
    C \|\Approx{y} - y^*\| \\
    & \leq
    C \Approx{\delta}
\end{align*}
For some $\tau$ such that $\|\tau\| \leq (1+C)\Approx{\delta}$,
\begin{align*}
    \|\begin{bmatrix}\Approx{y} \\ \Approx{\lambda}\end{bmatrix} - \begin{bmatrix}y^* \\ \lambda^*\end{bmatrix}\|
    & = 
    \|\begin{bmatrix}I & B^\top \\ B & 0\end{bmatrix}^{-1}\| \|\begin{bmatrix}\tau \\ 0\end{bmatrix}\| \\
    & \leq
    \|\begin{bmatrix}I & B^\top \\ B & 0\end{bmatrix}^{-1}\|(1+C)\Approx{\delta}
\end{align*}

\end{proof}


\section{Unconstrained Stackelberg Game}
The gradient of the potential function is 
\begin{align*}
    \nabla_{y_i}\phi_\epsilon(x,y) 
    &= 
    \frac{1}{\epsilon} [\nabla_{y_i} g_i(x, y_i + \epsilon (y_i - y^*_i), y_{\neg i})(1+\epsilon) - \nabla_{y_i} g_i(x,y)] \\
    & + \frac{1}{\epsilon} \sum_{j \neq i} [\nabla_{y_i} g_j(x, y_j + \epsilon (y_j - y^*_j), y_{\neg j}) - \nabla_{y_i} g_j(x,y)] \\
    &=
    \frac{1}{\epsilon} \sum_{j=1}^N [\nabla_{y_i} g_j(x, y_j + \epsilon (y_j - y^*_j), y_{\neg j}) - \nabla_{y_i} g_j(x,y)] + g_i(x, y_i + \epsilon (y_i - y^*_i), y_{\neg i})
\end{align*}

\begin{lemma}[Local strong Monotonicity and smoothness of the perturbed game $\bar{\mathcal{G}}_\delta$]\label{lemma:asym_local_strong_monotone}
Consider a closed ball around $y^*$, $\bar{B}_{R}(y^*)=\{y \in \mathbb{R}^n : \|y-y^*\| \leq R\}$. For $0 < \delta < \frac{\mu}{2 C_f}$ and $0 < R < \frac{\mu - \delta C_f}{S(N+1)(2+\sqrt{N})}$ and $0 < \epsilon < 1$, the perturbed game $\bar{\mathcal{G}}_\delta$ is locally $\bar{\mu}$-strongly monotone:
\begin{align*}
    \langle \bar{V}(x,y) - \bar{V}(x,z), y - z \rangle \geq \bar{\mu} \|y - z\|^2, \, \forall y, z \in \bar{B}_R(y^*)
\end{align*}
For any $\epsilon>0$, $\delta > 0$ and $R > 0$, $\bar{\mathcal{G}}_\delta$ is locally $\bar{C}$-smooth:
\begin{align*}
    \| \bar{V(}x,y) - \bar{V}(x,z)\| \leq \bar{C} \|y-z\|, \, \forall y,z\in \bar{B}_R(y^*)
\end{align*}
The local strong monotonicity modulus is $\bar{\mu} \define \mu - \frac{S R N}{2} \epsilon - (S R N)(1 + \frac{\epsilon}{2} + \sqrt{N}) + \epsilon \mu - S R \epsilon (1 + \frac{\epsilon}{2} + \sqrt{N}) - \delta C_f > 0$. The local smoothness modulus is $\bar{C} \define \sqrt{N}\{[C \sqrt{N} + S R(\sqrt{N}(1+\frac{\epsilon}{2}) + N) + \delta C_f] + \epsilon [ C + S R (2+\frac{\epsilon}{2})]\}$ 
\end{lemma}
\begin{proof}
By the Fundamental Theorem of Calculus,
\begin{align*}
    \nabla_{y_i} g_j(x, y_j + \epsilon (y_j - y^*_j), y_{\neg j}) - \nabla_{y_i} g_j(x,y)
    &=
    \epsilon \int^1_0 \partialD{\nabla_{y_i} g_j(x, y_j + t\epsilon (y_j - y^*_j), y_{\neg j})}{y_j} (y_j - y^*_j) dt
\end{align*}
We can express the gradient of the perturbed cost function as
\begin{align*}
    \nabla_{y_i} \bar{g}_i(x,y)
    &=
    \int_0^1 \sum_{j=1}^N \partialD{\nabla_{y_i}}{y_j} g_j(x, y_j + t\epsilon(y_j - y^*_j), y_{\neg j}) (y_j - y^*_j) dt \\
    & + \epsilon \int_0^1 \partialD{\nabla_{y_i}}{y_i} g_i(x, y_i + t\epsilon(y_i - y^*_i), y_{\neg i}) (y_i - y^*_i) dt + \delta \nabla_{y_i} f(x,y) \\
    &= 
    \int_0^1 \sum_{j=1}^N J_{ji}(x, y_j + t\epsilon(y_j - y^*_j), y_{\neg j})^\top (y_j - y^*_j) dt \\
    & + \epsilon \int_0^1 J_{ii}(x, y_i + t\epsilon(y_i - y^*_i), y_{\neg i}) (y_i - y^*_i) dt + \delta \nabla_{y_i} f(x,y) 
\end{align*}
We can further expand it according to the identity $A(y)(y-y^*) - A(z)(z-y^*) = A(y)(y-z) + [A(y) - A(z)](z-y^*)$ for arbitrary matrices $A(y)$ and $A(z)$. 
\begin{align*}
    & \langle \nabla_{y_i} \bar{g}_i(x,y) - \nabla_{y_i} \bar{g}_i(x,z), y_i - z_i\rangle \\
    &=
    \underbrace{\int_0^1 \sum_{j=1}^N \langle \, y_i - z_i, \, J_{ji}(x, y_j + t\epsilon(y_j - y^*_j), y_{\neg j})^\top (y_j - z_j) \, \rangle dt}_{\text{(a)}} \\ 
    & + \int_0^1 \sum_{j=1}^N \langle \, y_i - z_i, \, [J_{ji}(x, y_j+t\epsilon(y_j - y^*_j), y_{\neg j})^\top - J_{ji}(x, z_j+t\epsilon(z_j - y^*_j), z_{\neg j})^\top](z_j - y^*_j) \, \rangle dt \\
    & + \epsilon \int_0^1 \langle \, y_i - z_i, \, J_{ii}(x, y_i + t\epsilon(y_i-y^*_i), y_{\neg i})(y_i - z_i) \, \rangle dt \\
    & + \epsilon \int_0^1 \langle \, y_i - z_i, \, [J_{ii}(x, y_i + t\epsilon(y_i-y^*_i), y_{\neg i}) - J_{ii}(x, z_i + t\epsilon(z_i-y^*_i), z_{\neg i})](z_i - y^*_i) \, \rangle dt \\
    & + \delta \langle \, y_i - z_i, \nabla_{y_i} f(x,y) - \nabla_{y_i} f(x,z) \, \rangle
\end{align*}
We can further decompose term (a) according to the identity
\begin{align*}
    J_{ji}(x, y_j+t\epsilon(y_j - y^*_j), y_{\neg j})^\top
    = 
    J_{ji}(x,y)^\top + (J_{ji}(x, y_j + t\epsilon(y_j - y^*_j), y_{\neg j})^\top - J_{ji}(x,y)^\top )
\end{align*}
For all $y, z \in B_R(y^*)$
\begin{align*}
    & \langle \, \bar{V}(x,y) - \bar{V}(x,z), \, y - z \, \rangle \\
    &= 
    \sum_{i=1}^N \langle \, \nabla_{y_i}\bar{g}_i(x,y) - \nabla_{y_i}\bar{g}_i(x,z), \, y_i - z_i \, \rangle \\
    & \geq
    \mu \|y-z\|^2 
    - \underbrace{\int_0^1 \sum_{i=1}^N \sum_{j=1}^N \|y_i - z_i\| \|y_j - z_j\| S \|t \epsilon (y_j - y^*_j)\| dt}_{\text{(a)}} \\
    & - \underbrace{\int_0^1 \sum_{i=1}^N \sum_{j=1}^N \|y_i - z_i\| S ((1+t\epsilon)\|y_j - z_j\| + \|y_{\neg j} - z_{\neg j}\| ) \|z_j - y^*_j\| dt}_{\text{(b)}} \\
    & + \epsilon \mu \|y-z\|^2 
    - \underbrace{\epsilon \int_0^1 \sum_{i=1}^N \|y_i - z_i\| S ((1+t\epsilon)\| y_i - z_i\| + \|y_{\neg i} - z_{\neg i}\|) \|z_i - y^*_i\| dt}_{\text{(c)}} \\
    & - \delta C_f \|y-z\|^2
\end{align*}
\begin{align*}
    \text{(a)}
    &= 
    S \epsilon \int_0^1 \sum_{i=1}^N t \|y_i - z_i\| (\sum_j \|y_j - z_j\| \|y_j - y^*_j\|) \\
    &\leq 
    S \sqrt{N} \epsilon \int_0^1 t \|y - z\| (\sum_j \|y_j - z_j\| \|y_j - y^*_j\|) \\
    &\leq 
    S \sqrt{N} \epsilon \int_0^1 t \|y - z\| (\sum_j \|y_j - z_j\| \|y - y^*\|) \\
    &\leq 
    \frac{S N }{2} \epsilon \|y - z\|^2 \|y - y^*\| \\
    & \leq 
    \frac{S R N}{2} \epsilon \|y - z\|^2 \\
\end{align*}
where the first and third inequalities are obtained using the identity $\frac{1}{\sqrt{N}}\|y\| = \sqrt{\frac{1}{N}\sum_i\|y_i\|^2} \geq \sum_i \frac{1}{N} \|y_i\|$ derived from Jensen's inequality. 
Using similar techniques, we can bound
\begin{align*}
    \text{(b)}
    & \leq
    S N \|y-z\|^2\|z-y^*\| + \frac{S N}{2}\epsilon \|y-z\|^2 \|z-y^*\| + SN\sqrt{N}\|y-z\|^2\|z-y^*\| \\
    & \leq 
    (S R N)(1 + \frac{\epsilon}{2} + \sqrt{N})\|y-z\|^2
\end{align*}
\begin{align*}
    \text{(c)}
    & \leq
    S \epsilon ((1+\frac{\epsilon}{2})\|y-z\|^2\|z-y^*\| + \sqrt{N}\|y-z\|^2\|z-y^*\|) \\
    & \leq 
    S R \epsilon (1 + \frac{\epsilon}{2} + \sqrt{N})\|y-z\|^2
\end{align*}
Putting these together, we have local strong monotonicity with modulus $\bar{\mu}$ for sufficiently small $\epsilon$, $\delta$ and $R$
\begin{align*}
    & \langle \, \bar{V}(x,y) - \bar{V}(x,z), \, y - z \, \rangle \\
    & \geq
    [\mu - \frac{S R N}{2} \epsilon - (S R N)(1 + \frac{\epsilon}{2} + \sqrt{N}) + \epsilon \mu - S R \epsilon (1 + \frac{\epsilon}{2} + \sqrt{N}) - \delta C_f] \|y-z\|^2
\end{align*}
Rearranging terms, we get $\bar{\mu}>0$ when
\begin{align*}
    R 
    < 
    \frac{(1+\epsilon)\mu - \delta C_f}{SN(1+\epsilon+\sqrt{N}) + S\epsilon(1+\frac{\epsilon}{2}+\sqrt{N})}
\end{align*}
\begin{align*}
\frac{(1+\epsilon)\mu - \delta C_f}{SN(1+\epsilon+\sqrt{N}) + S\epsilon(1+\frac{\epsilon}{2}+\sqrt{N})} 
& >
\frac{\mu - \delta C_f}{SN(1+\epsilon+\sqrt{N}) + S\epsilon(1+\epsilon+\sqrt{N})}  \\
& >
\frac{\mu - \delta C_f}{S(N+1)(2+\sqrt{N})} \, \text{(assume $0<\epsilon<1$)}
\end{align*}
$0 < \delta < \frac{\mu}{2 C_f}$ and $0 < R < \frac{\mu - \delta C_f}{S(N+1)(2+\sqrt{N})}$ is a sufficient condition for $\bar{\mu}>0$. \\


\noindent
To derive the smoothness bound,
\begin{align*}
    \|\nabla_{y_i}\bar{g}_i(x,y) - \nabla_{y_i}\bar{g}(x,z)\|
    &= \|
    \int_0^1 \sum_{j=1}^N J_{ji}(x, y_j + t\epsilon(y_j - y^*_j), y_{\neg j})^\top (y_j - z_j)dt \\ 
    & + \int_0^1 \sum_{j=1}^N [J_{ji}(x, y_j+t\epsilon(y_j - y^*_j), y_{\neg j})^\top - J_{ji}(x, z_j+t\epsilon(z_j - y^*_j), z_{\neg j})^\top](z_j - y^*_j) dt \\
    & + \epsilon \int_0^1 J_{ii}(x, y_i + t\epsilon(y_i-y^*_i), y_{\neg i})(y_i - z_i) dt \\
    & + \epsilon \int_0^1 [J_{ii}(x, y_i + t\epsilon(y_i-y^*_i), y_{\neg i}) - J_{ii}(x, z_i + t\epsilon(z_i-y^*_i), z_{\neg i})](z_i - y^*_i)dt \\
    & + \delta \nabla_{y_i} f(x,y) - \nabla_{y_i} f(x,z)
    \| \\
    & \leq
    C \sqrt{N} \|y-z\|
    + S \|z-y^*\|[\sqrt{N}(1+\frac{\epsilon}{2})\|y-z\| + N\|y-z\|] \\
    & + \epsilon C\|y_i-z_i\| + \epsilon S\|z-y^*\|[(1+\frac{\epsilon}{2})\|y_i-z_i\| + \|y-z\|] 
    + \delta C_f \|y-z\| \\
    & \leq
    [C \sqrt{N} + S R(\sqrt{N}(1+\frac{\epsilon}{2}) + N) + \delta C_f] \|y-z\| + \epsilon [ C + S R (2+\frac{\epsilon}{2})] \|y-z\|
\end{align*}
Hence,
\begin{align*}
    & \|\bar{V}(x,y) - \bar{V}(x,z)\| \\
    & \leq 
    \sqrt{N}\{[C \sqrt{N} + S R(\sqrt{N}(1+\frac{\epsilon}{2}) + N) + \delta C_f] + \epsilon [ C + S R (2+\frac{\epsilon}{2})]\} \|y-z\|
\end{align*}
\end{proof}

\noindent
We need to guarantee the existence of Nash equilibrium $\bar{y}^*_\delta$ in the perturbed game $\bar{\mathcal{G}}_\delta$, and guarantee that $\|\bar{y}^*_\delta - y^*\|$ shrinks with decreasing $\delta>0$.

\begin{lemma}[Existence of Nash equilibrium in the perturbed game]\label{lemma:asym_y_delta}
Let $R$ be the radius of the closed ball $\bar{B}_R(y^*)$ on which $\bar{V}$ is $\bar{\mu}$-strongly-monotone and $\bar{C}$-Lipschitz. For $0 < \delta \leq \min\{\frac{\bar{\mu} R}{2 L_f}, \frac{\mu}{2 C_f}\}$, the perturbed game $\bar{\mathcal{G}}_\delta$ has a locally unique Nash equilibrium $\bar{y}^*_\delta$ in $\bar{B}_R(y^*)$, and
\begin{align*}
    \|\bar{y}^*_\delta - y^*\| 
    \leq
    \frac{2 L_f\delta}{\bar{\mu}}
\end{align*}
\end{lemma}
\begin{proof}
Let $\gamma=1-2\bar{\mu} \eta + \eta^2 \bar{C}^2$. Denote $T(y) = y-\eta \bar{V}(x,y)$. From Lemm \ref{lemma:proj_GD_VI_contraction}, we have
$\forall y,z \in \bar{B}_R(y^*)$,
\begin{align*}
    \|T(y) - T(z)\|^2
    & \leq \gamma\|y-z\|^2
\end{align*}
If $0 < \eta < \frac{2\bar{\mu}}{\bar{C}^2}$, then $0 < \gamma < 1$ and thus $T(y)$ is a contraction mapping in $y$. We want to prove that $T$ maps from $\bar{B}_{R}(y^*)$ to $\bar{B}_{R}(y^*)$. Since $T$ is a contraction, there exists a unique fixed point $\bar{y}^*_\delta = T(\bar{y}^*_\delta)$ by Banach fixed-point theorem. The fixed point is the Nash equilibrium by Lemma \ref{lemma:fixed_point}.

\noindent
Let $y \in \bar{B}_R(y^*)$,
\begin{align*}
    \|T(y) - y^*\| 
    &\leq 
    \|T(y) - T(y^*)\| + \|T(y^*) - y^*\| \\
    & \leq
    \sqrt{\gamma} \|y-y^*\| + \eta \|\bar{V}(y^*)\| \\
    & \leq
    \sqrt{\gamma} R + \eta L_f\delta
\end{align*}

\noindent
Set $\eta = \frac{\bar{\mu}}{\bar{C}^2}$, we have $1-\sqrt{\gamma} = 1 - \sqrt{1 - \frac{\bar{\mu}^2}{\bar{C}^2}}$. We need 
\begin{align*}
    \delta \leq \frac{(1-\sqrt{\gamma})R}{\eta L_f} = \frac{(1-\sqrt{1 - \frac{\bar{\mu}^2}{\bar{C}^2}})R}{\frac{\bar{\mu}}{\bar{C}^2} L_f} 
\end{align*}
For any concave function $h: \mathbb{R} \rightarrow \mathbb{R}$, $\forall a, b \in \mathbb{R}$, $h(b) \leq h(a) + \nabla h(a)^\top(b-a)$. Set $h(b) = \sqrt{1-b}$, $a=0$, we have the identity $\sqrt{1-b} \leq 1-\frac{b}{2}$. Then
\begin{align*}
    \frac{(1-\sqrt{1 - \frac{\bar{\mu}^2}{\bar{C}^2}})R}{\frac{\bar{\mu}}{\bar{C}^2} L_f} 
    \geq 
    \frac{(1-(1 - \frac{\bar{\mu}^2}{2\bar{C}^2}))R}{\frac{\bar{\mu}}{\bar{C}^2} L_f}
    =
    \frac{\bar{\mu} R}{2 L_f}
\end{align*}
A sufficient condition for the fixed point $\bar{y}^*_\delta$ to exist is $\delta \leq \frac{\bar{\mu} R}{2 L_f}$.
We can bound the distance to $y^*$ in terms of $\delta$,
\begin{align*}
    \|\bar{y}^*_\delta - y^*\|
    & \leq
    \sqrt{\gamma} \|\bar{y}^*_\delta-y^*\| + \eta L_f\delta \\
    \|\bar{y}^*_\delta - y^*\| 
    & \leq
    \frac{\eta L_f\delta}{1-\sqrt{\gamma}} \\
    \|\bar{y}^*_\delta - y^*\| 
    & \leq
    \frac{2 L_f\delta}{\bar{\mu}} \\
\end{align*}
\end{proof}

\begin{lemma}\label{lemma:asym_grad_1}
\begin{align*}
    \|\frac{\bar{y}^*_\delta - y^*}{\delta} - (J(x,y^*)^{\top} + \epsilon D(x,y^*))^{-1} (-\nabla_y f(x,y^*))\| \leq \frac{\sqrt{N}}{\bar{\mu}^2}\mathcal{O}(\delta)
\end{align*}
\end{lemma}
\begin{proof}
Let $y=\bar{y}^*_\delta$, for all $i \in [N]$, we have 
\begin{align*}
    0 &=
    \nabla_{y_i} \bar{g}_i(x,y) \\
    &= 
    \int_0^1 \sum_{j=1}^N J_{ji}(x, y_j + t\epsilon(y_j - y^*_j), y_{\neg j})^\top (y_j - y^*_j) dt \\
    & + \epsilon \int_0^1 J_{ii}(x, y_i + t\epsilon(y_i - y^*_i), y_{\neg i}) (y_i - y^*_i) dt + \delta \nabla_{y_i} f(x,y) \\ 
    &=
    \underbrace{\int_0^1 \sum_{j=1}^N [J_{ji}(x, y_j + t\epsilon(y_j - y^*_j), y_{\neg j})^\top - J_{ji}(x, y^*_j, y_{\neg j})^\top] (y_j - y^*_j) dt}_{\text{(a)}} \\
    & + \underbrace{\epsilon \int_0^1 [J_{ii}(x, y_i + t\epsilon(y_i - y^*_i), y_{\neg i}) - J_{ii}(x, y^*_i, y_{\neg i})] (y_i - y^*_i) dt}_{\text{(b)}}\\
    & + \underbrace{\delta(\nabla_{y_i} f(x,y) - \nabla_{y_i} f(x,y^*))}_{\text{(c)}} \\
    & + \underbrace{\sum_{j=1}^N [J_{ji}(x, y^*_j, y_{\neg j})^\top - J_{ji}(x, y^*)^\top] (y_j - y^*_j)}_{\text{(d)}} \\
    & + \underbrace{[J_{ii}(x, y^*_i, y_{\neg i})^\top - J_{ii}(x, y^*)^\top] (y_i - y^*_i)}_{\text{(e)}} \\
    & + \sum_{j=1}^N J_{ji}(x,y^*)^\top(y_j - y^*_j) + \epsilon J_{ii}(x, y^*)(y_i - y^*_i) + \delta \nabla_{y_i}f(x,y^*)
\end{align*}
\begin{align*}
    & \|a\| \leq S(1+\frac{\epsilon}{2})\|y-y^*\|^2 \\
    & \|b\| \leq S(1+\frac{\epsilon}{2})\|y-y^*\|^2 \\
    & \|c\| \leq \delta L_f \|y-y^*\| \\
    & \|d\| \leq S \|y-y^*\|^2 \\
    & \|e\| \leq \epsilon S \|y-y^*\|^2
\end{align*}
Combining the equations for $i=1,2,\cdots,N$,
\begin{align*}
    J(x,y^*)^\top (y-y^*) + \epsilon D(x,y^*) (y-y^*) + \delta \nabla_y f(x,y^*) + \text{extra} = 0
\end{align*}
where 
$\|\text{extra}\| \leq \sqrt{N}S(1+\frac{\epsilon}{2})\|y-y^*\|^2 
+ S(1+\frac{\epsilon}{2})\|y-y^*\|^2 
+ \delta L_f \|y-y^*\|
+ S \|y-y^*\|^2
+ \epsilon S \|y-y^*\|^2$. 
Solving the equation, we have
\begin{align*}
    \frac{\bar{y}^*_\delta - y^*}{\delta}
    =
    (J(x,y^*)^\top + \epsilon D(x,y^*))^{-1} (-\nabla_y f(x,y^*) -\frac{\text{extra}}{\delta} )
\end{align*}
Using Lemma \ref{lemma:asym_y_delta}, we have
\begin{align*}
    \|\frac{\bar{y}^*_\delta - y^*}{\delta} - (J(x,y^*)^{\top} + \epsilon D(x,y^*))^{-1} (-\nabla_y f(x,y^*))\|
    & \leq
    \|(J(x,y^*)^\top + \epsilon D(x,y^*))^{-1}\| \|\frac{\text{extra}}{\delta}\| \\
    & \leq \frac{1}{\mu(1+\epsilon)} [\mathcal{O}(\frac{\sqrt{N}\delta}{\bar{\mu}^2}) + \mathcal{O}(\frac{\sqrt{N}\delta \epsilon}{\bar{\mu}^2}))] \\
    & \leq
    \mathcal{O}(\frac{\sqrt{N}\delta}{\bar{\mu}^2}) \quad \text{(we ignore $\mathcal{O}(\delta \epsilon)$)}
\end{align*}
\end{proof}

\begin{lemma}\label{lemma:asym_grad_2}
\begin{align*}
    \|(J(x,y^*)^{\top} + \epsilon D(x,y^*))^{-1} (-\nabla_y f(x,y^*)) - J(x,y^*)^{-\top}(-\nabla_y f(x,y^*))\|
    \leq 
    \mathcal{O}(\epsilon)
\end{align*}
\end{lemma}
\begin{proof}
\begin{align*}
     & \|(J(x,y^*)^{\top} + \epsilon D(x,y^*))^{-1} (-\nabla_y f(x,y^*)) - J(x,y^*)^{-\top}(-\nabla_y f(x,y^*))\| \\
     &= 
     \|(J(x,y^*)^{\top} + \epsilon D(x,y^*))^{-1} [J(x,y^*)^{\top}-(J(x,y^*)^{\top} + \epsilon D(x,y^*))]J(x,y^*)^{-\top}(-\nabla_y f(x,y^*))\| \\
     &\leq
     \frac{\epsilon}{\mu^2(1+\epsilon)} L_f \\
     &=
     \mathcal{O}(\epsilon)
\end{align*}
The first equality uses the matrix identity $A^{-1} - B^{-1} = A^{-1} (B - A) B^{-1}$ for any invertible matrices $A$ and $B$ of the same dimensions.
\end{proof}

\subsection{Proof of Theorem \ref{theorem:asym_grad_approx}}
\begin{proof}
\begin{align*}
    v_\delta(x,\bar{y}^*_\delta) 
    &= 
    \frac{1}{\delta}\sum_{i=1}^N [\nabla_x g_i(x,\bar{y}^*_{\delta,i},y^*_{\neg i}) - \nabla_xg_i(x,y^*)] \\
    &=
    \frac{1}{\delta} \sum_{i=1}^N [\partialD{\nabla_x g_i(x,y^*)}{y_i} (\bar{y}^*_{\delta,i} - y^*_i) + \text{H.O.T.}] \\
    &=
    \partialD{V(x,y^*)}{x}^\top \frac{\bar{y}^*_{\delta} - y^*}{\delta} + N\frac{\text{H.O.T.}}{\delta}]
\end{align*}
where $\|\text{H.O.T.}\| = \mathcal{O}(\|\bar{y}^*_{\delta} - y^*\|^2) \leq \mathcal{O}(\frac{\delta^2}{\bar{\mu}^2})$ by Lemma \ref{lemma:asym_y_delta}. Using Lemma \ref{lemma:asym_grad_1} and \ref{lemma:asym_grad_2}, we have
\begin{align*}
    & \|v_\delta(x,\bar{y}^*_\delta) - \partialD{V(x,y^*)}{x}^\top \partialD{V(x,y^*)}{y}^{-\top} (-\nabla_y f(x,y^*))\| \\
    & \leq 
    \|v_\delta(x,\bar{y}^*_\delta) - \partialD{V(x,y^*)}{x}^\top (J(x,y^*)^{\top} + \epsilon D(x,y^*))^{-1} (-\nabla_y f(x,y^*))\| \\
    & + \|\partialD{V(x,y^*)}{x}^\top (J(x,y^*)^{\top} + \epsilon D(x,y^*))^{-1} (-\nabla_y f(x,y^*)) - \partialD{V(x,y^*)}{x}^\top \partialD{V(x,y^*)}{y}^{-\top} (-\nabla_y f(x,y^*))\| \\
    & \leq
    C(\mathcal{O}(\frac{\sqrt{N}\delta}{\bar{\mu}^2}) + \mathcal{O}(\frac{N \delta}{\bar{\mu}^2})
    + \mathcal{O}(\epsilon)) \\
    & = 
    \mathcal{O}(\frac{N \delta}{\bar{\mu}^2} + \epsilon)
\end{align*}
\end{proof}

\section{Stackelberg Game with Linear Equality Constraints}\label{sec:asym_constr}
By construction, $\bar{\mathcal{G}}_\delta$ has the same local strong monotonicity modulus $\bar{\mu}$ and local smoothness modulus $\bar{C}$ on the ball $\bar{B}_R(y^*)$ as the unconstrained case. We are making the following assumption. 
\begin{assumption}[Constraint qualification (LICQ)]\label{assum:LICQ} 
Assume the rows of the Jacobian of equality constraints and active inequality constraints are linearly independent for each player.
\end{assumption}

\begin{lemma}[Existence of Nash equilibrium in the perturbed game]\label{lemma:asym_y_delta_constr}
Let $R$ be the radius of the closed ball $\bar{B}_R(y^*)$ on which $\bar{V}$ is $\bar{\mu}$-strongly-monotone and $\bar{C}$-Lipschitz. For $0 < \delta \leq \frac{\bar{\mu} R}{2 L_f}$, the perturbed game $\bar{\mathcal{G}}_\delta$ has a locally unique Nash equilibrium $\bar{y}^*_\delta$ in $\bar{B}_R(y^*)$, and
\begin{align*}
    \|\bar{y}^*_\delta - y^*\| 
    \leq
    \frac{2 L_f\delta}{\bar{\mu}}
\end{align*}
\begin{align*}
    \|\bar{\lambda}^*_\delta - \lambda^*\| 
    \leq 
    \|B^{\dagger}\|L_f(1 + 2\frac{\bar{C}}{\bar{\mu}}) \delta
\end{align*}
\end{lemma}
\begin{proof}
Let $\gamma=1-2\bar{\mu} \eta + \eta^2 \bar{C}^2$. The feasible set is $Y = \{(y_1, \cdots,y_N) \in \mathbb{R}^n : e^{(i)}(x,y_i)=0 \, , \forall i \in [N]\}$. Denote $T(y) = \Pi_Y(y-\eta \bar{V}(x,y))$. From Lemma \ref{lemma:proj_GD_VI_contraction}, we have
$\forall y,z \in \bar{B}_R(y^*)$,
\begin{align*}
    \|T(y) - T(z)\|^2
    & \leq \gamma\|y-z\|^2
\end{align*}
If $0 < \eta < \frac{2\bar{\mu}}{\bar{C}^2}$, then $0 < \gamma < 1$ and thus $T(y)$ is a contraction mapping in $y$. We want to prove that $T$ maps from $\bar{B}_{R}(y^*)$ to $\bar{B}_{R}(y^*)$. Since $T$ is a contraction, there exists a unique fixed point $\bar{y}^*_\delta = T(\bar{y}^*_\delta)$ by Banach fixed-point theorem. The fixed point is the Nash equilibrium by Lemma \ref{lemma:fixed_point}.

\noindent
Denote $H(y) = \Pi_Y(y - \eta V(x,y))$. We know that $y^* = H(y^*)$. KKT condition says that $V(x,y^*) + B^\top \lambda^* = 0$. Let $y \in \bar{B}_R(y^*)$,
\begin{align*}
    \|T(y) - y^*\| 
    & \leq
    \|T(y) - T(y^*)\| + \|T(y^*) - H(y^*)\| + \|H(y^*) - y^*\| \\
    & \leq
    \sqrt{\gamma} \|y-y^*\| + \eta \|- V(x,y^*) - B^\top \lambda^* - \delta \nabla_y f(x,y^*)\| + 0 \\
    & \leq 
    \sqrt{\gamma} \|y-y^*\| + \eta L_f \delta \\
    & \leq
    \sqrt{\gamma} R + \eta \delta L_f
\end{align*}
Thus, we can show the existence of the Nash equilibrium $\bar{y}^*_\delta$ and the desired bound on $\|\bar{y}^*_\delta - y^*\|$ using the same steps in Lemma \ref{lemma:asym_y_delta}.

\noindent
By KKT condition,
\begin{align*}
    & \bar{V}(x,\bar{y}^*_\delta) + B^\top \bar{\lambda}^*_\delta = 0 \\
    & V(x,y^*) + B^\top \lambda^* = 0 
\end{align*}
subtracting these two equations and rearranging terms,
\begin{align*}
    \|\bar{\lambda}^*_\delta - \lambda^*\| 
    & \leq 
    \|B^{\dagger}\| \|\bar{V}(x,\bar{y}^*_\delta) - \bar{V}(x,y^*) + \bar{V}(x,y^*) - V(x,y^*)\| \\
    & \leq 
    \|B^{\dagger}\|(\bar{C}\|\bar{y}^*_\delta - y^*\| + \delta L_f) \\
    & \leq
    \|B^{\dagger}\|L_f(1 + 2\frac{\bar{C}}{\bar{\mu}}) \delta \\
\end{align*}
\end{proof}

\noindent
Let $K(x,y) = \begin{bmatrix}\partialD{V(x,y)}{y} & B^\top \\ B & 0\end{bmatrix}$. Let $K_\epsilon(x,y) = \begin{bmatrix}\partialD{V(x,y)}{y}+\epsilon D(x,y) & B^\top \\ B & 0\end{bmatrix}$.

\begin{lemma}\label{lemma:KKT_matrix_bound}
Given a matrix $M = \begin{bmatrix}J & B^\top \\ B & 0\end{bmatrix}$ where $J$ is invertible. Then,
\begin{align*}
    \|M^{-1}\| \leq (1 + \|J^{-1}\|\|B\|)^2\max(\|J^{-1}\|, \|(B J^{-1} B^\top)^{-1}\|)
\end{align*}
\end{lemma}
\begin{proof}
Let $S = - B J^{-1} B^\top$. 
Let $L = \begin{bmatrix}I & 0 \\ -B J^{-1} & I\end{bmatrix}$ 
and $U = \begin{bmatrix}I & -J^{-1}B^\top \\ 0 & I\end{bmatrix}$. Using $L$ to eliminate rows of $M$ and $U$ to eliminates columns of $M$, we have
\begin{align*} 
    M = L^{-1}
    \begin{bmatrix} 
    J & 0 \\ 0 & S
    \end{bmatrix}
    U^{-1}
\end{align*}
Taking inverse,
\begin{align*} 
    M^{-1} 
    & = 
    U 
    \begin{bmatrix} 
    J^{-1} & 0 \\ 0 & S^{-1}
    \end{bmatrix}
    L \\
    & = 
    \begin{bmatrix}I & -J^{-1}B^\top \\ 0 & I\end{bmatrix} 
    \begin{bmatrix} 
    J^{-1} & 0 \\ 0 & S^{-1}
    \end{bmatrix}
    \begin{bmatrix}I & 0 \\ -B J^{-1} & I\end{bmatrix}
\end{align*}
Then,
\begin{align*}
    \|M^{-1}\| \leq (1 + \|J^{-1}\|\|B\|)^2\max(\|J^{-1}\|, \|(B J^{-1} B^\top)^{-1}\|)
\end{align*}
\end{proof}

\begin{lemma}\label{lemma:asym_grad_1_constr}
\begin{align*}
    \|
    \begin{bmatrix} 
    \frac{\bar{y}^*_\delta - y^*}{\delta} \\
    \frac{\bar{\lambda}^*_\delta - \lambda^*}{\delta}
    \end{bmatrix}
    - 
    (K_\epsilon(x,y^*))^{-\top} 
    (-\begin{bmatrix} 
    \nabla_y f(x,y^*) \\ 0 
    \end{bmatrix}) 
    \| \leq 
    \mathcal{O}(\frac{\sqrt{N}\delta}{\bar{\mu}^2})
\end{align*}
\end{lemma}
\begin{proof}
Let $y=\bar{y}^*_\delta$ and $\lambda = \bar{\lambda}^*_\delta$. for all $i \in [N]$, we have 
\begin{align*}
    0 &=
    \nabla_{y_i} \bar{g}_i(x,y) + (B^\top \lambda)_i\\
    &= 
    \int_0^1 \sum_{j=1}^N J_{ji}(x, y_j + t\epsilon(y_j - y^*_j), y_{\neg j})^\top (y_j - y^*_j) dt \\
    & + \epsilon \int_0^1 J_{ii}(x, y_i + t\epsilon(y_i - y^*_i), y_{\neg i}) (y_i - y^*_i) dt + \delta \nabla_{y_i} f(x,y) + (B^\top (\lambda - \lambda^*))_i\\ 
    &=
    \underbrace{\int_0^1 \sum_{j=1}^N [J_{ji}(x, y_j + t\epsilon(y_j - y^*_j), y_{\neg j})^\top - J_{ji}(x, y^*_j, y_{\neg j})^\top] (y_j - y^*_j) dt}_{\text{(a)}} \\
    & + \underbrace{\epsilon \int_0^1 [J_{ii}(x, y_i + t\epsilon(y_i - y^*_i), y_{\neg i}) - J_{ii}(x, y^*_i, y_{\neg i})] (y_i - y^*_i) dt}_{\text{(b)}}\\
    & + \underbrace{\delta(\nabla_{y_i} f(x,y) - \nabla_{y_i} f(x,y^*))}_{\text{(c)}} \\
    & + \underbrace{\sum_{j=1}^N [J_{ji}(x, y^*_j, y_{\neg j})^\top - J_{ji}(x, y^*)^\top] (y_j - y^*_j)}_{\text{(d)}} \\
    & + \underbrace{[J_{ii}(x, y^*_i, y_{\neg i})^\top - J_{ii}(x, y^*)^\top] (y_i - y^*_i)}_{\text{(e)}} \\
    & + \sum_{j=1}^N J_{ji}(x,y^*)^\top(y_j - y^*_j) + \epsilon J_{ii}(x, y^*)(y_i - y^*_i) + \delta \nabla_{y_i}f(x,y^*) + + (B^\top (\lambda - \lambda^*))_i
\end{align*}
\begin{align*}
    & \|a\| \leq S(1+\frac{\epsilon}{2})\|y-y^*\|^2 \\
    & \|b\| \leq S(1+\frac{\epsilon}{2})\|y-y^*\|^2 \\
    & \|c\| \leq \delta L_f \|y-y^*\| \\
    & \|d\| \leq S \|y-y^*\|^2 \\
    & \|e\| \leq \epsilon S \|y-y^*\|^2
\end{align*}
Combining the equations for $i=1,2,\cdots,N$,
\begin{align*}
    J(x,y^*)^\top (y-y^*) + \epsilon D(x,y^*) (y-y^*) + \delta \nabla_y f(x,y^*) + B^\top (\lambda - \lambda^*) + \text{extra} = 0
\end{align*}
where 
$\|\text{extra}\| \leq \sqrt{N}(S(1+\frac{\epsilon}{2})\|y-y^*\|^2 
+ S(1+\frac{\epsilon}{2})\|y-y^*\|^2) 
+ \delta L_f \|y-y^*\|
+ S \|y-y^*\|^2
+ \epsilon S \|y-y^*\|^2$
Solving the equation with the constraint that $B(y-y^*)=0$, we have
\begin{align*}
    \begin{bmatrix} 
    \frac{\bar{y}^*_\delta - y^*}{\delta} \\
    \frac{\bar{\lambda}^*_\delta - \lambda^*}{\delta}
    \end{bmatrix}
    =
    (K_\epsilon(x,y^*))^{-\top} 
    (-\begin{bmatrix} 
    \nabla_y f(x,y^*) + \frac{\text{extra}}{\delta} \\ 0 
    \end{bmatrix})
\end{align*}
Using Lemma \ref{lemma:asym_y_delta_constr}, we have
\begin{align*}
    & \|
    \begin{bmatrix} 
    \frac{\bar{y}^*_\delta - y^*}{\delta} \\
    \frac{\bar{\lambda}^*_\delta - \lambda^*}{\delta}
    \end{bmatrix}
    - 
    (K_\epsilon(x,y^*))^{-\top} 
    (-\begin{bmatrix} 
    \nabla_y f(x,y^*) \\ 0 
    \end{bmatrix}) 
    \| \\
    & \leq
    \|K_\epsilon(x,y^*)^{-\top}\| \|\frac{\text{extra}}{\delta}\| \\
    & \leq 
    (1 + \frac{1}{\mu(1+\epsilon)}\|B\|)^2\max(\frac{1}{\mu(1+\epsilon)}, \frac{C^2}{\mu(1 + \epsilon)\sigma_{\min}(B)^2}) [\mathcal{O}(\frac{\sqrt{N}\delta}{\bar{\mu}^2}) + \mathcal{O}(\frac{\sqrt{N}\delta \epsilon}{\bar{\mu}^2})] \, \text{(Lemma \ref{lemma:KKT_matrix_bound})}\\
    & \leq
    \mathcal{O}(\frac{\sqrt{N}\delta}{\bar{\mu}^2}) \quad \text{(we ignore $\mathcal{O}(\delta \epsilon)$)}
\end{align*}
Denote $G = J(x,y^*) + \epsilon D(x,y^*)$. The second inequality uses Lemma \ref{lemma:KKT_matrix_bound} and the fact that $d^\top G d \geq (\mu + \epsilon \mu)\|d\|^2$. To bound $\|(B G^{-\top} B^\top)^{-1}\|)$, we use Lemma \ref{lemma:singular_value_3} and Lemma \ref{lemma:singular_value_1}:
\begin{align*}
    \|(B G^{-\top} B^\top)^{-1}\|
    & \leq
    \frac{1}{\sigma_{\min}(B G^{-\top} B^\top)} \\
    & \leq
    \frac{C^2}{\mu(1 + \epsilon)\sigma_{\min}(B)^2}
\end{align*}
\end{proof}

\begin{lemma}\label{lemma:asym_grad_2_constr}
\begin{align*}
    \|K_\epsilon(x,y^*)^{-\top} (-\nabla_y f(x,y^*)) - K(x,y^*)^{-\top}(-\nabla_y f(x,y^*))\|
    \leq 
    \mathcal{O}(\epsilon)
\end{align*}
\end{lemma}
\begin{proof}
\begin{align*}
     & \|K_\epsilon(x,y^*)^{-\top} \begin{bmatrix} -\nabla_y f(x,y^*) \\ 0 \end{bmatrix} 
     - 
     K(x,y^*)^{-\top}\begin{bmatrix} -\nabla_y f(x,y^*) \\ 0 \end{bmatrix}\| \\
     &= 
     \|K_\epsilon(x,y^*)^{-\top} [K(x,y^*)^{\top}-K_\epsilon(x,y^*)^\top]K(x,y^*)^{-\top}(-\nabla_y f(x,y^*))\| \\
     &\leq
     \mathcal{O}(L_f \epsilon) \\
     &=
     \mathcal{O}(\epsilon)
\end{align*}
The first equality uses the matrix identity $A^{-1} - B^{-1} = A^{-1} (B - A) B^{-1}$ for any invertible matrices $A$ and $B$ of the same dimensions. The second inequality uses lemma \ref{lemma:KKT_matrix_bound} which says $\|K_\epsilon(x,y^*)^{-1}\| = \mathcal{O}(1)$ and $\|K(x,y^*)^{-1}\| = \mathcal{O}(1)$ as $\epsilon$ decreases.
\end{proof}

\subsection{Proof of Theorem \ref{theorem:asym_grad_approx_constr}}
\begin{proof}
\begin{align*}
    v_\delta(x,\bar{y}^*_\delta, \bar{\lambda}^*_\delta) 
    &= 
    \frac{1}{\delta}\sum_{i=1}^N [\nabla_x g_i(x,\bar{y}^*_{\delta,i},y^*_{\neg i}) - \nabla_xg_i(x,y^*) 
    + \partialD{e^{(i)}(x,y^*_i)}{x}^\top (\bar{\lambda}^*_\delta - \lambda^*)] \\
    &=
    \frac{1}{\delta} \sum_{i=1}^N [\partialD{\nabla_x g_i(x,y^*)}{y_i} (\bar{y}^*_{\delta,i} - y^*_i) + \text{H.O.T.} 
    + \partialD{e^{(i)}(x,y^*_i)}{x}^\top (\bar{\lambda}^*_\delta - \lambda^*)] \\
    &=
    \partialD{V(x,y^*)}{x}^\top \frac{\bar{y}^*_{\delta} - y^*}{\delta} + N\frac{\text{H.O.T.}}{\delta} 
    + A^\top \frac{\bar{\lambda}^*_\delta - \lambda^*}{\delta}]
\end{align*}
where $\|\text{H.O.T.}\| = \mathcal{O}(\|\bar{y}^*_{\delta} - y^*\|^2) \leq \mathcal{O}(\frac{\delta^2}{\bar{\mu}^2})$ by Lemma \ref{lemma:asym_y_delta_constr}. 
Let $v^*_\epsilon \define K_\epsilon(x,y^*)^{-\top} \begin{bmatrix} -\nabla_y f(x,y^*) \\ 0 \end{bmatrix}$ and $ v^* \define K(x,y^*)^{-\top}\begin{bmatrix} -\nabla_y f(x,y^*) \\ 0 \end{bmatrix}$. 
Using Lemma \ref{lemma:asym_grad_1_constr} and \ref{lemma:asym_grad_2_constr}, we have
\begin{align*}
    & \|
    v_\delta(x,\bar{y}^*_\delta, \bar{\lambda}^*_\delta) 
    - 
    \begin{bmatrix}\partialD{V(x,y^*)}{x}^\top & A^\top\end{bmatrix}
    \begin{bmatrix}\partialD{V(x,y^*)}{y} & B^\top \\ B & 0\end{bmatrix}^{-\top} 
    \begin{bmatrix} -\nabla_yf(x,y^*) \\ 0\end{bmatrix}
    \| \\
    & \leq 
    \|v_\delta(x,\bar{y}^*_\delta, \bar{\lambda}^*_\delta) - \begin{bmatrix}\partialD{V(x,y^*)}{x}^\top & A^\top\end{bmatrix} v^*_\epsilon\| \\
    & + \|\begin{bmatrix}\partialD{V(x,y^*)}{x}^\top & A^\top\end{bmatrix}( v^*_\epsilon -  v^*)\| \\
    & \leq
    (C + \|A\|)(\mathcal{O}(\frac{\sqrt{N}\delta}{\bar{\mu}^2}) + \mathcal{O}(\frac{N \delta}{\bar{\mu}^2}) + \mathcal{O}(\epsilon)) \\
    & = 
    \mathcal{O}(\frac{N \delta}{\bar{\mu}^2} + \epsilon)
\end{align*}
\end{proof}

\begin{lemma}\label{lemma:inexact_potential}
Let $\Approx{y}$ and $\Approx{\lambda}$ be approximate solutions of the original lower-level game such that $\|y^*-\Approx{y}\| \leq \Approx{\delta}$ and $\|\lambda^*-\Approx{\lambda}\| \leq \Approx{\delta}$. Define $\Approx{\phi}_\epsilon(x,y) = \frac{1}{\epsilon} \sum_{j=1}^N g_j(x,y_j + \epsilon(y_j - \Approx{y}_j, y_{\neg_j}) - g_j(x,y)$, $\Approx{{g}}_i(x,y) = \Approx{\phi}_\epsilon(x,y) - g_i(x,y) + \delta f(x,y) - (\Approx{\lambda})^\top B y$ and $\Approx{V}(x,y) = (\nabla_{y_1}\Approx{g}_1(x,y), \cdots, \nabla_{y_N}\Approx{g}_N(x,y))$. Then
\begin{align*}
    \|\bar{V}(x,y) - \Approx{V}(x,y)\| \leq (C N + \|B\|) \Approx{\delta}
\end{align*}
\end{lemma}
\begin{proof}
\begin{align*}
    \|\nabla_{y_i}\Approx{\phi}_\epsilon(x,y) - \nabla_{y_i} \phi_\epsilon(x,y)\|
    \leq
    \frac{1}{\epsilon} \sum_j C \|\epsilon (y^*_j - \Approx{y}_j)\| 
    = 
    C \sqrt{N} \|y^* - \Approx{y}\| \\
\end{align*}
\begin{align*}
    \|V(x,y) - \Approx{V}(x,y)\|
    \leq  
    C N \|y^* - \Approx{y}\| + \|B\|\|\lambda^* - \Approx{\lambda}\|
    \leq
    (C N + \|B\|)\Approx{\delta}
\end{align*}
\end{proof}

\begin{restatable}{theorem}{asymGradApproxConstr}\label{theorem:asym_grad_approx_constr}
    For $0 < \delta < \frac{\mu}{2 C_f}$ and $0 < R < \frac{\mu - \delta C_f}{S(N+1)(2+\sqrt{N})}$ and $0 < \epsilon < 1$, $\bar{V}$ is locally $\bar{\mu}$-strongly monotone and $\bar{C}$-smooth on $\bar{B}_R(y^*) = \{y \in \mathbb{R}^n : \|y-y^*\| \leq R\}$ for some $\bar{\mu}, \, \bar{C}>0$. Furthermore, if $\delta < \min(\frac{\mu}{2 C_f}, \frac{\bar{\mu} R}{2 L_f})$, then we have a fully first-order hypergradient approximation with non-asymptotic rate
\begin{align*}
    \|v_\delta(x,\bar{y}^*_\delta, \bar{\lambda}^*_\delta) 
    - 
    \begin{bmatrix}\partialD{V(x,y^*)}{x}^\top & A^\top\end{bmatrix}
    \begin{bmatrix}\partialD{V(x,y^*)}{y} & B^\top \\ B & 0\end{bmatrix}^{-\top} 
    \begin{bmatrix} -\nabla_yf(x,y^*) \\ 0\end{bmatrix}
    \|
    \leq 
    \mathcal{O}(\frac{N \delta}{\bar{\mu}^2} + \epsilon)
\end{align*}
\end{restatable}
\begin{proof}[Proof Sketch]
Lemma \ref{lemma:asym_local_strong_monotone} proves local strong monoticity and smoothness. For $\delta < \min(\frac{\mu}{2 C_f}, \frac{\bar{\mu} R}{2 L_f})$, Lemma \ref{lemma:asym_y_delta_constr} shows that $T(y) = \Pi_Y(y-\eta \bar{V}(x,y))$ is a contraction on $\bar{B}_R(y^*) \cap Y$. It proves the existence of Nash equilibrium $\bar{y}^*_\delta$ and $\|\bar{y}^*_\delta - y^*\| \leq \frac{2 L_f\delta}{\bar{\mu}}$. 
Let $v^*_\epsilon \define K_\epsilon(x,y^*)^{-\top} \begin{bmatrix} -\nabla_y f(x,y^*) \\ 0 \end{bmatrix}$ 
and 
$ v^* \define K(x,y^*)^{-\top}\begin{bmatrix} -\nabla_y f(x,y^*) \\ 0 \end{bmatrix}$.  
Triangle inequality reduces the proof into bounding $\|v_\delta(x,\bar{y}^*_\delta, \bar{\lambda}^*_\delta) - \begin{bmatrix}\partialD{V(x,y^*)}{x}^\top & A^\top\end{bmatrix} v^*_\epsilon\| 
+ 
\|\begin{bmatrix}\partialD{V(x,y^*)}{x}^\top & A^\top\end{bmatrix} (v^*_\epsilon - v^*)\|$. Combining the following equations yields the desired result,
\begin{align*}
& \|v^*_\epsilon - v^*\| \leq \mathcal{O}(\epsilon) \, \text{(Lemma \ref{lemma:asym_grad_2_constr})} \\
& \|
\begin{bmatrix}
\frac{\bar{y}^*_\delta - y^*}{\delta} \\
\frac{\bar{\lambda}^*_\delta - \lambda^*}{\delta}
\end{bmatrix} 
- 
v^*_\epsilon
\| 
\leq 
\frac{\sqrt{N}}{\bar{\mu}^2}\mathcal{O}(\delta) \, \text{(Lemma \ref{lemma:asym_grad_1_constr})} \\
& v_\delta(x,\bar{y}^*_\delta, \bar{\lambda}^*_\delta) 
=
\partialD{V(x,y^*)}{x}^\top \frac{\bar{y}^*_{\delta} - y^*}{\delta} + N\frac{\text{H.O.T.}}{\delta} 
+ A^\top \frac{\bar{\lambda}^*_\delta - \lambda^*}{\delta} \, \text{(Taylor's expansion)}
\end{align*}
where $\|\text{H.O.T.}\| = \mathcal{O}(\|\bar{y}^*_{\delta} - y^*\|^2) \leq \mathcal{O}(\frac{\delta^2}{\bar{\mu}^2})$.
\end{proof}

\begin{theorem}\label{theorem:asym_convergence_constr}
Consider problem \ref{eq:equal_p1}. Under assumptions \ref{assum:VI}, \ref{assum:optim} and \ref{assum:LICQ}, inexact gradient descent (algorithm \ref{alg:ffo_algo}) outputs an $\mathcal{O}(\epsilon)$-stationary point of the hyperobjective $F$ with $\tilde{\bigO}(N^2\epsilon^{-2})$ first-order oracle calls to $f$ and $g_i$ in both equality constrained and unconstrained settings. 
\end{theorem}
\begin{proof}
Following the argument of Theorem~\ref{theorem:smooth_F}, the hyperobjective $F$ is smooth. Choose $\delta \leq \frac{\epsilon}{2N}$ and use the potential function $\phi_{\frac{\epsilon}{2}}$. By Lemma \ref{lemma:inexact_lower_level}, if the approximate lower-level solutions satisfy $\|y^*-\hat y\|\leq \delta^2$ and $\|\bar y_\delta^*-\hat y_\delta\|\leq \delta^2$, then the inexact hypergradient ${\hat\nabla F}(x)=\nabla_x f(x,\hat y)+\hat v_\delta(x,\hat y_\delta)$ satisfies $\|\hat\nabla F(x)-\nabla F(x)\|\leq \bigO(N\delta+\epsilon/2)\leq \bigO(\epsilon)$. Therefore, by Lemma~\ref{lemma:inexact_GD}, inexact gradient descent requires $\bigO(\epsilon^{-2})$ hypergradient evaluations to reach an $\bigO(\epsilon)$-stationary point of $F$. It remains to bound the cost of each hypergradient evaluation. Since the potential is constructed using the approximate equilibrium $\hat y$, its gradient operator is also inexact, with an error of $\bigO(\epsilon)$ (Lemma~\ref{lemma:inexact_potential}). Moreover, the auxiliary game is only locally strongly monotone on $\bar B_R(y^*)$. We therefore choose $\delta \leq \min\{\epsilon/(2N),\bar\mu R/(4L_f),\sqrt{R/2}\}$, which ensures $\|\bar y_\delta^*-\hat y\|\leq \|\bar y_\delta^*-y^*\|+\|y^*-\hat y\|\leq R$. Hence, gradient descent initialized at $\hat y$ remains in the local neighborhood and converges linearly to the auxiliary equilibrium. By Lemma~\ref{lemma:proj_GD_VI_convergence}, $\bigO(\log(1/\delta^2))$ lower-level iterations are sufficient to obtain the required $\delta^2$-accurate solutions. Each evaluation of the auxiliary game operator requires gradients of all $N$ surrogate costs, and evaluating each surrogate cost requires $\bigO(N)$ calls to the original follower objectives $g_i$. Thus, each lower-level iteration uses $\bigO(N^2)$ first-order oracle calls. Combining this with $\bigO(\epsilon^{-2})$ outer iterations and $\bigO(\log(1/\delta^2))$ lower-level iterations gives a total oracle complexity of $\bigO(N^2\epsilon^{-2}\log(1/\delta^2))=\widetilde{\bigO}(N^2\epsilon^{-2})$.
\end{proof}

\section{Extension to Stackelberg Game with Linear Inequality Constraints}\label{sec:ineq}
Consider the following Stackelberg game with linear inequality constraints,
\begin{align*}\label{eq:inequal_p1}
    \Min_{x\in \inputspace} \,& F(x) \define f(x, y^*(x)) \\
    \text{s.t. } \,
    &g_i(x,y^*(x)) \le g_i(x,y_i, y^*_{\neg i}(x)) \, , \, \forall i \in [N] \, , \, \forall y_i \in Y_i(x) \\
    & Y_i(x) = \bigl\{y_i \in \real^{n_i} : \ei(x, y_i) \leq 0\bigr\} \\
    & \ei(x, y_i) = A^{(i)}x + B^{(i)}y_i + b^{(i)} \tag{Q1}
\end{align*}
\noindent
In problem \ref{eq:inequal_p1}, the hyper-objective $F$ is Lipschitz continuous but nonsmooth because the active set may change discontinuously. $({\delta}, \epsilon)$-stationarity is used to measure convergence because obtaining an $\epsilon$-stationary point of a nonsmooth nonconvex function is intractable (\cite{zhang2020complexity}).

\begin{definition}[$({\delta}, \epsilon)$-stationarity]
Let $F:\mathbb{R}^d \rightarrow \mathbb{R}$ be a locally Lipschitz continuous function. Let $x\in\mathbb{R}^d$. Let $B_\delta(x)$ be the Euclidean ball of radius $\delta>0$ centered at $x$. The generalized gradient of $F$ at $x$ is the set $\partial F(x) = \text{conv}(\{g|g = \lim_{k \rightarrow \infty} \nabla f(x_k), \, x_k\rightarrow x\})$ (\cite{clarke1990optimization}). Also, $\partial F(x) = \cap_{\delta>0} \cup_{y \in B_{\delta}(x)} \partial F(y)$. The Goldstein subdifferential of $F$ at $x$ is the set $\partial_\delta F(x) = \text{conv}(\cup_{y\in B_\delta(x)} \partial F(y))$ (\cite{goldstein1977optimization}). $x$ is called a $({\delta}, \epsilon)$-stationary point of a locally Lipschitz continuous function $F$ if $\text{dist}(0, \partial_\delta F(x))\leq \epsilon$, where $\text{dist}(x,S)\define \inf_{y \in S}\|x-y\|$.
\end{definition}

\begin{assumption}[Active set identification]\label{assum:active_set}
For problem \ref{eq:inequal_p1}, assume that the active constraints can be correctly identified.
\end{assumption}
By identifying the active set at the Nash equilibrium (similar to \cite{zhao2025fully}), we can transform an arbitrary smooth and strongly monotone game $\mathcal{G} \define (x, [N], \{Y_i\}_{i\in[N]},  \{g_i\}_{i\in[N]})$ with linear inequality constraints into an equivalent game with linear equality constraints. Then, we can directly apply the hypergradient approximation derived for Stackelberg game with linear equality constraints in Section \ref{sec:sym} or \ref{sec:asym_constr}. 

\noindent
Given an arbitrary reference point $\bar{x} \in \inputspace$. For each player $i\in [N]$, we define the strict active set $\I_i \define \{j \in [m_i]:(\lambda^*_i)_j > 0\}$, active dual variable $\lambda^*_{i,\I_i} = (\lambda^*_{i,j})_{j \in \I_i}$ and the active equality constraint function $\ei_{\I_i} = (\ei_j)_{j\in\I_i}$. We define the ghost cost function 
$$\ghost{g}_i(x,y) = g_i(x,y) + (\lambda^*_i)^\top \ei(x, y_i)$$
and the ghost equality constraint $\ghost{Y}_i = \{y_i \in \real^{n_i}: \ghost{\ei}(x,y_i) = 0\}$, where 
$$\ghost{\ei}(x,y_i) = \partialD{\ei_{\I_i}(\bar{x},y^*_i)}{x}(x-\bar{x}) + \partialD{\ei_{\I_i}(\bar{x},y^*_i)}{y}(y_i - y^*_i)$$
are the active constraints. We replace the original lower-level game $\mathcal{G} = (x, [N], \{Y_i\}_{i\in[N]},  \{g_i\}_{i\in[N]})$ with ghost lower-level game $\ghost{\mathcal{G}} \define (x, [N], \{\ghost{Y}_i\}_{i\in[N]},  \{\ghost{g}_i\}_{i\in[N]})$. Denote $\ghost{y}^*$ as the Nash equilibrium of $\ghost{\mathcal{G}}$ and $\ghost{\lambda}^*$ as the optimal dual solution for $\ghost{\mathcal{G}}$. The ghost problem is defined as 
\begin{align*}\label{eq:inequal_p2}
    \Min_{x\in \inputspace} \,& \ghost{F}(x) \define f(x, \ghost{y}^*(x)) \\
    \text{s.t. } \,
    &\ghost{g}_i(x,\ghost{y}^*(x)) \le \ghost{g}_i(x,y_i, \ghost{y}^*_{\neg i}(x)) \, , \, \forall i \in [N] \, , \, \forall y_i \in \ghost{Y}_i(x) \\
    & \ghost{Y}_i(x) = \bigl\{y_i \in \real^{n_i} : \partialD{\ei_{\I_i}(\bar{x},y^*_i)}{x}(x-\bar{x}) + \partialD{\ei_{\I_i}(\bar{x},y^*_i)}{y}(y_i - y^*_i) =  0\bigr\} \tag{Q2}
\end{align*}
By comparing KKT conditions of the original and ghost lower-level problem (Lemma \ref{lemma:ghost_equivalence}), the ghost problem \ref{eq:inequal_p2} is in fact equivalent to the original Stackelberg game with inequality constraints \ref{eq:inequal_p1}, such that $\ghost{y}^*(\bar{x}) = y^*(\bar{x})$ and $\totalD{\ghost{y}^*(\bar{x})}{x} = \totalD{{y}^*(\bar{x})}{x}$. Hence, $\ghost{F}(\bar{x})=F (\bar{x})$ and $\nabla \tilde{F}(\bar{x}) = \nabla F(\bar{x})$.

\begin{lemma}\label{lemma:ghost_game}
The ghost lower-level game $\ghost{\mathcal{G}} \define (x, [N], \{\ghost{Y}_i\}_{i\in[N]},  \{\ghost{g}_i\}_{i\in[N]})$ in problem \ref{eq:inequal_p2} is $\mu$-strongly monotone, $C$-smooth and $S$-Hessian smooth.
\end{lemma}
\begin{proof}
Define $\ghost{V}(x,y) \define (\nabla_{y_1} \ghost{g}_1(x,y), \cdots, \nabla_{y_N} \ghost{g}_N(x,y)) \in \real^n$. Let $x, x^\prime \in \inputspace$ and $y, y^\prime \in \ghost{Y}$,
\begin{align*}
    \ghost{V}(x^\prime,y^\prime) - \ghost{V}(x,y)
    & = {V}(x^\prime,y^\prime) - {V}(x,y) + B^\top \lambda^* - B^\top \lambda^* \\
    & = {V}(x^\prime,y^\prime) - {V}(x,y)
\end{align*}
We can conclude the proof using the definitions of strong monotonicity, smoothness and Hessian smoothness in \ref{assum:VI}.
\end{proof}

\begin{lemma}[Equivalence of the ghost problem and the original problem]\label{lemma:ghost_equivalence}
Let $\ghost{y}^*(x) \in NE(\ghost{\mathcal{G}})$ be the Nash equilibrium of the ghost lower-level game \ref{eq:inequal_p2} and let $y^*(x) \in NE(\mathcal{G})$ be the Nash equilibrium of the original lower-level game \ref{eq:inequal_p1}. For all $\bar{x}\in \inputspace$ where active set is locally constant, $\ghost{y}^*(\bar{x}) = y^*(\bar{x})$, $\totalD{\ghost{y}^*(\bar{x})}{x} = \totalD{{y}^*(\bar{x})}{x}$, $\ghost{\lambda}^*(\bar{x}) = 0$ where $\ghost{\lambda}^*$ is the optimal dual variable of $\ghost{\mathcal{G}}$
\end{lemma}
\begin{proof}
Denote $B_\I \define \diag((\partialD{\ei_{\I_i}(\bar{x},y^*_i)}{y})_{i\in[N]}) $ as a block diagonal matrix corresponding to Jacobians of active inequality constraints, and 
$A_\I \define \begin{bmatrix} \partialD{e^{(1)}_{\I_1}(\bar{x},y^*_1)}{x} \\ \vdots \\ \partialD{e^{(N)}_{\I_N}(\bar{x},y^*_N)}{x}\end{bmatrix}$.
The KKT condition of the original problem \ref{eq:inequal_p1} is
\begin{align*}
    \Gamma(\bar{x}, y^*, \lambda^*) 
    = 
    \begin{bmatrix} V(\bar{x},y^*) + B_\I^\top \lambda^*_\I \\ \begin{bmatrix} \diag(\lambda^*_{1, \I_1}) e^{(1)}_{\I_1}(\bar{x},y^*_1) \\ \vdots \\ \diag(\lambda^*_{N, \I_N}) e^{(N)}_{\I_N}(\bar{x},y^*_N) \end{bmatrix} \end{bmatrix} 
    = 
    0
\end{align*}
The jacobian
\begin{align*}
    \partialD{\Gamma(\bar{x}, y^*, \lambda^*_\I)}{(y^*, \lambda^*_\I)} 
    = 
    \begin{bmatrix}\partialD{V(\bar{x},y^*)}{y} & B_\I^\top \\ \diag(\lambda^*_\I) B_\I & 0\end{bmatrix}
\end{align*}
is invertible following the same steps in Lemma \ref{lemma:invH}. Then by implicit function theorem, $y^*(\bar{x})$, $\lambda^*_\I(\bar{x})$ are continuously differentiable functions of x and their derivatives are given by
\begin{align}\label{eq:orig_ineq_KKT}
    \begin{bmatrix}\partialD{V(\bar{x},y^*)}{y} & B_\I^\top \\ \diag(\lambda^*_\I) B_\I & 0\end{bmatrix}
    \begin{bmatrix} \totalD{y^*(\bar{x})}{x} \\ \totalD{\lambda^*_\I(\bar{x})}{x} \end{bmatrix}
    & = 
    \begin{bmatrix}-\partialD{V(\bar{x},y^*)}{x} \\ -\diag(\lambda^*_\I)A_\I\end{bmatrix} 
    \\ \nonumber
    \begin{bmatrix}\partialD{V(\bar{x},y^*)}{y} & B_\I^\top \\ B_\I & 0\end{bmatrix}
    \begin{bmatrix} \totalD{y^*(\bar{x})}{x} \\ \totalD{\lambda^*_\I(\bar{x})}{x} \end{bmatrix}
    & = 
    \begin{bmatrix}-\partialD{V(\bar{x},y^*)}{x} \\ -A_\I\end{bmatrix} \\
\end{align}

\noindent
On the other hand, the KKT condition of the ghost problem \ref{eq:inequal_p2} is
\begin{align*}
    \ghost{\Gamma}(\bar{x}, \ghost{y}^*, \ghost{\lambda}^*) 
    = 
    \begin{bmatrix} V(\bar{x},\ghost{y}^*) + B_\I^\top \lambda^*_\I + B_\I^\top \ghost{\lambda}^* \\ \begin{bmatrix} \ghost{e}^{(1)}_{\I_1}(\bar{x},\ghost{y}^*_1) \\ \vdots \\ \ghost{e}^{(N)}_{\I_N}(\bar{x},\ghost{y}^*_N) \end{bmatrix} \end{bmatrix} 
    = 
    0
\end{align*}
Substituting $\ghost{y}^* = y^*$ and $\ghost{\lambda}^* = 0$, the KKT condition is satisfied. By LICQ and strong monotonicity (assumptions \ref{assum:LICQ}, \ref{assum:VI}), the unique Nash equilibrium and optimal dual variable of the ghost lower-level game are $\ghost{y}^* = y^*$ and $\ghost{\lambda}^* = 0$ respectively.
Using implicit function theorem, the derivatives of $\ghost{y}^*(\bar{x})$ and $\ghost{\lambda}^*(\bar{x})$ are given by
\begin{align}\label{eq:ghost_KKT}
    \begin{bmatrix}\partialD{V(\bar{x},\ghost{y}^*)}{y} & B_\I^\top \\ B_\I & 0\end{bmatrix}
    \begin{bmatrix} \totalD{\ghost{y}^*(\bar{x})}{x} \\ \totalD{\ghost{\lambda}^*_\I(\bar{x})}{x} \end{bmatrix}
    & = 
    \begin{bmatrix}-\partialD{V(\bar{x},\ghost{y}^*)}{x} \\ -A_\I\end{bmatrix} 
\end{align}
The derivatives of the KKT conditions of \ref{eq:inequal_p1} and \ref{eq:inequal_p2} are the same. Hence, $\totalD{\ghost{y}^*(\bar{x})}{x} = \totalD{{y}^*(\bar{x})}{x}$.
\end{proof}

Assuming exact active set identification (assumption \ref{assum:active_set}), we can directly apply the finite difference hypergradient approximation in problem \ref{eq:equal_p1} and the gradient approximation guarantees follow. 

\noindent
However, the hyper-objective $F$ is nonsmooth since the active set can change with $x$ discontinuously. We can utilize Lemma \ref{lemma:lipschitz_opt} (see \cite{kornowski2024first}) for minimizing Lipschitz nonsmooth nonconvex function $F$, and then we can conclude the convergence rate for problem \ref{eq:inequal_p1} with theorem \ref{thm:inequality}. Note that $\mathcal{O}(\hat{\delta}^{-1}\epsilon^{-3})$ for reaching a $(\hat{\delta}, \epsilon)$-stationary point is the optimal rate for nonsmooth nonconvex optimization \cite{cutkosky2023optimal}.
\begin{lemma}\label{lemma:lipschitz_opt}
Suppose $F: \real^d \rightarrow \real$ is $L$-Lipschitz and that $\|\Approx{\nabla}F(\cdot) - \nabla F(\cdot)\| \leq \alpha$. Then there exists an algorithm which outputs a $(\hat{\delta}, \epsilon + \alpha)$-stationary point with $T = \bigO(\hat{\delta}^{-1}\epsilon^{-3})$ calls to $\Approx{\nabla}F(\cdot)$ for arbitrary $\hat{\delta} > 0$.
\end{lemma}
\begin{theorem}\label{thm:inequality}
Consider problem \ref{eq:inequal_p1}. Under assumptions \ref{assum:VI}, \ref{assum:optim}, \ref{assum:LICQ} and \ref{assum:active_set}, there exists an algorithm that outputs an $(\hat{\delta}, 2\epsilon)$-stationary point of the hyper-objective $F$ with $\tilde{\bigO}(\hat{\delta}^{-1}\epsilon^{-3})$ oracle calls to $f$ and $g_i$ for $\hat{\delta} > 0$. 
\end{theorem}
\begin{proof}
The proof is identical to that of theorem \ref{theorem:sym_equality_convergence} and \ref{theorem:asym_convergence_constr}. Except that the complexity of the outer loop is $\mathcal{O}(\hat{\delta}^{-1}\epsilon^{-3})$ instead of $\mathcal{O}(\epsilon^{-2})$
\end{proof}

\section{Discussion: Challenges of Fully First-Order Hypergradient Approximation in Asymmetric Game}\label{sec:asym_toy_example}
Let $x\in\mathbb{R}$. Suppose the upper-level objective function is $f(x,y) = u^\top y$. Consider an unconstrained game $(x, [N], \{\mathbb{R}^{n_i}\}_{i\in[N]}, \{g_i\}_{i\in [N]})$, where the the cost function for each player $i \in [N]$ is
\begin{align*}
    g_i(x,y) = \frac{1}{2}y_i^\top Q_{ii} y_i + y_i^\top \sum_{j\neq i}Q_{ij}y_j + x c_i^\top y_i
\end{align*}
At equilibrium $y^*$, $\nabla_{y_i}g_i(x, y^*_i, y^*_{\neg i}) = 0$, $\forall i\in[N]$. Then,
\begin{align*}
    V(x,y^*) 
    = 
    \begin{bmatrix} 
    \nabla_{y_1} g_1(x,y^*) \\
    \vdots \\
    \nabla_{y_N} g_N(x,y^*)
    \end{bmatrix}
    = 
    Qy + xc
\end{align*}
Solving $V(x,y^*) = 0$, we have
\begin{align*}
    y^* = - Q^{-1}(xc)
\end{align*}
Using implict differentiation, the exact hypergradient term to be approximated is 
\begin{align}\label{eq:quad_game_true_grad}
    (\totalD{y^*(x)}{x})^\top \nabla_yf(x,y^*(x))
    &=
    (-Q^{-1} c)^\top u \nonumber \\ 
    & =
    - c^\top Q^{-\top} u 
\end{align}

\noindent
We may add perturbation $\delta f(x,y)$ to each player's cost function. The equilibrium $y^*_\delta$ of the perturbed game is given by
\begin{align*}
    y^*_\delta = - Q^{-1}(xc + \delta u) = y^* - \delta Q^{-1}u
\end{align*}
The first-order hypergradient approximation in Section \ref{sec:sym} only matches equation \ref{eq:quad_game_true_grad} when $Q$ is symmetric
\begin{align*}
    v_\delta(x, {y}^*_\delta) 
    &= 
    \frac{1}{\delta}\sum_{i=1}^N \nabla_x g_i(x,y^*_{\delta,i}, y^*_{\neg i}) - \nabla_xg_i(x,y^*) \\
    & = 
    \frac{1}{\delta}\sum_{i=1}^N c_i^\top (y^*_{\delta,i} - y^*_i) \\
    & =
    \frac{1}{\delta}\sum_{i=1}^N c_i^\top [\delta Q^{-1}u]_i \\
    & = 
    -c^\top Q^{-1} u
\end{align*}

\noindent
Alternatively, using results in Section \ref{sec:asym}, we can construct a perturbed game with cost function $\bar{g}_i(x,y) = \phi_\epsilon(x,y) - g_i(x,y) + \delta f(x,y)$, where $\phi_\epsilon(x,y) = \frac{1}{\epsilon} \sum_{j=1}^N [g_j(x, y_j + \epsilon (y_j - y^*_j), y_{\neg j}) - g_j(x,y)]$ is a potential function. Its equilibrium $\bar{y}^*_\delta$ satisfies $\lim_{\delta \rightarrow 0}\bar{y}^*_\delta = y^*$ and $\frac{d \bar{y}^*_\delta}{d \delta}|_{\delta=0} = (Q+\epsilon D)^{-\top} u$, where $D = \text{diag}(Q_{11}, \cdots, Q_{NN})$ is a block-diagonal matrix. Then the first-order hypergradient approximation $v_\delta(x, \bar{y}^*_\delta)$ exactly matches the exact hypergradient in this toy example. The intuition of the construction is that the solution of $\bar{V}(x,y) = (\nabla_{y_1}\bar{g}_1(x,y), \cdots, \nabla_{y_N}\bar{g}_N(x,y)) = 0$ approximates the solution of $Q^\top(y - y^*) + \delta u = 0$ as $\epsilon \rightarrow 0$. Solving $\nabla_{y_i} \phi_\epsilon(x, \bar{y}^*_\delta) - \nabla_{y_i}g_i(x,\bar{y}^*_\delta) + \delta \nabla_{y_i}f(x,\bar{y}^*_\delta) = 0 , \, \forall i \in [N]$: 
\begin{align*}
& \nabla_{y_i} \phi_\epsilon(x, y) - \nabla_{y_i}g_i(x,y) + \delta \nabla_{y_i}f(x,y) \\
& = 
\frac{1}{\epsilon} \sum_{j=1}^N [\nabla_{y_i} g_j(x, y_j + \epsilon (y_j - y^*_j), y_{\neg j}) - \nabla_{y_i} g_j(x,y)] + \nabla_{y_i} g_i(x, y_i + \epsilon (y_i - y^*_i), y_{\neg i}) - \nabla_{y_i}g_i(x,y) + \delta \nabla_{y_i}f(x,y) \\
&=
\frac{1}{\epsilon}[Q_{ii}(y_i + \epsilon(y_i - y^*_i)) + \sum_{j\neq i}Q_{ij} y_j - Q_{ii} y_i - \sum_{j\neq i}Q_{ij} y_j] \\
& + \frac{1}{\epsilon}[\sum_{j\neq i} Q^{\top}_{ji} (y_j+\epsilon(y_j - y^*_j)) - \sum_{j\neq i} Q^{\top}_{ji} y_j] \\
& + [Q_{ii} (y_i + \epsilon(y_i - y^*_i)) + \sum_{j \neq i} Q_{ij} y_j] \\
& - [Q_{ii} y_i + \sum_{j \neq i} Q_{ij} y_j] + \delta u_i \\
&=
Q_{ii}(y_i - y_i^*) + \sum_{j\neq i} Q^{\top}_{ji} (y_j - y^*_j) + Q_{ii} \epsilon (y_i - y^*_i)
= 0
\end{align*}
Stacking these equations together, we have
\begin{align*}
Q^\top (y - y^*) + \epsilon D(y-y^*) + \delta u &= 0 \\
y &= y^* + (Q^\top + \epsilon D)^{-1}(-\delta u)
\end{align*}
We have achieved a desired formula for $\bar{y}^*_\delta$ in this toy example.

\section{Failure of Revelation Principle}
To understand why non-truthful auction is important, we construct single-bidder and two-bidder settings where there exist non-truthful auctions with more revenue than any truthful auctions. Inspired by Example B.1. in \cite{berzack2026optimal}, the key idea is to introduce valuation-dependent constraints over bids. 

\subsection{Single-Item Single-Bidder}
Let the valuation be $t \sim \text{Uniform}[0,1]$. Fix $ \alpha > 1$. $\{a \in \mathbb{R}_{+}: \, a \leq k(t)\}$ is the feasible set of the bid given valuation $t$, where
\begin{align*}
    k(t) = 
    \begin{cases}
          \alpha        & \text{if } t < \frac{2}{3}, \\
          1         & \text{if } t \geq \frac{2}{3}.
    \end{cases}
\end{align*}
Consider the following allocation $\rho: [0,1] \rightarrow \{0,1\}$ and payment rule $P: [0,1] \rightarrow \mathbb{R}$,
\begin{align*}
    \rho(a) = 
    \begin{cases}
          1        & \text{if } a > 0, \\
          0         & \text{if } a = 0.
    \end{cases}
\end{align*}
\begin{align*}
    P(a) = 
    \begin{cases}
          0        & \text{if } a=0 \\
          \frac{2}{3}        & \text{if } 0 < a \leq 1, \\
          \frac{1}{3}         & \text{if } a > 1.
    \end{cases}
\end{align*}
The utility of the bidder is $u(a,t) = t \rho(a) - P(a)$. We can now derive the best response of the bidder that maximizes the utility. When $t \leq \frac{1}{3}$, bidding anything above zero would results in negative utility. The best response is to bid zero. If $\frac{1}{3} < t < \frac{2}{3}$, the constraint is $ a \leq \alpha$. The best response is to bid $a > 1$ and pay $\frac{1}{3}$. If $t \geq \frac{2}{3}$, the constraint is $a \leq 1$. The best response is to bid any $0 < a \leq 1$ and pay $\frac{2}{3}$. Hence, the expected revenue of this auction is
\begin{align*}
    0 + \mathbb{E}_{t \in [\frac{1}{3}, \frac{2}{3}]}\left[\frac{1}{3}\right] + \mathbb{E}_{t \geq \frac{2}{3}}\left[\frac{2}{3}\right] 
    &= \frac{1}{3} \cdot \frac{1}{3} + \frac{1}{3} \cdot \frac{2}{3} \\
    &= \frac{1}{3}
\end{align*}
On the other hand, the optimal truthful auction is Myerson's optimal auction \cite{myerson1981optimal} with an expected revenue of $\frac{1}{4}$. Hence, the non-truthful auction we propose in fact has a higher revenue than any truthful auction in this setting.

\subsection{Single-Item Two-Bidder}
Let the valuation of bidder $i$ be $t_i \sim \text{Uniform}[0,1]$ sampled independently. Consider a singleton $B_1 = \{l_1=0\}$ and contiguous intervals $B_j = (l_j, m_j]$ for $j=2,3,4$. $m_j = l_{j+1}$ for $j=2,3$ and $m_j > l_j > l_1$ for $j=2,3,4$. $\cup_{j=1}^4 B_j$ is the set of possible bids. Let $0< c_1 < c_2, < c_3 < 1$. $\{a_i \in \mathbb{R}_{+}: \, a_i \leq k(t_i)\}$ is the feasible set of the $i$th bidder's bid, where
\begin{align*}
    k(t) = 
    \begin{cases}
          m_4        & \text{if } t < c_2, \\
          m_3         & \text{if } c_2 \leq t < c_3, \\
          m_2          & \text{if } c_3 \leq t
    \end{cases}
\end{align*}
Consider the following allocation $\rho: [0,1] \times [0,1] \rightarrow \Delta^2$ and payment rule $P: [0,1] \times [0,1] \rightarrow \mathbb{R}^2$,
\begin{align*}
    \rho_i(a_1, a_2) = 
    \begin{cases}
          1        & \text{if } a_i \geq a_{\neg i}, \\
          0         & \text{if } a_i < a_{\neg i} \text{ or } a_i=0.
    \end{cases}
\end{align*}
\begin{align*}
    P_i(a_1, a_2) = 
    \begin{cases}
          0        & \text{if } a_i \in B_1 \\
          c_3 \rho_i(a)       & \text{if } a_i \in B_2, \\
          c_2 \rho_i(a)         & \text{if } a_i \in B_3, \\
          c_1 \rho_i(a)      & \text{if } a_i \in B_4. 
    \end{cases}
\end{align*}
The utility of the bidder $i$ given valuation $t_i$ is $u_i(a, t_i) = t_i \rho_i(a) - P_i(a)$. Now we derive the Nash equilibrium. Consider player $i$. When $t_i < c_1$, The dominant strategy is to bid zero, because the utility will be negative even if bidder $i$ wins. When $t_i \in [c_1, c_2)$, the constraint is $a_i \leq m_4$, the dominant strategy is bidding $a_i \in B_4$ because it results in the least payment and maximizes the chance of winning. Following this logic, a Nash equilibrium bidding strategy is
\begin{align*}
    \beta(t_i) = 
    \begin{cases}
          0        & \text{if }  t < c_1\\
          \text{any } a_i \in B_4        & \text{if } t \in [c_1, c_2), \\
          \text{any } a_i \in B_3         & \text{if } t \in [c_2, c_3), \\
          \text{any } a_i \in B_2         & \text{if } t \in [c_3, 1). 
    \end{cases}
\end{align*}
Denote $\mathbb{P}_1 = \Pr(t_i < c_1)$, $\mathbb{P}_2 = \Pr(t_i \in [c_1, c_2))$, $\mathbb{P}_3 = \Pr(t_i \in [c_2, c_3))$, $\mathbb{P}_4 = \Pr(t_i \in [c_3, c_4))$. The expected revenue can be computed by the following Python program:
\begin{lstlisting}
    def two_bidder_four_interval(C):
        # C is a list containing c1, c2, c3
        C = numpy.array(C)
        P = (C[1:]-C[:-1])
        P = numpy.concatenate([ [C[0]], P, [(1-C[-1])] ])
    
        C_dim = len(C)
        P_dim = len(P)
    
        c1, c2, c3 = C[0], C[1], C[2]
        payment_matrix = np.array([
            [0, c1, c2, c3],
            [c1, c1, c1, c1],
            [c2, c1, c2, c2],
            [c3, c1, c2, c3]
        ])
    
        revenue = 0
        for i in range(P_dim):
            for j in range(P_dim):
                revenue += payment_matrix[i, j] * P[i] * P[j]
        
        return revenue
\end{lstlisting}

Substitute $c_1=0.4$, $c_2=0.6$, $c_3=0.8$, $B_2=(0,1]$, $B_3=(1,1.5]$, $B_4=(1.5,2]$. The expected revenue is $0.472$. Myerson's optimal auction is the optimal truthful auction and has an expected revenue $\frac{5}{12}=0.417$. The proposed non-truthful auction has a revenue higher than any truthful auction.
\section{Experiments}\label{appx:exp}
\warmStart*
\begin{proof}
 Let $y^*(x_k) \in \text{LNE}(\mathcal{G}(x_k))$. Since the KKT matrix $K(x_k, y^*(x_k))$ is invertible, by implicit function theorem, there exists a neighborhood $N_{x_k} \times N_{y^*(x_k)}$ of $(x_k,y^*(x_k))$ and a continuously differentiable and $L_k$-Lipschitz continuous function $z^*_k: N_{x_k} \rightarrow N_{y^*(x_k)}$ such that $z^*_k(x) \in \text{LNE}(\mathcal{G}(x))$,  $\forall x \in N_{x_k}$, and $z^*_k(x_k) = y^*(x_k)$. Suppose $\|x_k - x_{k+1}\| \leq \mathcal{O}(\eta)$ is sufficiently small such that $x_{k+1} \in N_{x_k}$ and $\|x_k - x_{k+1}\| \leq \frac{r}{2 L_k}$, then $\|z^*_k(x_k) - z^*_k(x_{k+1})\| \leq \frac{r}{2}$. Suppose $\|y_k - y^*(x_k)\| \leq \frac{r}{2}$, $\|y_k - z^*_k(x_{k+1})\| \leq r$ by triangle inequality. For arbitrary precision $\hat{\delta}$, the solver outputs $y_{k+1} = S(x_{k+1}, y_k, \hat{\delta})$ such that $\|y_{k+1} - z^*_k(x_{k+1})\| \leq \hat{\delta}$. By induction, the claim is proved.
\end{proof}
\subsection{Experimental Details}
Each bidding strategy $\beta_{y_i}$ is two-layer fully connected neural network of width $10$ with SeLU activation (\cite{klambauer2017self}) and ReLU activation (\cite{agarap2018deep}) for the output layer. The allocation rule is a two-layer fully connected neural network of width $10$ with SeLU activation and Layer Normalization (\cite{ba2016layer}). We use Adam (\cite{kingma2014adam}) to optimize the parameters of the auction mechanism. We choose the hyperobjective step size $\eta=10^{-5}$ for single-item auction and $\eta=10^{-4}$ for multi-item auction. For single-item auction, we run gradient descent with step size $10^{-4}$ for $30000$ iterations to solve the original game, and with step size $10^{-4}$ for $1000$ iterations to solve the perturbed game. For multi-item auction, we run gradient descent with step size $10^{-4}$ for $40000$ iterations to solve the original game, and with step size $10^{-4}$ for $2000$ iterations to solve the perturbed game. 

We initialize the allocation rule of the mechanism to mimick first-price auction so that the initial lower-level game is easy to solve. Let $p \in \Delta^{N+1}$ be the output of the allocation rule of first-price auction. Let $q \in \Delta^{N+1}$ be the output of the allocation rule of the neural-network-parameterized allocation rule. Then, we run Adam with step size $10^{-4}$ for $10^{6}$ iterations to minimize their KL divergence $D_{\text{KL}}(p||q)\propto \sum_i p_i \log q_i$.

\begin{figure}[h!]
    \centering
    \centering
    \begin{subfigure}{0.32\textwidth}
        \centering
        \includegraphics[width=\linewidth]{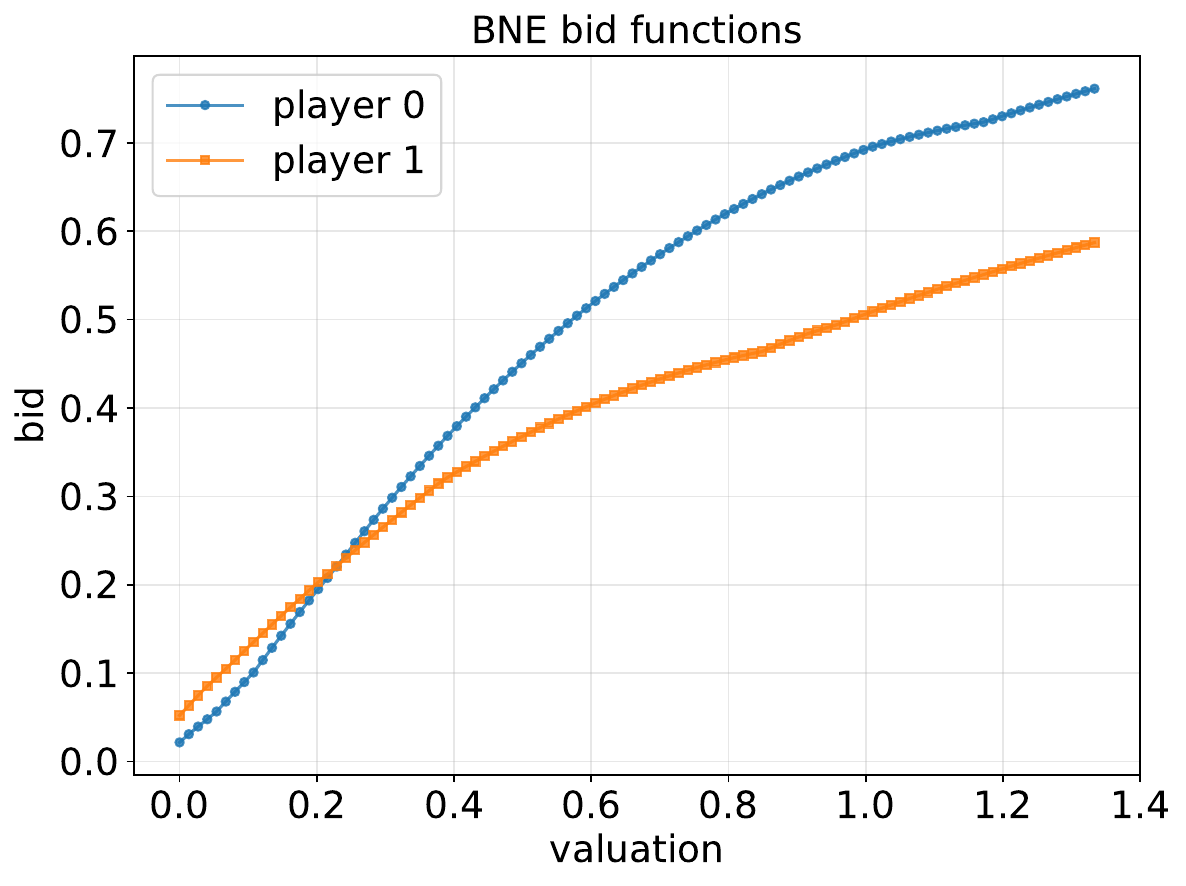}
        \caption{bidding strategy}
    \end{subfigure}
    \hfill
    \begin{subfigure}{0.32\textwidth}
        \centering
        \includegraphics[width=\linewidth]{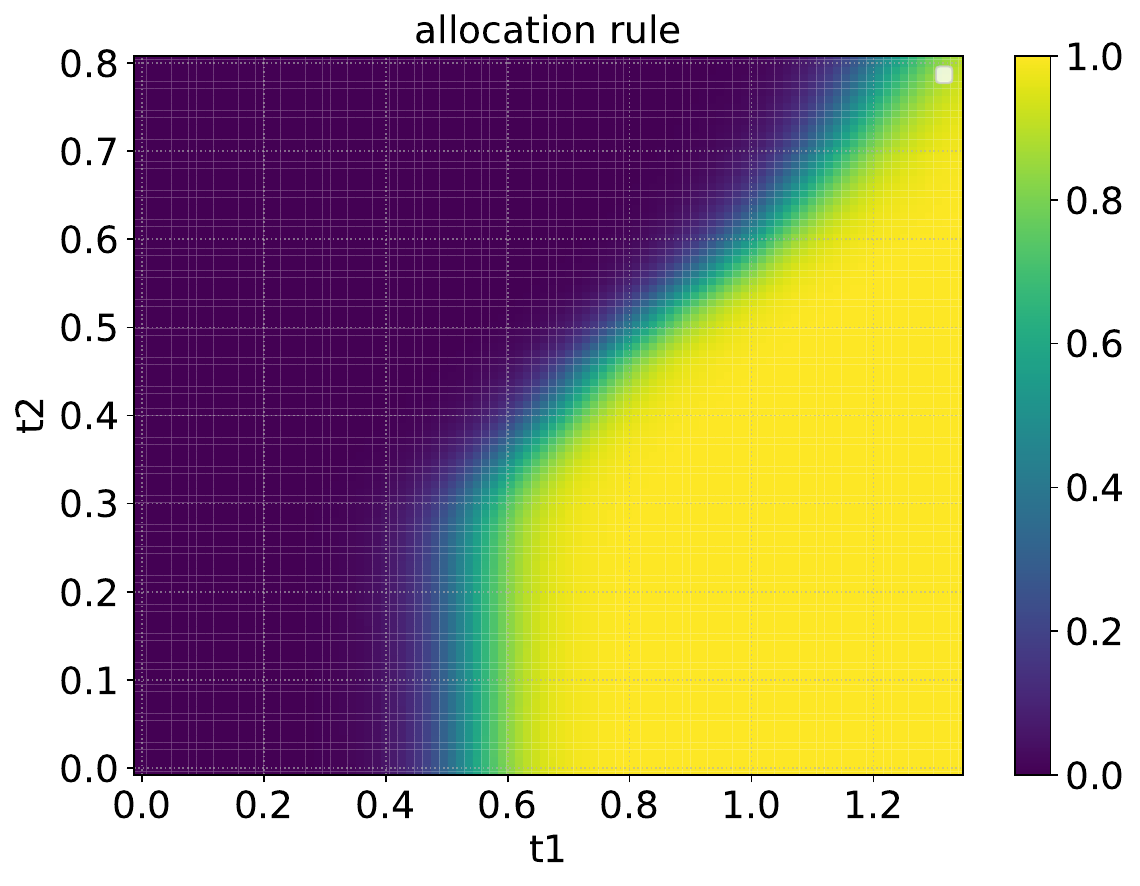}
        \caption{allocation for bidder 1}
    \end{subfigure}
    \hfill
    \begin{subfigure}{0.32\textwidth}
        \centering
        \includegraphics[width=\linewidth]{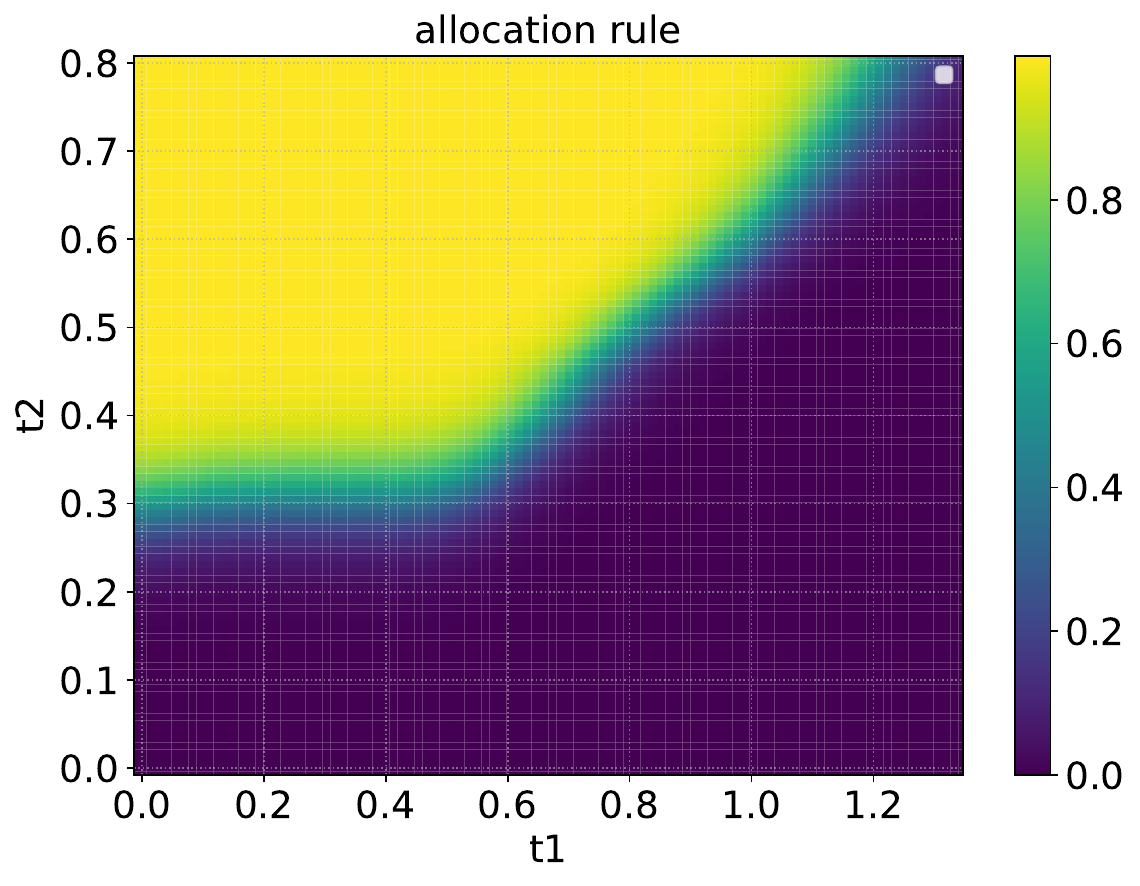}
        \caption{allocation for bidder 2}
    \end{subfigure}
    \caption{Equilibrium bidding strategy and allocation rule for asymmetric single-item auction. For allocation probability, x-axis and y-axis are the bid of bidder 1 and bidder 2 respectively.}
    \label{fig:asym_vis}
\end{figure}

\begin{figure}[h!]
    \centering
    \includegraphics[width=0.6\linewidth]{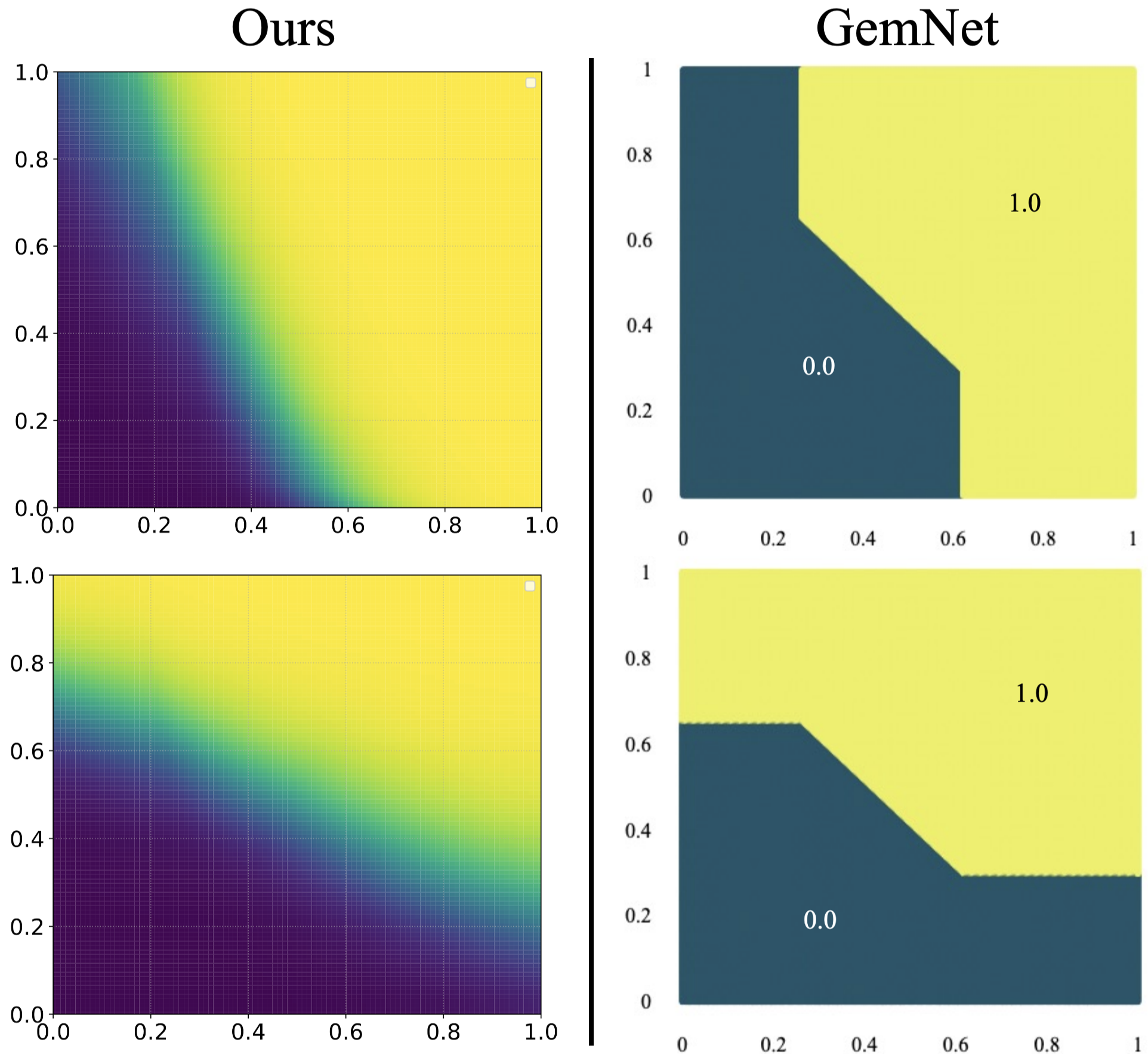}
    \caption{Comparison between non-truthful auction discovered by our method and truthful auction discovered by GemNet. For each plot, x-axis and y-axis is the bid of bidder 1 for item 1 and item 2 respectively. The top row and bottom row show the probability of item 1 and item 2 being allocated to bidder 1 respectively, given the bid of bidder 2 is (0,0.6).}
    \label{fig:multi_vis}
\end{figure}